\documentclass[journal]{IEEEtran}

\usepackage{epsfig,amsmath,amstext,psfrag,mathrsfs,amssymb,amsfonts,color,amsthm}
\usepackage{latexsym}
\usepackage{enumerate,cite}
\usepackage{graphicx}
\usepackage{dsfont}
\usepackage{indentfirst}
\usepackage{subfigure}

\newtheorem{definition}{\textbf{Definition}}

\newtheorem{lemma}{\textbf{Lemma}}
\newtheorem{theorem}{\textbf{Theorem}}
\newtheorem{proposition}{\textbf{Proposition}}
\newtheorem{remark}{\textbf{Remark}}

\newcommand{\nn}{\nonumber}

\newcommand{\mE}{\mathbb{E}}

\begin{document}

\title{Estimation of KL Divergence: \\Optimal Minimax Rate}
\author{Yuheng Bu, {\em Student Member, IEEE},
Shaofeng Zou, {\em Member, IEEE}, Yingbin Liang, {\em Senior Member, IEEE} \\
and Venugopal V. Veeravalli, {\em Fellow, IEEE}\\
\thanks{This work was presented in part at the International Symposium on Information Theory (ISIT), Barcelona, Spain, 2016 \cite{bu2016kl}.}
\thanks{The first two authors have contributed equally to this work.}
\thanks{The work of Y. Bu and V. V. Veeravalli was supported by the National Science Foundation under grants NSF 11-11342 and 1617789. The work of S. Zou and Y. Liang was supported by the Air Force Office of Scientific Research under grant FA 9550-16-1-0077. The work of Y. Liang was also supported in part by the National Science Foundation under Grant CCF-1801855 and by DARPA FunLoL program.}
\thanks{Y. Bu, S. Zou and V. V. Veeravalli are with the ECE Department and the Coordinated Science Laboratory,
University of Illinois at Urbana-Champaign, Urbana, IL 61801 USA (email: bu3/szou3/vvv@illinois.edu).}
\thanks{Y. Liang was with the Department of Electrical Engineering and Computer
Science, Syracuse University, Syracuse, NY 13244 USA. She is now with the
Department of Electrical and Computer Engineering, The Ohio State University, Columbus, OH 43210 USA (e-mail: liang.889@osu.edu).}
}

\maketitle

%


\begin{abstract}
The problem of estimating the Kullback-Leibler divergence $D(P\|Q)$ between two unknown distributions $P$ and $Q$ is studied, under the assumption that the alphabet size $k$ of the distributions can scale to infinity. The estimation is based on $m$ independent samples drawn from $P$ and $n$ independent samples drawn from $Q$. It is first shown that there does not exist any consistent estimator that guarantees asymptotically small worst-case quadratic risk over the set of all pairs of distributions. A restricted set that contains pairs of distributions, with density ratio bounded by a function $f(k)$ is further considered.
{An augmented plug-in estimator is proposed, and its worst-case quadratic risk is shown to be within a constant factor of $(\frac{k}{m}+\frac{kf(k)}{n})^2+\frac{\log ^2 f(k)}{m}+\frac{f(k)}{n}$, if $m$ and $n$ exceed a constant factor of $k$ and $kf(k)$, respectively.} Moreover, the minimax quadratic risk is characterized to be within a constant factor of $(\frac{k}{m\log  k}+\frac{kf(k)}{n\log k})^2+\frac{\log ^2 f(k)}{m}+\frac{f(k)}{n}$, if $m$ and $n$ exceed a constant factor of $k/\log(k)$ and $kf(k)/\log k$, respectively. The lower bound on the minimax quadratic risk is characterized by employing a generalized Le Cam's method.
A minimax optimal estimator is then constructed by employing both the polynomial approximation and the plug-in approaches.
\end{abstract}

\section{Introduction}

As an important quantity in information theory, the Kullback-Leibler (KL) divergence between two distributions has a wide range of applications in various domains. For example, KL divergence can be used as a similarity measure in nonparametric outlier detection \cite{bu2016}, multimedia classification \cite{moreno2003kullback,li2017inhomogeneous}, text classification \cite{dhillon2003divisive}, and the two-sample problem \cite{anderson1994two}. In these contexts, it is often desired to estimate KL divergence efficiently based on available data samples. This paper studies such a problem.

Consider the estimation of KL divergence between two probability distributions $P$ and $Q$ defined as
\begin{equation}
  D(P\|Q) = \sum_{i=1}^k P_i \log \frac{P_i}{Q_i},
\end{equation}
where $P$ and $Q$ are supported on a common alphabet set $[k] \triangleq \{1, \dots, k\}$, and
$P$ is absolutely continuous with respect to $Q$, i.e.,
if $Q_i=0$, $P_i=0$, for $i \in [k]$. We use $\mathcal{M}_k$ to denote the collection of all such pairs of distributions.

Suppose $P$ and $Q$ are unknown, and that $m$ independent and identically distributed (i.i.d.) samples $X_1,\ \dots\ ,X_m$ drawn from $P$ and $n$ i.i.d.\ samples $Y_1,\ \dots\ ,Y_n$ drawn from $Q$ are available for estimation. The sufficient statistics for estimating $D(P\|Q)$ are the histograms of the samples $M \triangleq (M_1, \dots ,M_k)$ and $N \triangleq (N_1, \dots ,N_k)$, where
\begin{equation}
  M_j = \sum_{i=1}^m \mathds{1}_{\{X_i=j\}}\quad \mbox{ and }\quad N_j = \sum_{i=1}^n \mathds{1}_{\{Y_i=j\}}
\end{equation}
record the numbers of occurrences of $j \in [k]$ in samples drawn from $P$ and $Q$, respectively. Then $M\sim  \mathrm{Multinomial}\,(m, P)$ and $N\sim  \mathrm{Multinomial}\,(n, Q)$.
An estimator $\hat{D}$ of $D(P\|Q)$ is then a function of the histograms $M$ and $N$, denoted by
$\hat{D}(M,N)$.


We adopt the following worst-case quadratic risk to measure the performance of estimators of the KL divergence:
\begin{equation}\label{eq:suprisk}
  R(\hat{D},k,m,n) \triangleq  \sup_{(P,Q) \in \mathcal{M}_k} \mathbb{E}\big[\big(\hat{D}(M,N)-D(P\|Q)\big)^2\big].
\end{equation}
We further define the  minimax quadratic risk as:
\begin{equation}
  R^*(k,m,n) \triangleq \inf_{\hat{D}} R(\hat{D},k,m,n).
\end{equation}

In this paper, we are interested in the large-alphabet regime with $k \rightarrow \infty$.  Furthermore, the number $m$ and $n$ of samples are functions of $k$, which are allowed to scale with $k$ to infinity.
\begin{definition}
A sequence of estimators $\hat{D}$, indexed by $k$, is said to be consistent under sample complexity $m(k)$ and $n(k)$ if
\begin{equation}
  \lim_{k\to \infty}  R(\hat{D},k,m,n) = 0.
\end{equation}
\end{definition}

We are also interested in the following set:
\begin{align}\label{eq:boundedratioset}
  \mathcal{M}&_ {k, f(k)}  \nn \\
  &= \Big\{(P,Q): |P|=|Q|=k,\frac{P_i}{Q_i} \le f(k),\ \forall i \in [k]\Big\},
\end{align}
which contains distributions $(P,Q)$ with density ratio bounded by $f(k)$.

We define the worst-case quadratic risk over $\mathcal{M}_{k,f(k)}$ as
\begin{align}
  R(\hat{D}, k, & m,n,f(k))    \nn \\
  \triangleq & \sup_{(P,Q) \in \mathcal{M}_{k,f(k)}} \mathbb{E}\big[\big(\hat{D}(M,N)-D(P\|Q)\big)^2\big],
\end{align}
and define the corresponding minimax quadratic risk as
\begin{equation}\label{eq:minimaxriskdensity}
  R^*(k,m,n,f(k)) \triangleq \inf_{\hat{D}} R(\hat{D},k,m,n,f(k)).
\end{equation}

\subsection{Notations}
We adopt the following notation to express asymptotic scaling of quantities with $n$: $f(n)\lesssim g(n)$ represents that there exists a constant $c$ s.t. $f(n)\leq cg(n)$; $f(n)\gtrsim g(n)$ represents that there exists a constant $c$ s.t.\ $f(n)\geq cg(n)$; $f(n)\asymp g(n)$ when $f(n)\gtrsim g(n)$ and $f(n)\lesssim g(n)$ hold simultaneously;   $f(n)\gg g(n)$ represents that for all $c>0$, there exists $n_0>0$ s.t.\ for all $ n>n_0$, $|f(n)|\geq c|g(n)|$; and $f(n)\ll g(n)$ represents that for all $c>0$, there exists $n_0>0$ s.t.\ for all $ n>n_0$, $|f(n)|\leq cg(n)$.
\subsection{Comparison to Related Problems}

Several estimators of KL divergence when $P$ and $Q$ are {\em continuous} have been proposed and shown to be consistent. The estimator proposed in \cite{wang2005divergence} is based on data-dependent partition on the densities, the estimator proposed in \cite{wang2009divergence} is based on a K-nearest neighbor approach, and the estimator developed in \cite{nguyen2010estimating} utilizes a kernel-based approach for estimating the density ratio. A more general problem of estimating the $f$-divergence was studied in \cite{moon2014ensemble}, where an estimator based on a weighted ensemble of plug-in estimators was proposed to trade bias with variance. All of these approaches exploit the smoothness of continuous densities or density ratios, which guarantees that samples falling into a certain neighborhood can be used to estimate the local density or density ratio accurately. However, such a smoothness property does not hold for discrete distributions, whose probabilities over adjacent point masses can vary significantly. In fact, an example is provided in \cite{wang2005divergence} to show that the estimation of KL divergence can be difficult even for continuous distributions if the density has sharp dips.


Estimation of KL divergence when the distributions $P$ and $Q$ are discrete has been studied in \cite{zhang2014nonparametric,cai2002,cai2006universal} for the regime with {\em fixed} alphabet size $k$ and large sample sizes $m$ and $n$. Such a regime is very different from the large-alphabet regime in which we are interested, with $k$ scaling to infinity. Clearly, as $k$ increases, the scaling of the sample sizes $m$ and $n$ must be fast enough with respect to $k$ in order to guarantee consistent estimation.


In the large-alphabet regime, KL divergence estimation is closely related to entropy estimation with a large alphabet recently studied in \cite{paninski2003estimation,antos2001convergence,valiant2011estimating,wu2014minimax,Jiao2014plugin,jiao2015minimax,jiao2017nearest}. Compared to entropy estimation, KL divergence estimation has one more dimension of uncertainty, that is regarding the distribution $Q$. Some distributions $Q$ can contain very small point masses that contribute significantly to the value of divergence, but are difficult to estimate because samples of these point masses occur rarely. In particular, such distributions dominate the risk
in \eqref{eq:suprisk}, and make the construction of consistent estimators challenging.

\subsection{Summary of Main Results}

We summarize our main results in the following three theorems, more details are given respectively in Sections \ref{sec:lowerbound1}, \ref{sec:plugin} and \ref{sec:minimax}.

Our first result, based on Le Cam's two-point method \cite{Tsybakov2008}, is that there is no consistent estimator of KL divergence over the distribution set $\mathcal{M}_k$.
\begin{theorem}\label{thm:MLB}
For any  $m,n \in \mathbb{N}$, and $k\ge2, $ $R^*(k,m,n)$ is infinite.
Therefore, there does not exist any consistent estimator of KL divergence over the set $\mathcal{M}_k$.
\end{theorem}

The intuition behind this result is that the set $\mathcal{M}_k$ contains distributions $Q$, that have arbitrarily small components that contribute significantly to KL divergence but require arbitrarily large number of samples to estimate accurately. However, in practical applications, it is reasonable to assume that the ratio of $P$ to $Q$ is bounded, so that the KL divergence $D(P\|Q)$ is bounded. Thus, we further focus on the set $\mathcal{M}_{k,f(k)}$ given in \eqref{eq:boundedratioset} that contains distribution pairs $(P,Q)$ with their density ratio bounded by $f(k)$.

\begin{remark}
Consider  the R\'enyi divergence of order $\alpha$ between distributions $P$ and $Q$, which is defined as
\begin{equation}
  D_{\alpha}(P\|Q) \triangleq \frac{1}{\alpha-1} \log\Big(\sum_{i=1}^k \frac{p^{\alpha}_i}{q^{\alpha-1}_i}\Big).
\end{equation}
Then the KL divergence is equivalent to the R\'enyi divergence of order one. Moreover, the bounded density ratio condition is equivalent to the following upper bound on the R\'enyi divergence of order infinity,
\begin{equation}
D_{\infty}(P\|Q) = \log \sup_{i\in [k]}\frac{P_i}{Q_i} \le \log f(k).
\end{equation}
It is shown in \cite{van2014renyi} that $D_{\alpha}$ is non-decreasing for $\alpha >1$, i.e., $D(P\|Q)\le D_\alpha(P\|Q)\le D_\infty(P\|Q)$. This implies that the conditions on any order $\alpha$ of the R\'enyi divergence can be used to bound the KL divergence as long as $\alpha>1$. For the sake of simplicity and convenience, we adopt the density ratio bound in this paper.
\end{remark}


%
{We construct an augmented plug-in estimator $\hat{D}_{\mathrm{A-plug-in}}$, defined in \eqref{eq:pluginest}, and characterize its worst-case quadratic risk over the set $\mathcal{M}_{k,f(k)}$ in the following theorem.
\begin{theorem}\label{thm:plugin}
For any  $k\in \mathbb{N}$, $m \ge k$ and $n \ge 10kf(k)$, the worst-case quadratic risk of the augmented plug-in estimator defined in \eqref{eq:pluginest} over the set $\mathcal{M}_{k,f(k)}$ satisfies
\begin{align}\label{eq:plugin_risk}
  R(\hat{D}_{\mathrm{A-plug-in}}&,\ k,m,n,f(k))  \nn \\
  \asymp &\left(\frac{kf(k)}{n} + \frac{k}{m}\right)^2 +\frac{\log^2f(k)}{m} + \frac{f(k)}{n}.
\end{align}
\end{theorem}}

The upper bound is derived by evaluating the bias and variance of the estimator separately. In order to prove the term $\big(\frac{kf(k)}{n} + \frac{k}{m}\big)^2$ in the lower bound, we analyze the bias of the estimator and construct different pairs of ``worst-case" distributions, respectively, for the cases where  the bias caused by insufficient samples from $P$ or the bias caused by insufficient samples from $Q$ dominates. The terms $\frac{\log^2f(k)}{m}$ and $\frac{f(k)}{n}$ in the lower bound are due to the variance and follow from the minimax lower bound given by Le Cam's two-point method with a judiciously chosen pair of distributions.

{Theorem \ref{thm:plugin} implies that the augmented plug-in estimator is consistent over $\mathcal{M}_{k,f(k)}$ if and only if
\begin{equation}\label{eq:aplugin}
m \gg k\vee\log^2f(k) \text{ and } n \gg kf(k).
\end{equation}}
Thus, the number of samples $m$ and $n$ should be larger than the alphabet size $k$ for the plug-in estimator to be consistent. This naturally inspires the question of whether the plug-in estimator achieves the minimax risk, and if not, what estimator is minimax optimal and what is the corresponding minimax risk.

We  show that the augmented plug-in estimator is not minimax optimal, and that the minimax optimality can be achieved by an estimator that employs both the polynomial approximation and plug-in approaches, and the following theorem characterizes the minimax risk.
\begin{theorem}\label{thm:minimax}
  If $m\gtrsim \frac{k}{\log k}$, $n\gtrsim \frac{kf(k)}{\log k}$, $ f(k) \ge \log^2 k$, $\log m\lesssim\log k$ and $\log ^2 n\lesssim k^{1-\epsilon}$, where $\epsilon$ is any positive constant, then the minimax risk satisfies
  \begin{flalign}\label{eq:minimax}
     R^*(k,m,&n,f(k))\nn \\
      \asymp  &\left(\frac{k}{m\log k}+\frac{kf(k)}{n\log k}\right)^2 +\frac{\log ^2 f(k)}{m}+\frac{f(k)}{n}.
  \end{flalign}
\end{theorem}

The key idea in the construction of the minimax optimal estimator is the application of a polynomial approximation to reduce the bias in the regime where the bias of the plug-in estimator is large. Compared to entropy estimation \cite{wu2014minimax,jiao2015minimax}, the challenge here is that the KL divergence is a function of two variables, for which a joint polynomial approximation is difficult to derive. We solve this problem by employing separate polynomial approximations for functions involving $P$ and $Q$ as well as judiciously using the density ratio constraint to bound the estimation error. The proof of the lower bound on the minimax risk is based on a generalized Le Cam's method involving two composite hypotheses,
as in the case of entropy estimation \cite{wu2014minimax}. But the challenge here that requires special technical treatment is the construction of prior distributions for $(P,Q)$ that satisfy the bounded density ratio constraint.

We note that the first term $\big(\frac{k}{m\log k}+\frac{kf(k)}{n\log k}\big)^2 $ in \eqref{eq:minimax} captures the squared bias, and the remaining terms correspond to the variance. If we compare the worst-case quadratic risk for the augmented plug-in estimator in \eqref{eq:plugin_risk} with the minimax risk in \eqref{eq:minimax}, there is a $\log k$ factor rate improvement in the bias.


Theorem \ref{thm:minimax} directly implies that in order to estimate the KL divergence over the set $\mathcal{M}_{k,f(k)}$ with vanishing mean squared error, the sufficient and necessary conditions on the sample complexity are given by
\begin{flalign}\label{eq:samplecomplexity}
  m\gg \Big(\log^2 f(k)\vee \frac{k}{\log k}\Big), \text { and }n\gg \frac{kf(k)}{\log k}.
\end{flalign}
The comparison of \eqref{eq:aplugin} with \eqref{eq:samplecomplexity} shows that the augmented plug-in estimator is strictly sub-optimal.

{While our results are proved under a ratio upper bound constraint $f(k)$, under certain upper bounds on $m$ and $n$, and by taking a worst case over $(P,Q)$, we make the following (non-rigorous) observation that is implied by our results, by disregarding the aforementioned assumptions to some extent. If $Q$ is known to be the uniform distribution (and hence $f(k)\le k$), then $D(P\|Q) = \sum P_i\log P_i-\log k$, and the KL divergence estimation problem (without taking the worst case over $Q$) reduces to the problem of estimating the entropy of distribution $P$. More specifically, letting $n = \infty$ and $f(k)\le k$ in \eqref{eq:minimax} directly yields $\big(\frac{k}{m\log k}\big)^2 + \frac{\log ^2 k}{m}$, which is the same as the minimax risk for entropy estimation \cite{wu2014minimax,jiao2015minimax}.} 

We note that after our initial conference submission was accepted by ISIT 2016 and during our preparation of this full version of the work, an independent study of the same problem of KL divergence estimation was posted on arXiv \cite{han2016}.

A comparison between our results and the results in \cite{han2016} shows that the constraints for the minimax rate to hold in Theorem \ref{thm:minimax} are weaker than those in \cite{han2016}.
On the other hand, our estimator requires the knowledge of $k$ to determine whether to use polynomial approximation or the plug-in approach, whereas the estimator proposed in \cite{han2016} based on the same idea does not require such knowledge. Our estimator also requires the knowledge of $f(k)$ to keep the estimate between $0$ and $\log f(k)$ in order to simplify the theoretical analysis, but our experiments demonstrate that desirable performance can be achieved even without such a step (and correspondingly without exploiting $f(k)$). However, in practice, the knowledge of $k$ is typically useful to determine the number of samples that should be taken in order to achieve a certain level of estimation accuracy. In situations without the knowledge of $k$, the estimator in [18] has the advantage of being independent from $k$ and performs well adaptively.

\section{No Consistent Estimator over $\mathcal{M}_k$}\label{sec:lowerbound1}

Theorem \ref{thm:MLB} states that the minimax risk over the set $\mathcal{M}_k$ is unbounded for arbitrary alphabet size $k$ and $m$ and $n$ samples, which suggests that there is no consistent estimator for the minimax risk over $\mathcal{M}_k$.

This result follows from Le Cam's two-point method \cite{Tsybakov2008}: If two pairs of distributions $(P^{(1)}, Q^{(1)})$ and $(P^{(2)}, Q^{(2)})$ are sufficiently close such that it is impossible to reliably distinguish between them using $m$ samples from $P$ and $n$ samples from $Q$ with error probability less than some constant, then any estimator suffers a quadratic risk proportional to the squared difference between the divergence values, $(D(P^{(1)}\|Q^{(1)})-D(P^{(2)}\|Q^{(2)}))^2$.

We next give examples to illustrate how the distributions used in Le Cam's two-point method can be constructed.

For the case $k=2$, we let $P^{(1)}=P^{(2)} = (\frac{1}{2},\frac{1}{2})$, $Q^{(1)}=(e^{-s},1-e^{-s})$ and $Q^{(2)} = (\frac{1}{2s},1-\frac{1}{2s})$, where $s>0$. For any $n \in \mathbb{N}$, we choose $s$ sufficiently large such that $D(Q^{(1)}\|Q^{(2)})<\frac{1}{n}$. Thus, the error probability for distinguishing $Q^{(1)}$ and $Q^{(2)}$ with $n$ samples is greater than a constant. However, $D(P^{(1)}\|Q^{(1)}) \asymp s$ and $D(P^{(2)}\|Q^{(2)}) \asymp\log s$. Hence, the minimax risk, which is lower bounded by the difference of the above divergences, can be made arbitrarily large by letting $s\rightarrow\infty$. This example demonstrates that two pairs of distributions $(P^{(1)},Q^{(1)})$ and $(P^{(2)},Q^{(2)})$ can be very close so that the data samples are almost indistinguishable, but the KL divergences $D(P^{(1)}\|Q^{(1)})$ and $D(P^{(2)}\|Q^{(2)})$ can still be far away.

For the case with $k > 2$, the distributions can be constructed based on the same idea by distributing one of the mass in the above $(P^{(1)}, Q^{(1)})$ and $(P^{(2)}, Q^{(2)})$ uniformly on the remaining $k-1$ bins. Thus, it is impossible to estimate the KL divergence accurately over the set $\mathcal{M}_k$ under the minimax setting.

\section{Augmented Plug-in Estimator over $\mathcal{M}_{k,f(k)}$}\label{sec:plugin}
Since there does not exist any consistent estimator of KL divergence over the set $\mathcal{M}_k$, we study KL divergence estimators over the set $\mathcal{M}_{k,f(k)}$.

The ``plug-in'' approach is a natural way to estimate the KL divergence, namely, first estimate the distributions and then substitute these estimates into the divergence function. This leads to the following plug-in estimator, i.e., the empirical divergence
\begin{equation}\label{eq:directplugin}
  \hat{D}_{\mathrm{plug-in}} (M,N) =D(\hat{P}\|\hat{Q}),
\end{equation}
where $\hat{P} = (\hat{P}_1,\dots,\hat{P}_k)$ and $\hat{Q} = (\hat{Q}_1,\dots,\hat{Q}_k)$ denote the empirical distributions with $\hat{P}_i=\frac{M_i}{m}$ and $\hat{Q}_i=\frac{N_i}{n}$, respectively, for $i=1,\cdots,k$.

Unlike the entropy estimation problem, where the plug-in estimator $\hat{H}_{\mathrm{plug-in}}$ is asymptotically efficient in the ``fixed $P$, large $n$'' regime, the direct plug-in estimator $\hat{D}_{\mathrm{plug-in}}$ in \eqref{eq:directplugin} of KL divergence has an infinite bias. This is because of the non-zero probability of $N_j=0$ and $M_j\ne 0$, for some $j\in[k]$, which leads to infinite $\hat{D}_{\mathrm{plug-in}}$.


{We can get around the above issue associated with the direct plug-in estimator as follows. We first add a fraction of a sample to each mass point of $Q$, and then normalize and take $\hat{Q'_i}=\frac{N_i+c}{n+kc}$ as an estimate of $Q_i$, where $c$ is a constant. We refer to $\hat{Q'_i}$ as the add-constant estimator \cite{orlitsky2015competitive} of $Q$. Clearly, $\hat{Q'_i}$ is non-zero for all $i$. We therefore propose the following ``augmented plug-in" estimator based on $\hat{Q'_i}$
\begin{equation}\label{eq:pluginest}
  \hat {D}_{\mathrm{A-plug-in}}(M,N) = \sum_{i=1}^k \frac{M_i}{m} \log\frac{M_i/m}{(N_i+c)/(n+kc)}.
\end{equation}
}



Theorem \ref{thm:plugin} characterizes the worst-case quadratic risk of the augmented plug-in estimator over $\mathcal{M}_{k,f(k)}$. The proof of Theorem \ref{thm:plugin} involves the following two propositions, which provide upper and lower bounds on $R(\hat{D}_{\mathrm{A-plug-in}},k,m,n,f(k))$, respectively. 
\begin{proposition}\label{prop:upper-plug}
For all $k \in \mathbb{N}$, 
\begin{align}\label{eq:pluginupper}
  R(\hat{D}_{\mathrm{A-plug-in}}&,k,m,n,f(k)) \nn \\
  \lesssim   & \left(\frac{kf(k)}{n} + \frac{k}{m}\right)^2 +\frac{\log^2f(k)}{m} + \frac{f(k)}{n}.
\end{align}
\end{proposition}

\begin{proof}[Outline of Proof]
The proof consists of separately bounding  the bias and variance of the augmented plug-in estimator. The details are provided in Appendix \ref{app:prop1}.
\end{proof}
It can be seen that in the risk bound \eqref{eq:pluginupper}, the first term captures the squared bias, and the remaining terms correspond to the variance.
{
\begin{proposition}\label{prop:lower-plug}
For all $k \in \mathbb{N}$, $m\ge k$, $n \ge 10kf(k)$, $f(k)\ge 10$, and $c \in [\frac{2}{3},\frac{5}{4}]$,
\begin{align}\label{eq:pluginlower}
  R(\hat{D}_{\mathrm{A-plug-in}}&,k,m,n,f(k)) \nn \\
   \gtrsim &\left(\frac{k}{m} + \frac{kf(k)}{n} \right)^2 +\frac{\log^2f(k)}{m} + \frac{f(k)}{n}.
\end{align}
\end{proposition}}

\begin{proof}[Outline of Proof]
We provide the basic idea of the proof here with the details provided in Appendix \ref{app:prop2}.

{
We first derive the terms corresponding to the squared bias in the lower bound
by choosing two different pairs of worst-case distributions. Note that the bias of the augmented plug-in estimator can be decomposed into: (1) the bias due to estimating $\sum_{i=1}^k P_i\log P_i$; and (2) the bias due to estimating $\sum_{i=1}^k -P_i\log Q_i$. Since $x\log x$ is a convex function, we can show that the first bias term is always positive. But the second bias term can be negative, so that the two bias terms may cancel out partially or even fully. Thus, to derive the minimax lower bound, we first determine which bias term dominates, and then construct a pair of distributions such that the dominant bias term is either lower bounded by some positive terms or upper bounded by negative terms.}

{\textbf{Case I:}
If $\frac{k}{m}\ge (1+\epsilon) \frac{c k f(k)}{5 n}$, for some $\epsilon>0$, which implies that the number of samples drawn from $P$ is relatively smaller than the number of samples drawn from $Q$, the first bias term dominates. Setting $P$ to be uniform and $Q  =\big(\frac{10}{kf(k)},\ \cdots,\  \frac{10}{kf(k)},\ 1-\frac{10(k-1)}{kf(k)}\big)$, it can be shown that the bias is lower bounded by $\frac{k}{m} + \frac{kf(k)}{n}$ in the order sense.}

{\textbf{Case II:}
If $\frac{k}{m}< (1+\epsilon) \frac{c k f(k)}{5n}$, which implies that the number of samples drawn from $P$ is relatively larger than the number of samples drawn from $Q$, the second bias term dominates. Setting $P=\big(\frac{f(k)}{4n},\cdots,\frac{f(k)}{4n},1-\frac{(k-1)f(k)}{4n}\big)$, and $Q=\big(\frac{1}{4n},\cdots,\frac{1}{4n}, 1-\frac{k-1}{4n}\big)$, it can be shown that the bias is upper bounded by $-\big(\frac{k}{m} + \frac{kf(k)}{n}\big)$ in the order sense.}


{The proof of the terms corresponding to the variance in the lower bound can be done using the minimax lower bound given by Le Cam's two-point method.} 
\end{proof}

\section{Minimax Quadratic Risk over $\mathcal{M}_{k,f(k)}$}\label{sec:minimax}

Our third main result Theorem \ref{thm:minimax} characterizes the minimax quadratic risk (within a constant factor) of estimating KL divergence over $\mathcal{M}_{k,f(k)}$. In this section, we describe ideas and central arguments to show this theorem with detailed proofs relegated to the appendix.

\subsection{Poisson Sampling}\label{sec:poisson}
The sufficient statistics for estimating $D(P\|Q)$ are the histograms of the samples $M=(M_1,\ldots,M_k)$ and $N=(N_1,\ldots,N_k)$, and $M$ and $N$ are multinomial distributed.   However, the histograms are not independent across different bins, which is hard to analyze.
In this subsection, we introduce the \emph{Poisson sampling} technique to handle the dependency of the multinomial distribution across different bins, as in \cite{wu2014minimax} for entropy estimation. Such a technique is used in our proofs to develop the lower and upper bounds on the minimax risk in Sections \ref{sec:minimaxlow} and \ref{sec:minimaxup}.

In Poisson sampling, we replace the deterministic sample sizes $m$ and $n$ with Poisson random variables $m' \sim \mathrm{Poi}(m)$ with mean $m$ and $n' \sim \mathrm{Poi}(n)$ with mean $n$, respectively. Under this model,
we draw $m'$ and $n'$ i.i.d. samples from $P$ and $Q$, respectively. The sufficient statistics $M_i \sim \mathrm{Poi}(nP_i)$ and $N_i \sim \mathrm{Poi}(nQ_i)$ are then independent across different bins, which significantly simplifies the analysis.


Analogous to the minimax risk \eqref{eq:minimaxriskdensity}, we define its counterpart under the Poisson sampling model as
\begin{flalign}\label{eq:poissonsampling}
   \widetilde{R}^*(&k,m,n,f(k))\nn \\
   &\triangleq \inf_{\hat{D}} \sup_{(P,Q) \in \mathcal{M}_{k,f(k)}} \mathbb{E}\big[\big(\hat{D}(M,N)-D(P\|Q)\big)^2\big],
\end{flalign}
where the expectation is taken over $M_i\sim \mathrm{Poi}(nP_i)$ and $N_i\sim \mathrm{Poi}(nQ_i)$ for $i\in [k]$. Since the Poisson sample sizes are concentrated near their means $m$ and $n$ with high probability, the minimax risk under Poisson sampling is close to that with fixed sample sizes as stated in the following lemma.
\begin{lemma}\label{lemma:poisson}
There exists a constant $b>\frac{1}{4}$ such that
\begin{align}\label{eq:poissonineq}
  \widetilde{R}^*&(k,2m,2n,f(k))-e^{-bm}\log^2 f(k) -e^{-bn}\log^2 f(k) \nn \\
  &\le R^*(k,m,n,f(k))\le 4\widetilde{R}^*(k,m/2,n/2,f(k)).
\end{align}
\end{lemma}
\begin{proof}
  See Appendix \ref{app:poisson}.
\end{proof}

Thus, in order to show Theorem \ref{thm:minimax}, it suffices to bound the Poisson risk $\widetilde{R}^*(k,m,n,f(k))$. In Section \ref{sec:minimaxlow}, a lower bound on the minimax risk with deterministic sample size is derived, and in Section \ref{sec:minimaxup}, an upper bound on the minimax risk with Poisson sampling is derived, which further yields an upper bound on the minimax risk with deterministic sample size. It can be shown that the upper and lower bounds match each other (up to a constant factor).


\subsection{Minimax Lower Bound}\label{sec:minimaxlow}

In this subsection, we develop the following lower bound on the minimax risk for the estimation of KL divergence over the set $\mathcal{M}_{k,f(k)}$.
\begin{proposition}\label{prop:GMLBentropy}
If $m\gtrsim \frac{k}{\log k}$, $n\gtrsim \frac{kf(k)}{\log k}$, $f(k) \geq \log^2 k$ and $\log ^2 n\lesssim k$,
\begin{align}\label{eq:lowerbound}
  {R}^*&(k,m,n,f(k)) \nn \\
  &\gtrsim \left(\frac{k}{m\log k}+\frac{kf(k)}{n\log k}\right)^2 +\frac{\log ^2 f(k)}{m}+\frac{f(k)}{n}.
\end{align}
\end{proposition}
\begin{proof}[Outline of Proof]
We describe the main idea in the development of the lower bound, with the detailed proof provided in Appendix \ref{app:prop3}.

To show Proposition \ref{prop:GMLBentropy}, it suffices to show that the minimax risk is lower bounded separately by each individual terms in \eqref{eq:lowerbound} in the order sense. The proof for the last two terms requires the Le Cam's two-point method, and the proof for the first term requires more general method, as we outline in the following.

\textbf{Le Cam's two-point method:} The last two terms in the lower bound correspond to the variance of the estimator.

The bound ${R}^*(k,m,n,f(k))\gtrsim \frac{\log^2 f(k)}{m}$ can be shown by setting
\begin{small}
\begin{align}
      P^{(1)} & =\Big(\frac{1}{3(k-1)},\ \dots,\ \frac{1}{3(k-1)},\ \frac{2}{3} \Big), \\
      P^{(2)} & = \Big(\frac{1-\epsilon}{3(k-1)},\ \dots,\ \frac{1-\epsilon}{3(k-1)},\ \frac{2+\epsilon}{3} \Big),\\
      Q^{(1)} & =Q^{(2)} \nn \\
      & =\Big(\frac{1}{3(k-1)f(k)},\ \dots,\ \frac{1}{3(k-1)f(k)},\ 1-\frac{1}{3f(k)} \Big) ,
\end{align}
\end{small}
where $\epsilon = \frac{1}{\sqrt{m}}$.


The bound ${R}^*(k,m,n,f(k))\gtrsim \frac{f(k)}{n}$ can be shown by choosing
\begin{small}
\begin{align}
      P^{(1)} & =P^{(2)}  =\Big(\frac{1}{3(k-1)},\ 0,\ \dots,\ \frac{1}{3(k-1)},\ 0,\ \frac{5}{6} \Big), \\
      Q^{(1)} & = \Big(\frac{1}{2(k-1)f(k)},\ \dots,\ \frac{1}{2(k-1)f(k)},\ 1-\frac{1}{2f(k)} \Big),\\
      Q^{(2)}& =\Big(\frac{1-\epsilon}{2(k-1)f(k)},\ \frac{1+\epsilon}{2(k-1)f(k)},\ \dots,\ \nn \\
      & \quad \quad \frac{1-\epsilon}{2(k-1)f(k)},\ \frac{1+\epsilon}{2(k-1)f(k)},\ 1-\frac{1}{2f(k)} \Big),
\end{align}
\end{small}
where $\epsilon= \sqrt{\frac{f(k)}{n}}$.

\textbf{Generalized Le Cam's method:} In order to show ${R}^*(k,m,n,f(k))\gtrsim \Big(\frac{k}{m\log k}+\frac{kf(k)}{n\log k}\Big)^2$, it suffices to show that ${R}^*(k,m,n,f(k))\gtrsim \Big(\frac{k}{m\log k}\Big)^2$ and ${R}^*(k,m,n,f(k))\gtrsim\Big(\frac{kf(k)}{n\log k}\Big)^2$.
These two lower bounds can be shown by applying a generalized Le Cam's method, which involves the following two composite hypotheses \cite{Tsybakov2008}:
\begin{flalign*}
  H_0: D(P\|Q)\leq t \quad \text{ versus } \quad H_1: D(P\|Q)\geq t+d.
\end{flalign*}
Le Cam's two-point approach is a special case of this generalized method. If no test can distinguish $H_0$ and $H_1$ reliably, then we obtain a lower bound on the quadratic risk with order $d^2$. Furthermore, the optimal probability of error for composite hypothesis testing is equivalent to the Bayesian risk under the least favorable priors. Our goal here is to construct two prior distributions on $(P,Q)$ (respectively for two hypotheses), such that the two corresponding divergence values are separated (by $d$), but the error probability of distinguishing between the two hypotheses is large. However, it is difficult to design joint prior distributions on $(P,Q)$ that satisfy all the above desired property. In order to simplify this procedure, we set one of the distributions $P$ and $Q$ to be known. Then the minimax risk when both $P$ and $Q$ are unknown is lower bounded by the minimax risk with only either $P$ or $Q$ being unknown. In this way, we only need to design priors on one distribution, which can be shown to be sufficient for the proof of the lower bound.

{As shown in \cite{wu2014minimax}, the strategy which chooses two random variables with moments matching up to a certain degree ensures the impossibility to test in the minimax entropy estimation problem. The minimax lower bound is then obtained by maximizing the expected separation $d$ subject to the moment matching condition. For our KL divergence estimation problem, this approach also yields the optimal minimax lower bound, but the challenge here that requires special technical treatment is the construction of prior distributions for $(P,Q)$ that satisfy the bounded density ratio constraint.}

In order to show ${R}^*(k,m,n,f(k))\gtrsim \Big(\frac{k}{m\log k}\Big)^2$, we set $Q$ to be the uniform distribution and assume it is known. Therefore, the estimation of $D(P\|Q)$ reduces to the estimation of $\sum_{i=1}^k P_i\log P_i$, which is the minus entropy of $P$. Following steps similar to those in \cite{wu2014minimax}, we can obtain the desired result.

In order to show ${R}^*(k,m,n,f(k))\gtrsim\big(\frac{kf(k)}{n\log k}\big)^2$, we set
\begin{equation}
P=\Big(\frac{f(k)}{n\log k},\ldots,\frac{f(k)}{n\log k},1-\frac{(k-1)f(k)}{n\log k}\Big),
\end{equation}
and assume $P$ is known.
Therefore, the estimation of $D(P\|Q)$ reduces to the estimation of $\sum_{i=1}^kP_i\log Q_i$. We then properly design priors on $Q$ and apply the generalized Le Cam's method to obtain the desired result.
\end{proof}

We note that the proof of Proposition \ref{prop:GMLBentropy} may be strengthened by designing jointly distributed priors on $(P,Q)$, instead of treating them separately. This may help to relax or remove the conditions $f(k)\ge \log^2 k$  and $\log^2 n\lesssim k$ in Proposition \ref{prop:GMLBentropy}.

\subsection{Minimax Upper Bound via Optimal Estimator}\label{sec:minimaxup}

Comparing the lower bound in Proposition \ref{prop:GMLBentropy} with the upper bound in Proposition \ref{prop:upper-plug} that characterizes an upper bound on the risk for the augmented plug-in estimator, it is clear that there is a difference of a $\log k$ factor in the bias terms, which implies that the augmented plug-in estimator is not minimax optimal. A promising approach to fill in this gap is to design an improved estimator. Entropy estimation \cite{wu2014minimax,jiao2015minimax} suggests the idea to incorporate a polynomial approximation into the estimator in order to reduce the bias with price of the variance. In this subsection, we construct an estimator using this approach, and characterize an upper bound on the minimax risk in Proposition \ref{prop:upperbound}.


The KL divergence $D(P\|Q)$ can be written as
\begin{equation}
D(P\|Q)=\sum_{i=1}^k P_i\log P_i-\sum_{i=1}^k P_i\log Q_i.
\end{equation}
The first term equals the minus entropy of $P$, and the minimax optimal entropy estimator (denoted by $\hat D_1$) in \cite{wu2014minimax} can be applied to estimate it.
The major challenge in estimating $D(P\|Q)$ arises due to the second term.
We overcome the challenge by using a polynomial approximation to reduce the bias when $Q_i$ is small. Under Poisson sampling model, unbiased estimators can be constructed for any polynomials of $P_i$ and $Q_i$. Thus, if we approximate $P_i\log Q_i$ by polynomials, and then construct unbiased estimator for the polynomials, the bias of estimating $P_i\log Q_i$ is reduced to the error in the approximation of $P_i\log Q_i$ using polynomials.

A natural idea is to construct polynomial approximation for $|P_i\log Q_i|$ in two dimensions, exploiting the fact that $|P_i\log Q_i|$ is bounded by $f(k)$ in the order sense. The authors of \cite{han2016} also discuss the idea of a two-dimensional polynomial approximation. However, it is challenging to find the explicit form of the two-dimensional polynomial approximation for estimating KL divergence as shown in \cite{totik2014polynomial}. However, for some problems it is still worth exploring the two-dimensional approximation directly, as shown in \cite{jiao2016minimax}, where no one-dimensional approximation can achieve the minimax rate.

On the other hand, a one-dimensional polynomial approximation of $\log Q_i$ also appears challenging to develop. First of all, the function $\log x$ on interval $(0,1]$ is not bounded due to the singularity point at $x=0$. Hence, the approximation of $\log x$ when $x$ is near the point $x=0$ is inaccurate. Secondly, such an approach implicitly ignores the fact that $\frac{P_i}{Q_i}\leq f(k)$, which implies that when $Q_i$ is small, the value of $P_i$ should also be small.

Another approach is to rewrite the function $P_i\log Q_i$ as $(\frac{P_i}{Q_i})Q_i\log Q_i$, and then estimate $\frac{P_i}{Q_i}$ and $Q_i\log Q_i$ separately. Although the function $Q_i\log Q_i$ can be approximated using polynomial approximation and then estimated accurately (see \cite[Section 7.5.4]{tim63} and \cite{wu2014minimax}), it is difficult to find a good estimator for $\frac{P_i}{Q_i}$.

Motivated by those unsuccessful approaches, we design our estimator as follows. We rewrite $P_i\log Q_i$ as $P_i\frac{Q_i\log Q_i}{Q_i}$. When $Q_i$ is small, we construct a polynomial approximation $\mu_L(Q_i)$ for $Q_i\log Q_i$, which does not contain a zero-degree term. Then, $\frac{\mu_L(Q_i)}{Q_i}$ is also a polynomial, which can be used to approximate $\log Q_i$. Thus, an unbiased estimator for $\frac{\mu_L(Q_i)}{Q_i}$ is constructed. Note that the error in the approximation of $\log Q_i$ using $\frac{\mu_L(Q_i)}{Q_i}$ is not bounded, which implies that the bias of using unbiased estimator of $\frac{\mu_L(Q_i)}{Q_i}$ to estimate $\log Q_i$ is not bounded.  However, we can show that the bias of estimating $P_i\log Q_i$ is bounded, which is due to the density ratio constraint $f(k)$. The fact that when $Q_i$ is small, $P_i$ is also small helps to reduce the bias. In the following, we will introduce how we construct our estimator in detail.


By Lemma \ref{lemma:poisson}, we apply Poisson sampling to simplify the analysis. We first draw $m_1'\sim$Poi$(m)$, and $m_2'\sim$Poi$(m)$, and then draw $m_1'$ and $m_2'$ i.i.d. samples from distribution $P$, where we use $M=(M_1,\ldots,M_k)$ and $M'=(M'_1,\ldots,M'_k)$ to denote the histograms of $m'_1$ samples and $m'_2$ samples, respectively. We then use these  samples to estimate $\sum_{i=1}^k P_i\log P_i$ following the entropy estimator proposed in \cite{wu2014minimax}.
Next, we draw $n'_1\sim$Poi$(n)$ and $n'_2\sim$Poi$(n)$ independently. We then draw $n'_1$ and $n'_2$ i.i.d. samples from distribution $Q$, where we use $N=(N_1,\ldots,N_k)$ and $N'=(N'_1,\ldots,N'_k)$ to denote the histograms of $n'_1$ samples and $n'_2$ samples, respectively. We note that $N_i{\sim}$Poi$(nQ_i)$, and $N_i'{\sim}$Poi$(nQ_i)$.

We then focus on the estimation of $\sum_{i=1}^k P_i\log Q_i$. If $Q_i\in[0,\frac{c_1\log k}{n}]$, we construct a polynomial approximation for the function $P_i\log Q_i$ and further estimate the polynomial function. And if $Q_i\in [\frac{c_1\log k}{n},1]$, we use the bias-corrected augmented plug-in estimator.
We use $N'$ to determine whether to use polynomial estimator or plug-in estimator, and use $N$ to estimate $\sum_{i=1}^k P_i\log Q_i$. Intuitively, if $N'_i$ is large, then $Q_i$ is more likely to be large, and vice versa. Based on the generation scheme, $N$ and $N'$ are independent. Such independence significantly simplifies the analysis.

We let $L=\lfloor c_0\log k\rfloor$, where $c_0$ is a constant to be determined later, and denote the degree-$L$ best polynomial approximation of the function $x\log x$ over the interval $[0,1]$ as $\sum_{j=0}^L a_jx^j$. We further scale the interval $[0,1]$ to $[0,\frac{c_1\log k}{n}]$. Then we have the best polynomial approximation of the function $x\log x$ over the interval $[0,\frac{c_1\log k}{n}] $ as follows:
\begin{flalign}\label{eq:32}
  \gamma_L(x)=\sum_{j=0}^L \frac{a_j n^{j-1}}{(c_1\log k)^{j-1}}x^j -\Big(\log \frac{n}{c_1\log k}\Big)x.
\end{flalign}
Following the result in \cite[Section 7.5.4]{tim63}, the error in approximating $x\log x$ by $\gamma_L(x)$ over the interval $[0,\frac{c_1\log k}{n}]$ can be upper bounded as follows:
\begin{flalign}
  \sup_{x\in [0,\frac{c_1\log k}{n}]}|\gamma_L(x)-x\log x|\lesssim \frac{1}{n\log k}.
\end{flalign}
Therefore, we have $|\gamma_L(0)-0\log 0|\lesssim \frac{1}{n\log k}$, which implies that the zero-degree term in $\gamma_L(x)$ satisfies:
\begin{flalign}
  \frac{a_0c_1\log k}{n}\lesssim \frac{1}{n\log k}.
\end{flalign}
Now, subtracting the zero-degree term from $\gamma_L(x)$ in \eqref{eq:32} yields the following polynomial:
\begin{flalign}
  \mu_L(x)  &\triangleq \gamma_L(x)-\frac{a_0c_1\log k}{n} \nn \\
  &=\sum_{j=1}^L \frac{a_j n^{j-1}}{(c_1\log k)^{j-1}}x^j - \Big(\log \frac{n}{c_1\log k}\Big)x.
\end{flalign}
The error in approximating $x\log x$ by $\mu_L(x)$ over the interval $[0,\frac{c_1\log k}{n}]$ can also be upper bounded by $\frac{1}{n\log k}$, because
\begin{flalign}\label{eq:34}
  &\sup_{x\in [0,\frac{c_1\log k}{n}]}|\mu_L(x)-x\log x|\nn \\
  &=\sup_{x\in [0,\frac{c_1\log k}{n}]}\left|\gamma_L(x)-x\log x-\frac{a_0c_1\log k}{n}\right|\nn\\
  &\leq \sup_{x\in [0,\frac{c_1\log k}{n}]}\left|\gamma_L(x)-x\log x\right|+\left|\frac{a_0c_1\log k}{n}\right|\nn\\
  &\lesssim \frac{1}{n\log k}.
\end{flalign}
The bound in \eqref{eq:34} implies that although $\mu_L(x)$ is not the best polynomial approximation of $x\log x$, the error in the approximation by $\mu_L(x)$ has the same order as that by $\gamma_L(x)$. Compared to $\gamma_L(x)$, there is no zero-degree term in $\mu_L(x)$, and hence $\frac{\mu_L(x)}{x}$ is a valid polynomial approximation of $\log x$. Although the approximation error of $\log x$ using $\frac{\mu_L(x)}{x}$ is unbounded, the error in the approximation of $P_i\log Q_i$ using $P_i\frac{\mu_L(Q_i)}{Q_i}$ can be bounded. More importantly, by the way in which we constructed $\mu_L(x)$, $P_i\frac{\mu_L(Q_i)}{Q_i}$ is a polynomial function of $P_i$ and $Q_i$, for which an unbiased estimator can be constructed.
More specifically, the error in using $P_i\frac{\mu_L(Q_i)}{Q_i}$ to approximate $P_i\log Q_i$ can be bounded as follows:
\begin{flalign}
  \left|P_i\frac{\mu_L(Q_i)}{Q_i}-P_i\log Q_i\right|
  &=\frac{P_i}{Q_i}|\mu_L(Q_i)-Q_i\log Q_i| \nn \\
  &\lesssim \frac{f(k)}{n\log k},
\end{flalign}
for $Q_i\in[0,\frac{c_1\log k}{n}]$.
We further define the factorial moment of $x$ by $(x)_j\triangleq \frac{x!}{(x-j)!}$. If $X\sim$Poi$(\lambda)$, $\mE[(X)_j]=\lambda^j$. Based on this fact, we construct an unbiased estimator for $\frac{\mu_L(Q_i)}{Q_i}$ as follows:
\begin{flalign}\label{eq:g_L}
g_L(N_i)= \sum_{j=1}^L \frac{a_j }{(c_1\log k)^{j-1}} (N_i)_{j-1} -\log \frac{n}{c_1\log k}.
\end{flalign}

We then construct our estimator for $\sum_{i=1}^k P_i\log Q_i $ as follows:
\begin{flalign}\label{eq:D2}
  \hat D_2=\sum_{i=1}^k &\Big(\frac{M_i}{m} g_L( {N_i} ) \mathds{1}_{\{N_i'\leq c_2\log k\}} \nn \\
  &+ \frac{M_i}{m} \big(\log \frac{N_i+1}{n} -\frac{1}{2(N_i+1)}\big)\mathds{1}_{\{N_i'> c_2\log k\}}\Big).
\end{flalign}
{Note that we set $c=1$ for the bias-corrected augmented plug-in estimator used here, and we do not normalize the estimate of $Q_i$. When we normalize with $n+k$ instead of $n$, the estimate differs by at most $\log(1 + k/n) < k/n$, which requires a different bias correction term. For the purpose of convenience, we do not normalize $Q_i$.}

For the term $\sum_{i=1}^k P_i\log P_i $ in $D(P\|Q)$, we use the minimax optimal entropy estimator proposed in \cite{wu2014minimax}. We note that $\gamma_L(x)$ is the best polynomial approximation of the function $x\log x$. And an unbiased estimator of $\gamma_L(x)$ is as follows:
\begin{flalign}
  g_L'(M_i)=\frac{1}{m}\sum_{j=1}^{L'} \frac{a_j }{(c_1'\log k)^{j-1}}(M_i)_j -\Big(\log \frac{m}{c_1'\log k}\Big)M_i.
\end{flalign}
Based on $g_L'(M_i)$, the estimator $\hat D_1$ for $\sum_{i=1}^k P_i\log P_i $ is constructed as follows:
\begin{flalign}\label{eq:D1}
  \hat D_1=\sum_{i=1}^k \Big(& g_L'(M_i)\mathds{1}_{\{M_i'\leq c'_2\log k\}} \nn \\
    &+ \big(\frac{M_i}{m}\log \frac{M_i}{m} - \frac{1}{2m} \big)\mathds{1}_{\{M_i'> c'_2\log k\}}\Big).
\end{flalign}

Combining the estimator $\hat D_1$ in \eqref{eq:D1} for $\sum_{i=1}^k P_i\log P_i $ and the estimator $\hat D_2$ in \eqref{eq:D2} for $\sum_{i=1}^k P_i\log Q_i $, we obtain the estimator $\widetilde {D}_{\mathrm{opt}}$ for KL divergence $D(P\|Q)$ as
\begin{flalign}\label{eq:optimal}
  \widetilde {D}_{\mathrm{opt}}=\hat D_1-\hat D_2.
\end{flalign}
Due to the density ratio constraint, we can show that $0\leq D(P\|Q)\leq \log f(k)$. We therefore construct an estimator $\hat D_{\mathrm{opt}}$ as follows:
\begin{flalign}\label{eq:optimalfk}
  \hat D_{\mathrm{opt}}=\widetilde D_{\mathrm{opt}} \vee 0\wedge \log f(k).
\end{flalign}

The following proposition characterizes an upper bound on the worse-case quadratic risk of $\hat D_{\mathrm{opt}}$.
\begin{proposition}\label{prop:upperbound}
  If $\log ^2n\lesssim k^{1-\epsilon}$, where $\epsilon$ is any positive constant, and $\log m\leq C \log k$ for some constant $C$, then there exists $c_0$, $c_1$ and $c_2$ depending on $C$ only, such that
  \begin{flalign}
    \widetilde{R}(\hat D_{\mathrm{opt}},&k,m,n,f(k)) \nn \\
    \lesssim & \left(\frac{k}{m\log k}+\frac{kf(k)}{n\log k}\right)^2 +\frac{\log^2 f(k)}{m}+ \frac{f(k)}{n}.
  \end{flalign}
\end{proposition}

\begin{proof}
See Appendix \ref{app:prop4}.
\end{proof}
It is clear that the upper bound in Proposition \ref{prop:upperbound} matches the lower bound in Proposition \ref{prop:GMLBentropy} (up to a constant factor), and thus the constructed estimator is minimax optimal, and the minimax risk in Theorem \ref{thm:minimax} is established.

\section{Numerical Experiments}
In this section, we provide numerical results to demonstrate the performance of our estimators, and compare our augmented plug-in estimator, minimax optimal estimator with a number of other KL divergence estimators. \footnote{The implementation of our estimator is available at {https://github.com/buyuheng/Minimax-KL-divergence-estimator}.}

\begin{figure}[!htp]
\centering
\subfigure[$P =\left(\frac{1}{k},\ \frac{1}{k},\ \cdots,\ \frac{1}{k}\right), \quad  Q =\left(\frac{1}{kf(k)},\ \cdots,\  \frac{1}{kf(k)},\ 1-\frac{k-1}{kf(k)}\right)$.]{
\includegraphics[width=9cm]{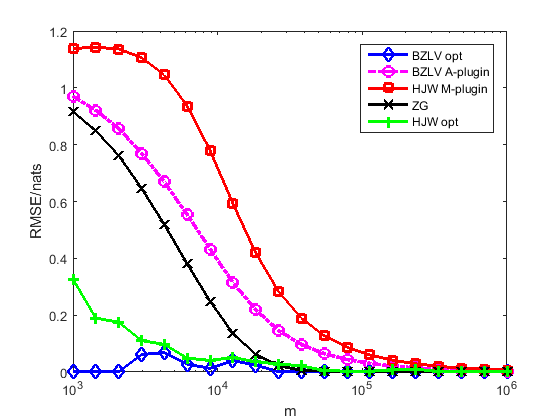}
}
\subfigure[$P=\mathrm{Zipf}(1), Q=\mathrm{Zipf}(0.8)$.]{
\includegraphics[width=9cm]{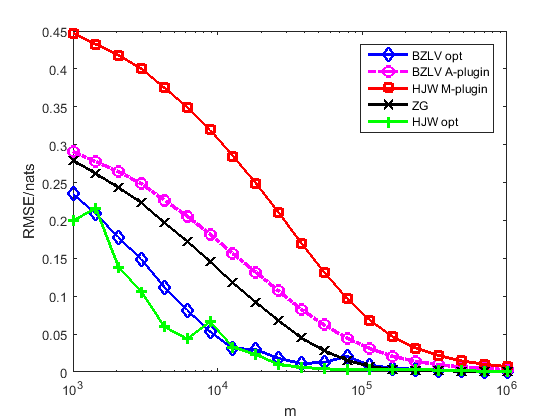}}
\subfigure[$P=\mathrm{Zipf}(1), Q=\mathrm{Zipf}(0.6)$.]{
\includegraphics[width=9cm]{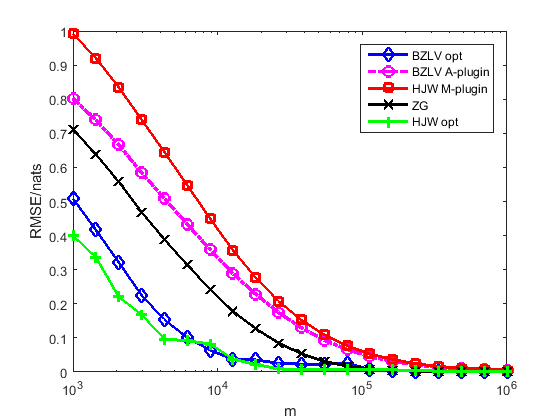}}
\caption{Comparison of five estimators under traditional setting with $k=10^4$, $m$ ranging from $10^3$ to $10^6$ and $n\asymp f(k)m$.}
\label{Fig1}
\end{figure}

\begin{figure}[!htp]
\centering
\subfigure[$P =\left(\frac{1}{k},\ \frac{1}{k},\ \cdots,\ \frac{1}{k}\right), \quad  Q =\left(\frac{1}{kf(k)},\ \cdots,\  \frac{1}{kf(k)},\ 1-\frac{k-1}{kf(k)}\right)$.]{
\includegraphics[width=9cm]{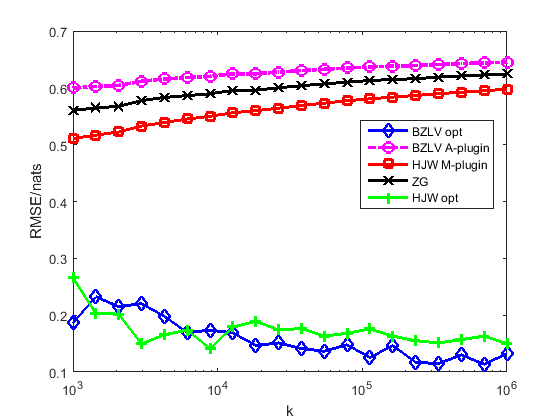}}
\subfigure[$P=\mathrm{Zipf}(1), Q=\mathrm{Zipf}(0.8)$.]{
\includegraphics[width=9cm]{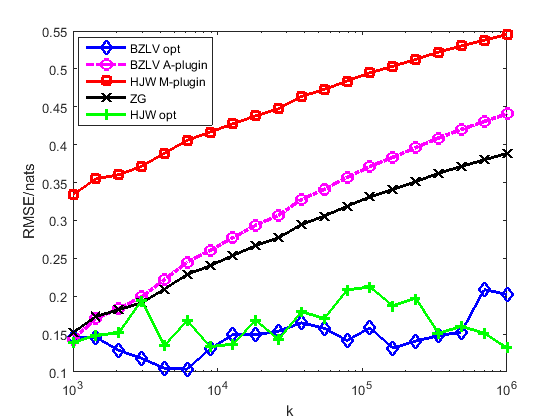}}
\subfigure[$P=\mathrm{Zipf}(1), Q=\mathrm{Zipf}(0.6)$.]{
\includegraphics[width=9cm]{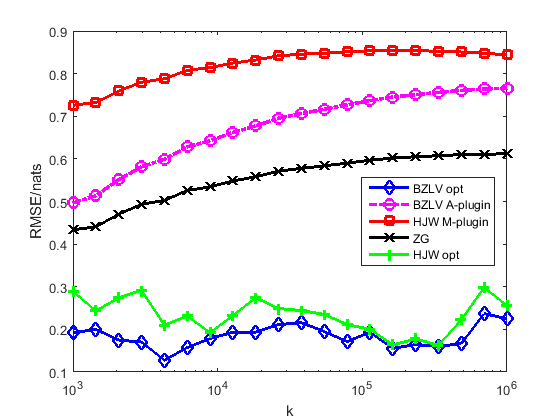}}
\caption{Comparison of five estimators under large-alphabet setting with $k$ ranging from $10^3$ to $10^6$, $m=\frac{2k}{\log k}$ and $n=\frac{kf(k)}{\log k}$. }
\label{Fig2}
\end{figure}

To implement the minimax optimal estimator, we first compute the coefficients of the best polynomial approximation by applying the Remez algorithm \cite{petrushev2011rational}. In our experiments, we replace the $N_i'$ and $M_i'$ in \eqref{eq:D2} and \eqref{eq:D1} with $N_i$ and $M_i$, which means we use all the samples for both selecting estimators (polynomial or plug-in) and estimation. We choose the constants $c_0$, $c_1$ and $c_2$ following ideas in \cite{jiao2015minimax}. More specifically, we set $c_0=1.2$, $c_2 \in [0.05,0.2]$ and $c_1=2c_2$.

We compare the performance of the following five estimators: 1) our augmented plug-in estimator (BZLV A-plugin) in \eqref{eq:pluginest}, with $c=1$; 2) our minimax optimal estimator (BZLV opt) in \eqref{eq:optimal}; 3) Han, Jiao and Weissman's modified plug-in estimator (HJW M-plugin) in \cite{han2016} ; 4) Han, Jiao and Weissman's minimax optimal estimator (HJW opt) \cite{han2016}; 5) Zhang and Grabchak's estimator (ZG) in \cite{zhang2014nonparametric} which is constructed for fix-alphabet size setting.

{Note that the term $f(k)$ in the minimax optimal estimator \eqref{eq:optimalfk} serves to keep the estimate between $0$ and $\log f(k)$, which benefits the theoretical analysis. The BZLV opt estimator in the simulation refers to the one defined in \eqref{eq:optimal}, which does not require the knowledge of $f(k)$, but only the knowledge of $k$.}

We first compare the performance of the five estimators under the traditional setting in which we set $k=10^4$ and let $m$ and $n$ change. We choose two types of distributions $(P,Q)$. The first type is given by $ P =\left(\frac{1}{k},\ \frac{1}{k},\ \cdots,\ \frac{1}{k}\right), \quad  Q =\big(\frac{1}{kf(k)},\ \cdots,\  \frac{1}{kf(k)},\ 1-\frac{k-1}{kf(k)}\big)$, where $f(k)=5$. For this pair of $(P,Q)$, the density ratio is $f(k)$ for all but one of the bins, which is in a sense the worst case for the KL divergence estimation problem. We let $m$ range from $10^3$ to $10^6$ and set $n=3f(k)m$. The second type is given by  $(P,Q)=(\mathrm{Zipf}(1),\mathrm{Zipf}(0.8))$ and $(P,Q)=(\mathrm{Zipf}(1),\mathrm{Zipf}(0.6))$. The Zipf distribution is a discrete distribution that is commonly used in linguistics, insurance, and the modeling of rare events. If $P= \mathrm{Zipf}(\alpha)$, then $P_i = \frac{i^{-\alpha}}{\sum_{j=1}^kj^{-\alpha}}$, for $i \in [k]$. We let $m$ range from $10^3$ to $10^6$ and set $n=0.5f(k)m$, where $f(k)$ is computed for these two pairs of Zipf distributions, respectively.

In Fig.~\ref{Fig1}, we plot the root mean square errors (RMSE) of the five estimators as a function of the sample size $m$ for these three pairs of distributions. It is clear from the figure that our minimax optimal estimator (BZLV opt) and the HJW minimax optimal estimator (HJW opt) outperform the other three approaches. Such a performance improvement is significant especially when the sample size is small. Furthermore, our augmented plug-in estimator (BZLV A-plugin) has a much better performance than the HJW modified plug-in estimator (HJW M-plugin), because the bias of estimating $\sum_{i=1}^kP_i\log P_i$ and the bias of estimating $\sum_{i=1}^kP_i\log Q_i$ may cancel each other out by the design of our augmented plug-in estimator. Furthermore, the RMSEs of all the five estimators converge to zero when the number of samples are sufficiently large.





We next compare the performance of the five estimators under the large-alphabet setting, in which we let $k$ range from $10^3$ to $10^6$, and set $m=\frac{2k}{\log k}$ and $n=\frac{kf(k)}{\log k}$. We use the same three pairs of distributions as in the previous setting. In Fig.~\ref{Fig2}, we plot the RMSEs of the five estimators as a function of $k$. It is clear from the figure that our minimax optimal estimator (BZLV opt) and the HJW minimax optimal estimator (HJW opt) have very small estimation errors, which is consistent with our theoretical results of the minimax risk bound. However, the RMSEs of the other three approaches increase with $k$, which implies that $m=\frac{2k}{\log k}$, $n=\frac{kf(k)}{\log k}$ are insufficient for those estimators.

{We also observe that our minimax optimal estimator (BZLV opt) and the HJW minimax optimal estimator (HJW opt) achieve almost equally good performance in all experiments. This is to be expected, because the difference between the two estimators is mainly captured by the threshold between the use of the polynomial approximation and the use of the plug-in estimator, i.e., BZLV opt exploits the knowledge of $k$ to set the threshold, whereas HJW opt sets the threshold adaptively without using the information about $k$. In fact, for typical sample sizes in the experiments, such thresholds in two estimators are not very different.}

\section{Conclusion}

In this paper, we studied the estimation of KL divergence between large-alphabet distributions.
We showed that there exists no consistent estimator for KL divergence under the worst-case quadratic risk over all distribution pairs. We then studied a more practical set of distribution pairs with bounded density ratio. We proposed an augmented plug-in estimator, and characterized the worst-case quadratic risk of such an estimator. We further designed a minimax optimal estimator by employing a polynomial approximation along with the plug-in approach, and established the optimal minimax rate. We anticipate that the designed KL divergence estimator can be used in various application contexts including classification, anomaly detection, community clustering, and nonparametric hypothesis testing.

\section*{Acknowledgement}
{The authors are grateful to the Associate Editor and the anonymous reviewers for their helpful comments. In particular, we thank one of the reviewers for pointing out that the $\frac{\log^2k}{m}$ term in the original variance bound for the augmented plug-in estimator can be removed using the techniques in \cite{han2016}. We also thank Yanjun Han at Stanford University for pointing out the same issue to us. }

\appendices

\section{Proof of Proposition \ref{prop:upper-plug}}
\label{app:prop1}

The quadratic risk can be decomposed into the sum of square of the bias and the variance as follows:
\begin{align*}
   &\mathbb{E}\big[(\hat {D}_{\mathrm{A-plug-in}}(M,N)-D(P\|Q))^2\big] \nn \\
   & =  \Big(\mathbb{E}\big[\hat {D}_{\mathrm{A-plug-in}}(M,N)-D(P\|Q)\big]\Big)^2 \nn \\
   &\quad + \mathrm{Var}\big[\hat {D}_{\mathrm{A-plug-in}}(M,N)\big].
\end{align*}

We bound the bias and the variance in the following two subsections, respectively.
\subsection{Bounds on the Bias}
The bias of the augmented plug-in estimator can be written as
\begin{align}\label{eq:bias_pl}
  &\left|\mathbb{E} \big[ \hat{D}_{\mathrm{A-plug-in}} (M,N) - D(P\|Q) \big] \right|\nn\\
  &=\left|\mathbb{E} \left[\sum_{i=1}^k \Big(\frac{M_i}{m} \log\frac{M_i/m}{(N_i+c)/(n+kc)}-P_i\log \frac{P_i}{Q_i} \Big)\right]\right| \nn\\
  &\le \left|\mathbb{E} \left[ \sum_{i=1}^k  \Big( \frac{M_i}{m} \log \frac{M_i}{m} - P_i \log P_i \Big) \right]\right| \nn \\
  &\quad + \left|\mathbb{E} \left[ \sum_{i=1}^k \Big( P_i \log \frac{(n+kc)Q_i}{N_i+c} \Big)  \right] \right|.
\end{align}

The first term in \eqref{eq:bias_pl} is the bias of the plug-in estimator for entropy estimation, which can be bounded as in \cite{paninski2003estimation}:
\begin{align}\label{eq:bias_entropy_upper}
   &\left|\mathbb{E} \left[ \sum_{i=1}^k \Big( \frac{M_i}{m} \log \frac{M_i}{m} - P_i \log P_i \Big) \right]\right| \nn \\
   &\le \log\Big(1 + \frac{k-1}{m}\Big)  < \frac{k}{m}.
\end{align}

Next, we lower bound the second term in \eqref{eq:bias_pl} as follows:
\begin{align}\label{eq:plogq_lower}
  &\mathbb{E} \left[ \sum_{i=1}^k P_i \log \frac{(n+kc)Q_i}{N_i+c}  \right] \nn \\
  &= -  \sum_{i=1}^k P_i \mathbb{E} \left[\log\Big(1+ \frac{N_i+c-(n+kc)Q_i} {(n+kc)Q_i}\Big ) \right] \nn \\
  &\overset{(a)}{\ge} - \sum_{i=1}^k P_i  \mathbb{E} \left[\frac{N_i+c-(n+kc)Q_i} {(n+kc)Q_i} \right]\nn\\
  &= \sum_{i=1}^k P_i \frac{kc} {n+kc}-\sum_{i=1}^k P_i \frac{c} {(n+kc)Q_i}\nn \\
  &\ge - \frac{ckf(k)}{n},
\end{align}
where (a) is due to the fact that $\log(1+x)\le x$.

Then, we need to find an upper bound for the second term in \eqref{eq:bias_pl}, which can be rewritten as,
\begin{align}\label{eq:bias_decomp}
  &\mathbb{E} \left[ \sum_{i=1}^k P_i \log \frac{(n+kc)Q_i}{N_i+c}  \right] \nn \\
  &=\underbrace{\sum_{i=1}^k P_i \left( \log\Big(Q_i+\frac{c}{n}\Big) - \mathbb{E}\Big[ \log \frac{N_i+c}{n}\Big]\right)}_{A_1} \nn \\
  &\quad+ \underbrace{\sum_{i=1}^k P_i \log \frac{(n+kc)Q_i}{nQ_i+c}}_{A_2}.
\end{align}

The term $A_1$ can be upper bounded as follows:
\begin{align}\label{eq:plogq_upper1}
  A_1 =& \sum_{i=1}^k P_i \left( \log(Q_i+\frac{c}{n}) - \mathbb{E} \log (\frac{N_i+c}{n})\right) \nn \\
  \overset{(a)}{\le}& \sum_{i=1}^k f(k) \bigg( \Big( Q_i+\frac{c}{n}\Big) \log\Big(Q_i+\frac{c}{n}\Big) \nn \\
   &\quad\quad - \Big( Q_i+\frac{c}{n}\Big)\mathbb{E} \Big[\log \frac{N_i+c}{n}\Big] \bigg) \nn\\
  \le & f(k) \sum_{i=1}^k \bigg(  \Big( Q_i+\frac{c}{n}\Big) \log\Big(Q_i+\frac{c}{n}\Big)  \nn \\
  &- \mathbb{E} \Big[\frac{N_i+c}{n} \log \frac{N_i+c}{n}\Big]  +\left|\mathbb{E} \Big[\Big( \frac{N_i}{n}-Q_i\Big) \log \frac{N_i+c}{n}\Big] \right|\bigg)\nn\\
  \overset{(b)}{=}& f(k) \sum_{i=1}^k \bigg(  \mathbb{E} \Big[\frac{N_i+c}{n} \log \frac{nQ_i+c}{N_i+c}\Big]  \nn \\
  &+\left|\mathbb{E} \Big[\Big( \frac{N_i}{n}-Q_i\Big) \log \frac{N_i+c}{n}\Big] \right|\bigg)\nn \\
  \le& f(k)  \frac{n+kc}{n}\mathbb{E} \Big[\sum_{i=1}^k \frac{N_i+c}{n+kc} \log \frac{(N_i+c)/(n+kc)}{(nQ_i+c)/(n+kc)}\Big]  \nn\\
  &+ f(k) \sum_{i=1}^k \left|\mathbb{E} \Big[\Big( \frac{N_i}{n}-Q_i\Big) \log \frac{N_i+c}{n}\Big] \right|,
\end{align}
where $(a)$ follows from the assumption $P_i\le f(k)Q_i$ and $ \mathbb{E}\big[ \log \frac{N_i+c}{n}\big] \le \log\big(Q_i+\frac{c}{n}\big)$, and $(b)$ is due to the fact $\mathbb{E}[N_i]=nQ_i$.

The first term in \eqref{eq:plogq_upper1} is the expectation of the KL divergence between the smoothed empirical distribution and the true distribution. It was shown in \cite{paninski2003estimation} that the KL divergence between two distributions $p$ and $q$ can be upper bounded by the $\chi^2$ divergence, i.e.,
\begin{equation}
D(p\|q) \le \log(1+\chi^2(p,q)) \le \chi^2(p,q) ,
\end{equation}
where $\chi^2(p,q)$ is defined as  $\chi^2(p,q)\triangleq \sum_{i=1}^k \frac{(p_i-q_i)^2}{q_i}$. Applying this result to the first term in \eqref{eq:plogq_upper1}, we obtain
\begin{align}\label{eq:plogq_upper1term1}
&\mathbb{E} \Big[\sum_{i=1}^k \frac{N_i+c}{n+kc} \log \frac{(N_i+c)/(n+kc)}{(nQ_i+c)/(n+kc)}\Big]\nn \\
&\le  \sum_{i=1}^k \frac{\mathbb{E}\big[(N_i-nQ_i)^2\big]}{(n+kc)(nQ_i+c)} \nn \\
& =  \sum_{i=1}^k \frac{nQ_i(1-Q_i)}{(n+kc)(nQ_i+c)}  \le \frac{k-1}{n+kc}.
\end{align}

For the second term in \eqref{eq:plogq_upper1}, by the fact $\mathbb{E}[N_i]=nQ_i$ and the Cauchy-Schwartz inequality, we obtain
\begin{align}\label{eq:plogq_upper1term2}
  &\left|\mathbb{E} \Big[\big( \frac{N_i}{n}-Q_i\big) \log \frac{N_i+c}{n}\Big] \right|^2 \nn \\
&= \left|\mathbb{E} \bigg[\big( \frac{N_i}{n}-Q_i\big)\Big( \log \frac{N_i+c}{n}-\mathbb{E}\Big[\log \frac{N_i+c}{n}\Big]\Big)\bigg] \right|^2 \nn \\
&\le \frac{Q_i}{n}\mathrm{Var}\Big[ \log \frac{N_i+c}{n}\Big]
\lesssim \frac{1}{n^2}.
\end{align}
The last step above follows because
\begin{align}\label{eq:var_plogq2}
  &\mathrm{Var}\Big[ \log \frac{N_i+c}{n}\Big] \nn \\
  &\le \mathbb{E}\left[  \Big(\log \frac{N_i+c}{nQ_i+c}\Big) ^2\right] \nn \\
  & =\mathbb{E}\left[  \Big(\log \frac{N_i+c}{nQ_i+c}\Big) ^2\mathds{1}_{\{N_i\ge \frac{nQ_i}{2} \}} \right] \nn \\
  &\qquad +\mathbb{E}\left[  \Big(\log \frac{N_i+c}{nQ_i+c}\Big) ^2\mathds{1}_{\{N_i < \frac{nQ_i}{2} \}} \right] \nn\\
  &\overset{(a)}{\le} \sup_{\xi \ge \frac{nQ_i}{2}} \frac{1}{(\xi+c)^2} \mathbb{E}\left[  (N_i-nQ_i)^2 \right]\nn \\
  &\qquad +\sup_{\xi \ge 0} \frac{1}{(\xi+c)^2} \mathbb{E}\left[  (N_i-nQ_i)^2\mathds{1}_{\{N_i< \frac{nQ_i}{2} \}} \right] \nn\\
  &\le   \frac{4}{n^2Q_i^2} nQ_i(1-Q_i) + \frac{n^2Q_i^2}{c^2}\mathbb{P}\Big(N_i< \frac{nQ_i}{2}\Big)\nn\\
  &\overset{(b)}{\le} \frac{4}{nQ_i}  + \frac{n^2Q_i^2}{c^2}e^{-nQ_i/8} \nn \\
  &\overset{(c)}{\lesssim} \frac{1}{nQ_i}.
\end{align}
where $(a)$ is due to the mean value theorem; $(b)$ uses the Chernoff bound of Binomial distribution; $(c)$ is due to the fact that $x^3e^{-\frac{x}{8}}$ is upper bounded by some constant for $x>0$.

Thus, substituting the bounds in \eqref{eq:plogq_upper1term1} and \eqref{eq:plogq_upper1term2} into \eqref{eq:plogq_upper1}, we obtain
\begin{equation}\label{eq:plogq_upper1_final}
  A_1 \lesssim \frac{kf(k)}{n}.
\end{equation}


Furthermore, the term $A_2$ can be upper bounded by,
\begin{align}\label{eq:plogq_upper2}
  A_2&= \sum_{i=1}^k P_i \log \Big( 1+\frac{(kQ_i-1)c}{nQ_i+c}\Big)\nn \\
   &\le \sum_{i=1}^k P_i \frac{(kQ_i)c}{nQ_i+c}\nn \\
   &\le \frac{ckf(k)}{n}.
\end{align}

Combining  \eqref{eq:plogq_lower}, \eqref{eq:plogq_upper1_final} and \eqref{eq:plogq_upper2}, we obtain the following upper bound for the second term in the bias,
\begin{equation}
\left|\mathbb{E} \left( \sum_{i=1}^k P_i \log \frac{(n+kc)Q_i}{N_i+c}  \right) \right| \lesssim \frac{kf(k)}{n}.
\end{equation}
Hence,
\begin{align}\label{eq:pluginbias}
  &\left|\mathbb{E} \bigg( \hat{D}_{\mathrm{A-plug-in}} (M,N) - D(P\|Q) \bigg)\right| \nn \\
  &\lesssim \frac{k}{m} +  \frac{kf(k)}{n} .
\end{align}

\subsection{Bounds on the Variance}
The variance of the augmented plug-in estimator can be upper bounded by

\begin{align}
 &\mathrm{Var}\big[\hat {D}_{\mathrm{A-plug-in}}(M,N)\big] \nn \\
&= \sum_{i=1}^k \mathbb{E}\bigg[\Big( \frac{M_i}{m} \log\frac{M_i/m}{(N_i+c)/(n+kc)} \nn \\
&\qquad\qquad- \mathbb{E} \Big[ \frac{M_i}{m} \log\frac{M_i/m}{(N_i+c)/(n+kc)}\Big] \Big)^2\bigg]\nn \\
&\le  \sum_{i=1}^k \mathbb{E}\bigg[ \Big(  \frac{M_i}{m} \log\frac{M_i/m}{(N_i+c)/(n+kc)} \nn \\
 &\qquad\qquad-  P_i \log \frac{P_i}{(nQ_i+c)/(n+kc)} \Big)^2\bigg]\nn \\
&\le 3\sum_{i=1}^k \underbrace{\mathbb{E}\left[ \Big(  \frac{M_i}{m} \big(\log\frac{M_i}{m} - \log P_i\big)\Big)^2\right]}_{V^{(1)}_i}\nn \\
& \quad + 3\sum_{i=1}^k\underbrace{\mathbb{E}\left[ \Big( \frac{M_i}{m}\big(\log \frac{N_i+c}{n+kc}-\log \frac{nQ_i+c}{n+kc}\big) \Big)^2\right]}_{V^{(2)}_i}\nn\\
& \quad + 3\sum_{i=1}^k\underbrace{\mathbb{E}\left[ \Big( \big(\frac{M_i}{m}-P_i\big) \log \frac{P_i}{(nQ_i+c)/(n+kc)}\Big)^2\right]}_{V^{(3)}_i}.
\end{align}


We first split $V^{(1)}_i$ into two parts,
\begin{align}
  V^{(1)}_i =& \mathbb{E}\left[ \bigg(  \frac{M_i}{m} \Big(\log\frac{M_i}{m} - \log P_i\Big)\bigg)^2 \mathds{1}_{\{M_i\le m P_i\}}\right] \nn \\
  &+\mathbb{E}\left[ \bigg(  \frac{M_i}{m} \Big(\log\frac{M_i}{m} - \log P_i\Big)\bigg)^2 \mathds{1}_{\{M_i > mP_i\}}\right],
\end{align}
where the first term can be upper bounded by using the mean value theorem,
\begin{align}\label{eq:var_plogp1}
  &\mathbb{E}\left[ \bigg(  \frac{M_i}{m} \Big(\log\frac{M_i}{m} - \log P_i\Big)\bigg)^2 \mathds{1}_{\{M_i\le m P_i\}}\right] \nn \\
  &\le \mathbb{E}\left[\frac{M_i^2}{m^2}\sup_{\xi\ge M_i/m}\frac{1}{\xi^2}\Big(\frac{M_i}{m}-P_i\Big)^2 \mathds{1}_{\{M_i\le m P_i\}} \right] \nn \\
  &\le \mathbb{E}\left[\Big(\frac{M_i}{m}-P_i\Big)^2\right] = \frac{P_i(1-P_i)}{m}.
\end{align}
For the second part, applying the mean value theorem, we have
\begin{align}
  &\mathbb{E}\left[ \bigg(  \frac{M_i}{m} \Big(\log\frac{M_i}{m} - \log P_i\Big)\bigg)^2 \mathds{1}_{\{M_i > m P_i\}}\right] \nn \\
 &\le \mathbb{E}\left[\frac{M_i^2}{m^2}\sup_{\xi\ge P_i}\frac{1}{\xi^2}\Big(\frac{M_i}{m}-P_i\Big)^2 \mathds{1}_{\{M_i> m P_i\}} \right] \nn \\
 &\le\frac{1}{m^2P_i^2}\mathbb{E}\left[\Big(M_i\big(\frac{M_i}{m}-P_i\big)\Big)^2\right] \nn \\
 &= \frac{P_i(1-P_i)}{m}-\frac{7P_i(1-P_i)}{m^2}+5\frac{1-P_i}{m^2}+\frac{12P_i-7}{m^3}+\frac{1}{m^3P_i}\nn \\
 &\le \frac{P_i}{m}+\frac{5}{m^2}+\frac{12P_i}{m^3}+\frac{1}{m^3P_i}.
\end{align}
If $P_i \ge \frac{1}{m}$, we have
\begin{align}\label{eq:var_plogp2}
  &\mathbb{E}\left[ \bigg(  \frac{M_i}{m} \Big(\log\frac{M_i}{m} - \log P_i\Big)\bigg)^2 \mathds{1}_{\{M_i > m P_i, m P_i \ge 1 \}}\right] \nn \\
  &\lesssim  \frac{P_i}{m}+\frac{1}{m^2}.
\end{align}
If $P_i < \frac{1}{m}$, a more careful bound can be derived as follows:
\begin{flalign}
  &\mathbb{E}\left[ \bigg(  \frac{M_i}{m} \Big(\log\frac{M_i}{m} - \log P_i\Big)\bigg)^2 \mathds{1}_{\{M_i > m P_i, m P_i < 1 \}}\right]\nn \\
  &= \sum_{j=1}^m {m \choose j}(1-P_i)^{m-j}P_i^j \frac{j^2}{m^2}\log^2\frac{j}{mP_i} \cdot \mathds{1}_{\{P_i<\frac{1}{m}\}}\nn \\
  &= \sum_{j=1}^m \frac{(mP_i)^j}{j!}\frac{m!(1-P_i)^{m-j}}{m^j(m-j)!} \frac{j^2}{m^2}\log^2\frac{j}{mP_i}\cdot\mathds{1}_{\{P_i<\frac{1}{m}\}}\nn \\
  &\le  \sum_{j=1}^m \frac{(mP_i)^j}{j!}\frac{j^2}{m^2}\log^2\frac{j}{mP_i}\cdot\mathds{1}_{\{P_i<\frac{1}{m}\}},
\end{flalign}
where the last step follows from the facts that $\frac{m!}{m^j(m-j)!}\le1$ and $(1-P_i)^{m-j}\le1$. Note that $mP_i<1$,
\begin{align}
 &\sup_{mP_i\le 1}\frac{(mP_i)^j}{j!}\frac{j^2}{m^2}\log^2\frac{j}{mP_i} \nn \\
 &\le 2\sup_{mP_i\le 1}\frac{(mP_i)^j}{j!}\frac{j^2}{m^2}(\log^2j+\log^2mP_i)\nn \\
 &\le \frac{2\log^2j}{j!}\frac{j^2}{m^2} + 2\sup_{mP_i\le 1}\frac{mP_i\log^2mP_i}{j!}\frac{j^2}{m^2} \nn \\
 &\le \frac{2(\log^2j+1)}{j!}\frac{j^2}{m^2},
\end{align}
where we use the fact that $x^j\log^2 x < x \log^2 x <1$, for $x \in (0,1)$.
Hence,
\begin{align}\label{eq:var_plogp3}
& \mathbb{E}\left[ \bigg(  \frac{M_i}{m} \Big(\log\frac{M_i}{m} - \log P_i\Big)\bigg)^2 \mathds{1}_{\{M_i > m P_i, m P_i < 1 \}}\right] \nn \\
& \le \frac{2}{m^2}\sum_{j=1}^\infty \frac{j^2(\log^2j+1)}{j!}<\frac{22}{m^2}.
\end{align}
The last step above follows because the infinite sum converges to
\begin{equation}
  \sum_{j=1}^\infty \frac{j^2(\log^2j+1)}{j!}\approx 10.24<11.
\end{equation}
Combining \eqref{eq:var_plogp1}, \eqref{eq:var_plogp2} and \eqref{eq:var_plogp3}, we upper bound $V^{(1)}_i$ as
\begin{equation}\label{eq:var_plogp}
  V^{(1)}_i\lesssim \frac{P_i}{m}+\frac{1}{m^2}.
\end{equation}

We next proceed to bound $V^{(2)}_i$, which can be written as
\begin{align}\label{eq:var_plogq1}
  V^{(2)}_i  
  =& \mathbb{E}\Big[\frac{M_i^2}{m^2}\Big]\mathbb{E}\left[  \Big(\log \frac{N_i+c}{nQ_i+c}\Big)^2\right]  \nn \\
  =& \left(\frac{P_i(1-P_i)}{m} +P_i^2 \right)\mathbb{E}\left[  \Big(\log \frac{N_i+c}{nQ_i+c}\Big) ^2\right].
\end{align}
Using the result in \eqref{eq:var_plogq2}, $V^{(2)}_i$ can be upper bounded by
\begin{equation}\label{eq:var_plogq}
  V^{(2)}_i \lesssim \left(\frac{P_i}{m} +P_i^2 \right) \frac{1}{nQ_i} \lesssim \frac{f(k)}{n}\Big(\frac{1}{m} + P_i \Big).
\end{equation}

We further derive the following bound on $V^{(3)}_i$
\begin{align}\label{eq:var_approx_1}
  V^{(3)}_i &\le \frac{P_i}{m} \log^2 \frac{P_i(n+kc)}{nQ_i+c}\nn \\
  &=  \frac{P_i}{m}\left(\log \frac{P_i}{Q_i}+ \log \frac{Q_i(n+kc)}{nQ_i+c}\right)^2 \nn \\
  &\le \frac{2P_i}{m}\log^2 \frac{P_i}{Q_i}+ \frac{2P_i}{m} \log^2 \frac{Q_i(n+kc)}{nQ_i+c}.
\end{align}
The first term in \eqref{eq:var_approx_1} can be upper bounded by
\begin{align}\label{eq:var_approx_2}
  &\frac{2P_i}{m}\log ^2\frac{P_i}{Q_i} \nn \\
&=\frac{2}{m}\Big( {P_i}\log ^2\frac{P_i}{Q_i} \mathds 1_{\{\frac{1}{f(k)}\leq \frac{P_i}{Q_i}\leq f(k)\}} \nn \\
&\qquad \qquad +Q_i\frac{P_i}{Q_i}\log ^2\frac{P_i}{Q_i}\mathds 1_{\{\frac{P_i}{Q_i}\leq\frac{1}{f(k)}\}}\Big) \nn \\
&\lesssim \frac{P_i\log^2f(k)}{m}+\frac{Q_i}{m},
\end{align}
where the  last inequality follows because $x\log^2 x$ is bounded by a constant on the interval $[0,1/f(k)]$.

We bound the second term in \eqref{eq:var_approx_1} by splitting it into two parts,
\begin{align}\label{eq:var_approx_decomp}
  &\frac{2P_i}{m} \log^2 \frac{Q_i(n+kc)}{nQ_i+c} \nn \\
  &=\frac{2P_i}{m} \left(\log^2 \frac{Q_i(n+kc)}{nQ_i+c}\right)\mathds{1}_{\{Q_i > \frac{1}{k}\}} \nn \\
   &\quad + \frac{2P_i}{m}\left(\log^2 \frac{nQ_i+c}{Q_i(n+kc)}\right) \mathds{1}_{\{Q_i\le \frac{1}{k}\}}.
\end{align}
The first term in \eqref{eq:var_approx_decomp} can be bounded as follows,
\begin{align}\label{eq:var_approx_3}
  &\frac{2P_i}{m} \left(\log^2 \frac{Q_i(n+kc)}{nQ_i+c}\right)\mathds{1}_{\{Q_i > \frac{1}{k}\}} \nn \\
  &= \frac{2P_i}{m} \log^2 \left( 1+\frac{(kQ_i-1)c}{nQ_i+c}\right)\mathds{1}_{\{Q_i > \frac{1}{k}\}} \nn \\
  &\le \frac{2P_i}{m}\left(\frac{kQ_ic}{nQ_i+c}\right)^2 \lesssim \frac{k^2P_i}{mn^2}.
\end{align}
The second term in \eqref{eq:var_approx_decomp} requires more delicate analysis. We first bound it as
\begin{align}
&\frac{2P_i}{m}\left(\log^2 \frac{nQ_i+c}{Q_i(n+kc)}\right) \mathds{1}_{\{Q_i\le \frac{1}{k}\}} \nn \\
& \le  \frac{2f(k)}{m}Q_i\left(\log^2 \frac{n+c/Q_i}{n+kc}\right) \mathds{1}_{\{Q_i\le \frac{1}{k}\}}.
\end{align}
Consider the function $h(q) = q  \log^2 \frac{n+c/q}{n+kc}$, for $q \in [0,\frac{1}{k}]$. It can be shown that the maximizer $q^*$ of $h(q)$ on the interval $[0,\frac{1}{k}]$ satisfies 
\begin{equation*}
  \log \frac{n+c/q^*}{n+kc} = \frac{2c}{nq^*+c}.
\end{equation*}
Then,
\begin{align}\label{eq:var_approx_4}
&\frac{2P_i}{m}\left(\log^2 \frac{nQ_i+c}{Q_i(n+kc)}\right) \mathds{1}_{\{Q_i\le \frac{1}{k}\}} \nn \\
& \le \frac{2f(k)}{m}q^*  \frac{4c^2}{(nq^*+c)^2} \lesssim \frac{f(k)}{mn},
\end{align}
where the  last inequality follows because $\frac{q^*}{(nq^*+c)^2} \le \frac{1}{4cn}$, for $q^*\in[0,\frac{1}{k}]$.

Combining \eqref{eq:var_approx_2}, \eqref{eq:var_approx_3} and  \eqref{eq:var_approx_4}, $V^{(3)}_i$  is upper bounded by
\begin{equation}\label{eq:var_approx}
  V^{(3)}_i \lesssim \frac{P_i\log^2f(k)}{m}+\frac{Q_i}{m}+ \frac{k^2P_i}{mn^2}+\frac{f(k)}{mn}.
\end{equation}

%
A combination of the upper bounds on $V^{(1)}_i$, $V^{(2)}_i$ and $V^{(3)}_i$ yields
\begin{align}\label{eq:pluginvariance}
 &\mathrm{Var}\big[\hat {D}_{\mathrm{A-plug-in}}(M,N)\big] \nn \\
&\le 3\sum_{i=1}^k V^{(1)}_i+ 3\sum_{i=1}^kV^{(2)}_i+ 3\sum_{i=1}^kV^{(3)}_i\nn\\
&\lesssim  \sum_{i=1}^k \bigg( \frac{P_i}{m}+\frac{1}{m^2} +\frac{f(k)}{n}\Big(\frac{1}{m} + P_i \Big) \nn \\
 & \qquad \qquad +\frac{P_i\log^2f(k)}{m}+\frac{Q_i}{m}+ \frac{k^2P_i}{mn^2}+\frac{f(k)}{mn}\bigg)\nn \\
&\lesssim  \frac{k^2}{m^2}+\frac{kf(k)}{mn}+\frac{k^2}{mn^2}+\frac{f(k)}{n}+\frac{\log^2f(k)}{m}.
\end{align}
Note that the terms $\frac{kf(k)}{nm}$ and $\frac{k^2}{mn^2}$ in the variance can be further upper bounded as follows
\begin{align}\label{eq:49}
  \frac{kf(k)}{mn} \le \frac{k}{m}\frac{kf(k)}{n} &\le \left(\frac{k}{m} +\frac{kf(k)}{n} \right)^2, \nn \\
  \frac{k^2}{mn^2} \le \frac{k}{m}\frac{1}{n}\frac{kf(k)}{n} &\le \left(\frac{k}{m} +\frac{kf(k)}{n}\right)^2.
\end{align}

Combining \eqref{eq:pluginbias}, \eqref{eq:pluginvariance} and \eqref{eq:49}, we obtain the following upper bound on the worst-case quadratic risk for the augmented plug-in estimator:
\begin{align}
  R(\hat{D}_{\mathrm{A-plug-in}},&k,m,n,f(k))\nn \\
  \lesssim&\left(\frac{k}{m} +\frac{kf(k)}{n}\right)^2 + \frac{\log^2 f(k)}{m} + \frac{f(k)}{n}.
\end{align}

\section{Proof of Proposition \ref{prop:lower-plug}}\label{app:prop2}

In this section, we derive the lower bound on the worst-case quadratic risk of the augmented plug-in estimator over the set $\mathcal{M}_{k,f(k)}$. We first prove the lower bound terms corresponding to the squared bias by choosing two different pairs of worst-case distributions. We then prove the lower bound terms corresponding to the variance using the minimax lower bound given by Le Cam's two-point method.

\subsection{Bounds on the Terms Corresponding to the Squared Bias}
It can be shown that the mean square error is lower bounded by the squared bias given as follows:
\begin{align}
   &\mathbb{E}\left[\big(\hat {D}_{\mathrm{A-plug-in}}(M,N)-D(P\|Q)\big)^2\right] \nn \\
& \ge \left(\mathbb{E}\left[\hat {D}_{\mathrm{A-plug-in}}(M,N)-D(P\|Q)\right]\right)^2.
\end{align}
We first decompose the bias into two parts:
\begin{align}\label{eq:des}
&\mathbb{E}[\hat {D}_{\mathrm{A-plug-in}}(M,N)-D(P\|Q)] \nn \\
&=\mathbb{E} \left[ \sum_{i=1}^k \Big( \frac{M_i}{m} \log \frac{M_i}{m} - P_i \log P_i \Big) \right] \nn \\
&\quad +\mathbb{E} \left[ \sum_{i=1}^k P_i \log \frac{(n+kc)Q_i}{N_i+c}  \right] .
\end{align}
The first term in \eqref{eq:des} is the bias of the plug-in entropy estimator. As shown in \cite{wu2014minimax} and \cite{paninski2003estimation}, the worst-case quadratic risk of the first term can be bounded as follows if $m \ge k$ holds,

\begin{align}\label{eq:entropy1}
    \mathbb{E} \left[ \sum_{i=1}^k \Big( \frac{M_i}{m} \log \frac{M_i}{m} - P_i \log P_i \Big) \right] \ge \frac{k}{2m},
\end{align}
if $P$ is the uniform distribution,
\begin{align}\label{eq:entropy2}
    \mathbb{E} \left[ \sum_{i=1}^k \Big( \frac{M_i}{m} \log \frac{M_i}{m} - P_i \log P_i \Big) \right] \le \log\Big(1+\frac{k-1}{m} \Big),
\end{align}
for any $P$.

In the proof of Proposition \ref{prop:upper-plug}, \eqref{eq:plogq_lower}, \eqref{eq:plogq_upper1} and \eqref{eq:plogq_upper2} show that the second term in \eqref{eq:des} can be bounded by
\begin{equation}
  - \frac{ckf(k)}{n}\le \mathbb{E} \left[ \sum_{i=1}^k P_i \log \frac{(n+kc)Q_i}{N_i+c}  \right]  \lesssim \frac{kf(k)}{n}.
\end{equation}

Note that the bias of the augmented plug-in estimator can be decomposed into: 1) the bias due to estimating $\sum_{i=1}^k P_i\log P_i$; and 2) the bias due to estimating $\sum_{i=1}^k - P_i\log Q_i$. As shown above, the first bias term is always positive, but the second bias term can be negative. Hence, the two bias terms may cancel out partially or even fully. Thus, to prove the minimax lower bound, we first determine which bias term dominates, and then construct a pair of distributions such that the dominant bias term is either lower bounded by positive terms or upper bounded by negative terms.


We recall \eqref{eq:bias_decomp}, and rewrite it here for convenience.
\begin{align*}
&\mathbb{E} \left[ \sum_{i=1}^k P_i \log \frac{(n+kc)Q_i}{N_i+c}  \right] \nn \\
&=\underbrace{\sum_{i=1}^k P_i \left( \log\Big(Q_i+\frac{c}{n}\Big) - \mathbb{E}\Big[ \log \frac{N_i+c}{n}\Big]\right)}_{A_1} \nn \\
& \quad + \underbrace{\sum_{i=1}^k P_i \log \frac{(n+kc)Q_i}{nQ_i+c}}_{A_2}
\end{align*}
We next derive tighter bounds for the terms $A_1$ and $A_2$ using two different pairs of worst-case distributions by considering the following two cases.

\textbf{Case I:} If $\frac{k}{m} > (1+\epsilon)\frac{ckf(k)}{5n}$, where $\epsilon>0$ is a constant, and which implies that the number of samples drawn from $P$ is relatively smaller than the number of samples drawn from $Q$, then the first bias term dominates. To obtain a tight lower bound on the second term in \eqref{eq:des}, we choose the following $(P,Q)$:
\begin{flalign}\label{eq:distribution_lower}
  P &=\left(\frac{1}{k},\ \frac{1}{k},\ \cdots,\ \frac{1}{k}\right),\nn\\
  Q &=\left(\frac{10}{kf(k)},\ \cdots,\  \frac{10}{kf(k)},\ 1-\frac{10(k-1)}{kf(k)}\right).
\end{flalign}
It can be verified that $P$ and $Q$ are distributions and satisfy the density ratio constraint, if $f(k)\ge 10$. For this $(P,Q)$ pair, $A_1$ can be lower bounded by
\begin{align}
 A_1 & =\sum_{i=1}^k P_i \mathbb{E}\Big[ \log\frac{nQ_i+c}{N_i+c}\Big]\nn\\
 &= - \sum_{i=1}^k P_i \mathbb{E}\Big[ \log\Big(1+\frac{N_i-nQ_i}{nQ_i+c}\Big)\Big]\nn\\
 &\ge - \sum_{i=1}^k P_i  \frac{\mathbb{E}[N_i]-nQ_i}{nQ_i+c}=0.
\end{align}
Due to $\log(1+x) \ge \frac{x}{1+x}$, we lower bound $A_2$ by
\begin{align}
  A_2  
  &\ge \sum_{i=1}^k P_i \frac{(kQ_i-1)c}{(n+kc)Q_i} \nn \\
  &\ge \sum_{i=1}^{k-1} P_i \frac{(kQ_i-1)c}{(n+kc)Q_i}\nn \\
  &\ge \sum_{i=1}^{k-1} -P_i \frac{c}{(n+kc)Q_i} \nn \\
  &= -\frac{(k-1)f(k)c}{10(n+kc)}.
\end{align}
Thus, for the $(P,Q)$ pair in \eqref{eq:distribution_lower}, we have
\begin{align}\label{eq:tight_lower}
&\mathbb{E} \left[ \sum_{i=1}^k P_i \log \frac{(n+kc)Q_i}{N_i+c}  \right]  \nn \\
&=  A_1+A_2\ge -\frac{c(k-1)f(k)}{10(n+kc)}\ge -\frac{ckf(k)}{10n}.
\end{align}
Note that for the $(P,Q)$ in \eqref{eq:distribution_lower}, $P$ is an uniform distribution. Thus, we combine the bound \eqref{eq:entropy1} with \eqref{eq:tight_lower}, and obtain
\begin{align}
   &\mathbb{E}[\hat {D}_{\mathrm{A-plug-in}}(M,N)-D(P\|Q)]  \nn \\
&\ge \frac{k}{2m}-\frac{ckf(k)}{10n}\ge \frac{\epsilon k}{2(1+\epsilon)m},
\end{align}
where the last step follows from the assumption
\begin{equation}\label{eq:condition1}
  \frac{k}{m} > (1+\epsilon)\frac{ckf(k)}{5n},\quad \epsilon>0.
\end{equation}
Thus,
\begin{equation}\label{eq:bias_lower_1}
  \mathbb{E}[\hat {D}_{\mathrm{A-plug-in}}(M,N)-D(P\|Q)] \gtrsim \frac{k}{m} {\asymp} \frac{k}{m}+\frac{kf(k)}{n},
\end{equation}
where the last step holds under condition \eqref{eq:condition1}.

\textbf{Case II:} If $\frac{k}{m} \le (1+\epsilon)\frac{ckf(k)}{5n}$, which implies that the number of samples drawn from $P$ is relatively larger than the number of samples drawn from $Q$, then the second bias term dominates. We choose the following $(P,Q)$:
\begin{flalign}\label{eq:distribution_upper}
  P&=\Bigg(\frac{f(k)}{4n},\ \cdots,\ \frac{f(k)}{4n},\ 1-\frac{(k-1)f(k)}{4n}\Bigg), \nn \\
  Q&=\Bigg(\frac{1}{4n},\ \cdots,\ \frac{1}{4n},\ 1-\frac{k-1}{4n}\Bigg).
\end{flalign}
By the assumption that $n \ge 10 kf(k)$, it can be verified that $P$ and $Q$ are distributions and satisfy the density ratio constraint. For this $(P,Q)$ pair, $A_1$ can be upper bounded using the following lemma.

\begin{lemma}\label{lemma:Bernstein}\cite[Equation 10.3.4]{devore1993constructive}
If $f$ is twice continuously differentiable, then for $X\sim B(n,x)$
\begin{equation*}
  \big|f(x)-\mathbb{E}[f(X/n)]\big|\le \|f''\|_\infty\frac{x(1-x)}{2n}, \quad x\in [0,1].
\end{equation*}
\end{lemma}
Let $f(x)=\log(x+\frac{c}{n})$, and $g(x)=x\log(x+\frac{c}{n})$. It can be shown that $\|f^{(2)}\|_\infty = \frac{n^2}{c^2}$ and $\|g^{(2)}\|_\infty = \frac{2n}{c}$. Hence,
\begin{align}
  \left|\log\Big(x+\frac{c}{n}\Big)-\mathbb{E}\Big[\log\frac{X+c}{n}\Big]\right| &\le \frac{nx(1-x)}{2c^2},\label{eq:Bernstein1} \\
  \left|x\log\Big(x+\frac{c}{n}\Big)-\mathbb{E}\Big[\frac{X}{n}\log\frac{X+c}{n}\Big]\right| &\le \frac{x(1-x)}{c}.\label{eq:Bernstein2}
\end{align}
Thus, $A_1$ can be upper bounded by
\begin{align}\label{eq:bias_A1_upper1}
 A_1  =& \sum_{i=1}^{k-1}P_i\left( \log\Big(Q_i+\frac{c}{n}\Big) - \mathbb{E} \Big[\log \frac{N_i+c}{n}\Big] \right) \nn \\
 &+ P_k\left( \log\Big(Q_k+\frac{c}{n}\Big) - \mathbb{E} \Big[\log \frac{N_k+c}{n}\Big]\right)\nn \\
 \overset{(a)}{\le}& \sum_{i=1}^{k-1}P_i\left| \log\Big(Q_i+\frac{c}{n}\Big) - \mathbb{E}\Big[ \log \frac{N_i+c}{n}\Big]\right| \nn\\
 & + \left|Q_k \log\Big(Q_k+\frac{c}{n}\Big) - Q_k \mathbb{E}\Big[ \log \frac{N_k+c}{n}\Big]\right|\nn\\
 \overset{(b)}{\le} &\sum_{i=1}^{k-1}P_i \frac{nQ_i(1-Q_i)}{2c^2} \nn \\
 &+ \left|Q_k \log\Big(Q_k+\frac{c}{n}\Big) - \mathbb{E} \Big[\frac{N_k}{n}  \log \frac{N_k+c}{n}\Big]\right| \nn \\
 &+\left|\mathbb{E} \Big[\Big( \frac{N_k}{n}-Q_k\Big) \log \frac{N_k+c}{n}\Big] \right|\nn \\
 \overset{(c)}{\le} &\frac{(k-1)f(k)}{32c^2n} + \frac{Q_k(1-Q_k)}{c}\nn \\
 &+\left|\mathbb{E} \Big[\Big( \frac{N_k}{n}-Q_k\Big) \log \frac{N_k+c}{n}\Big] \right|\nn \\
 \le & \frac{kf(k)}{32c^2n} + \frac{k}{4cn}+\left|\mathbb{E} \Big[\Big( \frac{N_k}{n}-Q_k\Big) \log \frac{N_k+c}{n}\Big] \right|,
\end{align}
where $(a)$ is due to the fact $P_k\le Q_k$ for the $(P,Q)$ given in \eqref{eq:distribution_upper}; and $(b)$ and $(c)$ follow from \eqref{eq:Bernstein1} and \eqref{eq:Bernstein2}, respectively.

The last term in \eqref{eq:bias_A1_upper1} can be further bounded by using steps similar to those in \eqref{eq:plogq_upper1term2} and \eqref{eq:var_plogq2} as follows:
\begin{align}\label{eq:bias_A1_upper2}
  &\left|\mathbb{E} \Big[\Big( \frac{N_k}{n}-Q_k\Big) \log \frac{N_k+c}{n}\Big] \right|^2 \nn \\
  &= \left|\mathbb{E} \left[\Big( \frac{N_k}{n}-Q_k\Big)\Big( \log \frac{N_k+c}{n}-\mathbb{E}\Big[\log \frac{N_k+c}{n}\Big]\Big)\right] \right|^2 \nn \\
  &\overset{(a)}{\le} \frac{Q_k(1-Q_k)}{n}\mathrm{Var}\Big[ \log \frac{N_k+c}{n}\Big]\nn \\
  &\overset{(b)}{\le} \frac{Q_k(1-Q_k)}{n} \left(\frac{4}{nQ_k}  + \frac{n^2Q_k^2}{c^2}e^{-nQ_k/8}\right)\nn \\
  &\overset{(c)}{\le} \frac{Q_k(1-Q_k)}{n} \left(\frac{4}{nQ_k}  + \frac{4}{c^2nQ_k}\right) \nn \\
  &\le  (1+\frac{1}{c^2})\frac{k}{n^3},
\end{align}
where $(a)$ follows from the Cauchy-Schwartz inequality; $(b)$  follows from \eqref{eq:var_plogq2}; and $(c)$ follows because $x^3e^{-x/8} < 4$ for $x>100$, $nQ_k=n-\frac{k-1}{4}$, and $n\ge 10kf(k)$.

Combining \eqref{eq:bias_A1_upper1} and \eqref{eq:bias_A1_upper2}, we have
\begin{equation}\label{eq:bias_A1_upper}
  A_1 \le \frac{kf(k)}{32c^2n} + \frac{k}{4cn}+ (\frac{1}{c}+1)\frac{\sqrt{k}}{n^{3/2}}.
\end{equation}

We next bound $A_2$ as follows,
\begin{align}
  A_2  
  =&  \frac{(k-1)f(k)}{4n} \log \frac{1+kc/n}{1+4c} \nn \\
  &\qquad+\left(1-\frac{(k-1)f(k)}{4n}\right) \log \frac{n+kc}{n+4cn/(4n-k+1)} \nn \\
  \overset{(a)}{\le}& \frac{(k-1)f(k)}{4n} \log \frac{1+c/10}{1+4c}+\log \frac{n+kc}{n+4cn/(4n-k+1)} \nn \\
  \overset{(b)}{\le}& \frac{(k-1)f(k)}{4n} \log \frac{1+c/10}{1+4c}+ \frac{kc-4cn/(4n-k+1)}{n+4cn/(4n-k+1)} \nn \\
  \le& \frac{kf(k)}{4n} \log \frac{1+c/10}{1+4c}+ \frac{kc}{n},
\end{align}
where $(a)$ is due to the assumption $n \ge 10kf(k)$; and $(b)$ follows because $\log(1+x)\le x$.

Thus, for the $(P,Q)$ pair in \eqref{eq:distribution_upper}, we have
\begin{align}\label{eq:tight_upper}
  &\mathbb{E} \left( \sum_{i=1}^k P_i \log \frac{(n+kc)Q_i}{N_i+c}  \right) \nn \\
  &\le \frac{kf(k)}{n}\Big(\frac{1}{4}\log \frac{1+c/10}{1+4c}+\frac{1}{32c^2} \nn \\
  & \qquad\qquad\qquad+\frac{c+1/(4c)}{f(k)} +\frac{c+1}{c}\frac{1}{\sqrt{nk}f(k)}\Big).
\end{align}
Since $\frac{1}{\sqrt{nk}}$ converges to zero as $n$ and $k$ go to infinity, we omit it in the following analysis.

For the distribution pair in \eqref{eq:distribution_upper}, we combine \eqref{eq:entropy2} and \eqref{eq:tight_upper}, and obtain
\begin{align}
&\mathbb{E}[\hat {D}_{\mathrm{A-plug-in}}(M,N)-D(P\|Q)] \nn \\
&\le \frac{k}{m}+\frac{kf(k)}{n}\left(\frac{1}{4}\log \frac{1+c/10}{1+4c}+\frac{1}{32c^2}+\frac{c+1/(4c)}{f(k)}\right).
\end{align}
Note that under the assumption $n\ge 10kf(k)$ and $f(k)\ge10$, we can always find $\epsilon>0$, such that
\begin{align}\label{eq:condition2}
&(1-\epsilon) \left(\frac{1}{4}\log \frac{1+4c}{1+c/10}-\frac{1}{32c^2}-\frac{c+1/(4c)}{f(k)} \right)\nn \\
&> (1+\epsilon)\frac{c}{5}\ge \frac{k/m}{kf(k)/n}
\end{align}
holds for all $\frac{2}{3}\le c \le \frac{5}{4}$. 
Then, the worst-case bias of the augmented plug-in estimator for the $(P,Q)$ in \eqref{eq:distribution_upper} is upper bounded by
\begin{align}\label{eq:bias_lower_2}
  &\mathbb{E}[\hat {D}_{\mathrm{A-plug-in}}(M,N)-D(P\|Q)] \nn \\
  &\le -\left(\frac{1}{4}\log \frac{1+4c}{1+c/10}-\frac{1}{32c^2}-\frac{c+1/(4c)}{f(k)} \right) \frac{\epsilon kf(k)}{n}\nn\\
  &\lesssim -\frac{kf(k)}{n} {\asymp} -\frac{kf(k)}{n}-\frac{k}{m},
\end{align}
where the last step holds under condition \eqref{eq:condition2}.

Following \eqref{eq:bias_lower_1} and \eqref{eq:bias_lower_2}, we conclude that
\begin{align}
 &R(\hat{D}_{\mathrm{A-plug-in}},k,m,n,f(k)) \nn \\
 &= \sup_{(P,Q) \in \mathcal{M}_{k,f(k)}} \mathbb{E}[(\hat{D}_{\mathrm{A-plug-in}}(M,N)-D(P\|Q))^2] \nn \\
 &\gtrsim  \left(\frac{kf(k)}{n} + \frac{k}{m}\right)^2.
\end{align}

\subsection{Bounds on the Terms Corresponding to the Variance}

\subsubsection{Proof of $R(\hat{D}_{\mathrm{A-plug-in}},k,m,n,f(k))\gtrsim \frac{\log^2 f(k)}{m} $}\label{app:c21}
We use the minimax risk as a lower bound on the worst-case quadratic risk for the augmented plug-in estimator. To this end, we apply Le Cam's two-point method. We first construct two pairs of distributions as follows:
\begin{small}
\begin{align}
  P^{(1)} & =\Big(\frac{1}{3(k-1)},\ \dots,\ \frac{1}{3(k-1)},\ \frac{2}{3} \Big), \\
  P^{(2)} & = \Big(\frac{1-\epsilon}{3(k-1)},\ \dots,\ \frac{1-\epsilon}{3(k-1)},\ \frac{2+\epsilon}{3} \Big),\\
  Q^{(1)} & =Q^{(2)} \nn \\
   &=\Big(\frac{1}{3(k-1)f(k)},\dots,\frac{1}{3(k-1)f(k)},1-\frac{1}{3f(k)} \Big) ,
\end{align}
\end{small}

The above distributions satisfy:
\begin{align}
  D(P^{(1)}\|Q^{(1)}) =& \frac{1}{3} \log f(k)+\frac{2}{3}\log \frac{2f(k)}{3f(k)-1},\\
  D(P^{(2)}\|Q^{(2)}) =& \frac{1-\epsilon}{3} \log (1-\epsilon) f(k)
  +\frac{2+\epsilon}{3}\log \frac{(2+\epsilon)f(k)}{3f(k)-1},\\
  D(P^{(1)}\|P^{(2)}) =& \frac{1}{3} \log \frac{1}{1-\epsilon} +\frac{2}{3}\log \frac{2}{2+\epsilon}.
\end{align}
We set $\epsilon = \frac{1}{\sqrt{m}}$, and obtain
\begin{align}
  &D(P^{(1)}\|P^{(2)}) \nn \\
  &=  \frac{1}{3} \log\Big( 1+ \frac{\epsilon}{1-\epsilon}\Big) +\frac{2}{3}\log \Big(1-\frac{\epsilon}{2+\epsilon}\Big) \nn \\
  &\le  \frac{\epsilon}{3(1-\epsilon)} -\frac{2}{3}\frac{\epsilon}{2+\epsilon} \nn \\
  &=  \frac{\epsilon^2}{(1-\epsilon)(2+\epsilon)}\le\frac{1}{m}.
\end{align}
Furthermore,
\begin{align}
  &D(P^{(1)}\|Q^{(1)}) - D(P^{(2)}\|Q^{(2)}) \nn\\
  &= \frac{1}{3}\log \frac{1}{1-\epsilon}+\frac{\epsilon}{3}\log(1-\epsilon)f(k)\nn \\
  &\qquad +\frac{2}{3}\log \frac{2}{2+\epsilon}-\frac{\epsilon}{3}\log \frac{2+\epsilon}{3-\frac{1}{f(k)}}\nn\\
  &= \frac{1}{3}\log \frac{1}{1-\epsilon}\frac{4}{(2+\epsilon)^2}-\frac{\epsilon}{3}\log \frac{2+\epsilon}{(1-\epsilon)(3f(k)-1)},
\end{align}
which implies that
\begin{align}
  &\big(D(P^{(1)}\|Q^{(1)}) - D(P^{(2)}\|Q^{(2)})\big)^2 \nn \\
&\gtrsim  \epsilon^2 \log^2 \frac{2}{(3f(k)-1)} \asymp \frac{\log^2 f(k)}{m},
\end{align}
as $m\to \infty$. Now applying Le Cam's two-point method, we obtain
\begin{small}
\begin{align}
  &R(\hat{D}_{\mathrm{A-plug-in}},k,m,n,f(k)) \nn \\
  &\ge R^*(k,m,n,f(k)) \nn \\
  &\ge \frac{1}{16}\big(D(P^{(1)}\|Q^{(1)})-D(P^{(2)}\|Q^{(2)})\big)^2\exp\big(-mD(P^{(1)}\|P^{(2)}) \big) \nn \\
  &\gtrsim  \frac{\log^2 f(k)}{m}.
\end{align}
\end{small}

\subsubsection{Proof of $R(\hat{D}_{\mathrm{A-plug-in}},k,m,n,f(k))\gtrsim \frac{f(k)}{n}$}\label{app:c22}
We construct two pairs of distributions as follows:
\begin{small}
\begin{align}
      P^{(1)} & =P^{(2)}  =\Big(\frac{1}{3(k-1)},\ 0,\ \dots,\ \frac{1}{3(k-1)},\ 0,\ \frac{5}{6} \Big), \\
      Q^{(1)} & = \Big(\frac{1}{2(k-1)f(k)}, \dots,\ \frac{1}{2(k-1)f(k)},1-\frac{1}{2f(k)} \Big),\\
      Q^{(2)}& =\Big(\frac{1-\epsilon}{2(k-1)f(k)},\ \frac{1+\epsilon}{2(k-1)f(k)},\ \dots,\ \nn \\
      & \quad \quad \frac{1-\epsilon}{2(k-1)f(k)},\ \frac{1+\epsilon}{2(k-1)f(k)},\ 1-\frac{1}{2f(k)} \Big),
\end{align}
\end{small}
It can be verified that if $\epsilon < \frac{1}{3}$, then the density ratio is bounded by $\frac{2f(k)}{3(1-\epsilon)}\le f(k)$.
We set $\epsilon = \sqrt{\frac{f(k)}{n}}$. The above distributions satisfy:
\begin{align}
  D(Q^{(1)}\|Q^{(2)}) =& \frac{1}{4f(k)}\log \frac{1}{1+\epsilon}+\frac{1}{4f(k)}\log \frac{1}{1-\epsilon}, \\
  D(P^{(1)}\|Q^{(1)})-&D(P^{(2)}\|Q^{(2)}) \nn \\
  =&\frac{1}{6}\log (1-\epsilon)\leq -\frac{\epsilon}{6}.
\end{align}

Due to $\epsilon = \sqrt{\frac{f(k)}{n}}$, it can be shown that
\begin{align}
  D(Q^{(1)}\|Q^{(2)}) =& \frac{1}{4f(k)}\log (1+\frac{\epsilon^2}{1-\epsilon^2}) \nn \\
  \le& \frac{1}{4f(k)} \frac{\epsilon^2}{1-\epsilon^2}
  < \frac{\epsilon^2}{f(k)}=\frac{1}{n}.
\end{align}



We apply Le Cam's two-point method, and obtain
\begin{flalign}
&R^*(k,m,n,f(k)) \nn \\
&\ge\frac{1}{16}\big(D(P^{(1)}\|Q^{(1)})-D(P^{(2)}\|Q^{(2)})\big)^2 \nn \\
& \quad \cdot\exp   \big(-mD(P^{(1)}\|P^{(2)})-nD(Q^{(1)}\|Q^{(2)}) \big)\nn\\
&\gtrsim  (D(P^{(1)}\|Q^{(1)}) - D(P^{(2)}\|Q^{(2)}))^2 \gtrsim \epsilon^2 \nn\\
&=  \frac{f(k)}{n}.
\end{flalign}

\section{Proof of Lemma \ref{lemma:poisson}}\label{app:poisson}
We prove the inequality \eqref{eq:poissonineq} that connects the minimax risk  \eqref{eq:minimaxriskdensity} under the deterministic sample size to the risk \eqref{eq:poissonsampling} under the Poisson sampling model. We first prove the left hand side of \eqref{eq:poissonineq}. Recall that $0\le R^*(k,m,n,f(k))\le \log ^2f(k)$ and $R^*(k,m,n,f(k))$ is decreasing with $m,n$. Therefore,
\begin{align}
  &\widetilde{R}^*(k,2m,2n,f(k))\nn \\
  &=\sum_{i \ge 0} \sum_{j \ge 0} R^*(k,i,j,f(k)) \mathrm{Poi}(2m,i)\mathrm{Poi}(2n,i)\nn\\
  &= \sum_{i \ge m+1} \sum_{j \ge n+1} R^*(k,i,j,f(k)) \mathrm{Poi}(2m,i)\mathrm{Poi}(2n,i)\nn \\
  &\qquad +\sum_{i \ge 0} \sum_{j =0}^n R^*(k,i,j,f(k)) \mathrm{Poi}(2m,i)\mathrm{Poi}(2n,i)\nn\\
  & \qquad+\sum_{i =0}^m \sum_{j \ge n+1} R^*(k,i,j,f(k)) \mathrm{Poi}(2m,i)\mathrm{Poi}(2n,i)\nn\\
  &\le  R^*(k,m,n,f(k))+ e^{-(1-\log 2)n}\log^2 f(k)\nn \\
   &\qquad+e^{-(1-\log 2)m}\log^2 f(k),
\end{align}
where the last inequality follows from the Chernoff bound $\mathbb{P}[\mathrm{Poi(2n)}\le n]\le \exp(-(1-\log2)n)$.
We then prove the right hand side of \eqref{eq:poissonineq}. By the minimax theorem,
\begin{equation}
  R^*(k,m,n,f(k)) = \sup_\pi \inf_{\hat{D}}\mathbb{E}[(\hat{D}(M,N)-D(P\|Q))^2],
\end{equation}
where $\pi$ ranges over all probability distribution pairs on  $\mathcal{M}_{k,f(k)}$ and the expectation is over $(P,Q)\sim \pi$.

Fix a prior $\pi$ and an arbitrary sequence of estimators $\{\hat{D}_{m,n}\}$ indexed by the sample sizes $m$ and $n$. It is unclear whether the sequence of Batesian risks $\alpha_{m,n} = \mathbb{E}[(\hat{D}_{m,n}(M,N)-D(P\|Q))^2]$ with respect to $\pi$ is decreasing in $m$ or $n$. However, we can define $\{\widetilde{\alpha}_{i,j}\}$ as
\begin{equation}
  \widetilde{\alpha}_{0,0} =\alpha_{0,0},\quad \widetilde{\alpha}_{i,j}= \alpha_{i,j} \wedge \alpha_{i-1,j} \wedge \alpha_{i,j-1}.
\end{equation}
Further define,
\begin{equation}
  \widetilde{D}_{m,n}(M,N) \triangleq \left\{
                              \begin{array}{ll}
                                \hat{D}_{m,n}(M,N), &\text{ if } \hbox{$\widetilde{\alpha}_{m,n}= \alpha_{m,n}$;} \\
                                \hat{D}_{m-1,n}(M,N), &\text{ if } \hbox{$\widetilde{\alpha}_{m,n}= \alpha_{m-1,n}$;} \\
                                \hat{D}_{m,n-1}(M,N), &\text{ if } \hbox{$\widetilde{\alpha}_{m,n}= \alpha_{m,n-1}$.}
                              \end{array}
                            \right.
\end{equation}
Then for $m'\sim\mathrm{Poi}(m/2)$ and $n'\sim\mathrm{Poi}(n/2)$, and $(P,Q) \sim \pi$, we have
\begin{small}
\begin{align}
&\mathbb{E}\left[\big(\hat{D}_{m',n'}(M',N')-D(P\|Q)\big)^2\right] \nn \\
&=\sum_{i\ge 0}\sum_{j\ge 0} \mathbb{E}\left[\big(\hat{D}_{i,j}(M',N')-D(P\|Q)\big)^2\right] \mathrm{Poi}(\frac{m}{2},i) \mathrm{Poi}(\frac{n}{2},j)\nn\\
&\ge  \sum_{i\ge 0}\sum_{j\ge 0} \mathbb{E}\left[\big(\widetilde{D}_{i,j}(M,N)-D(P\|Q)\big)^2\right] \mathrm{Poi}(\frac{m}{2},i) \mathrm{Poi}(\frac{n}{2},j)\nn\\
&\ge  \sum_{i=0}^m\sum_{j=0}^n \mathbb{E}\left[\big(\widetilde{D}_{i,j}(M,N)-D(P\|Q)\big)^2\right] \mathrm{Poi}(\frac{m}{2},i) \mathrm{Poi}(\frac{n}{2},j)\nn\\
&\overset{(a)}{\ge}  \frac{1}{4}\mathbb{E}\left[\big(\widetilde{D}_{m,n}(M,N)-D(P\|Q)\big)^2\right],
\end{align}
\end{small}
where $(a)$ is due to the Markov's inequality: $\mathbb{P} [\mathrm{Poi}(n/2) \ge n] \le \frac{1}{2}$. If we take infimum of the left hand side over $\hat D_{m,n}$, then take supremum of both sides over $\pi$, and use the Batesian risk as a  lower bound on the minimax risk, then we can show that
\begin{equation}
  \widetilde{R}^*(k,\frac{m}{2},\frac{n}{2},f(k))\ge \frac{1}{4} R^*(k,{m},{n},f(k)).
\end{equation}

\section{Proof of Proposition \ref{thm:minimax}}
\label{app:prop3}
\subsection{Bounds Using Le Cam's Two-Point Method}
\subsubsection{Proof of $R^*(k,m,n,f(k))\gtrsim \frac{\log^2 f(k)}{m} $}

Following the same steps in Appendix \ref{app:c21}, we can show
\begin{align}
  R^*(k,m,n,f(k)) \gtrsim & \big(D(P^{(1)}\|Q^{(1)}) - D(P^{(2)}\|Q^{(2)})\big)^2 \nn \\
  \gtrsim & \frac{\log ^2 f(k)}{m}.
\end{align}


\subsubsection{Proof of $R^*(k,m,n,f(k))\gtrsim \frac{f(k)}{n}$}

Following the same steps in Appendix \ref{app:c22}, we can show
\begin{align}
  R^*(k,m,n,f(k)) \gtrsim & \big(D(P^{(1)}\|Q^{(1)}) - D(P^{(2)}\|Q^{(2)})\big)^2 \nn \\
  \gtrsim & \frac{f(k)}{n}.
\end{align}

\subsection{Bounds Using Generalized Le Cam's Method}
\subsubsection{Proof of ${R}^*(k,m,n,f(k))\gtrsim (\frac{k}{m\log k})^2$}
Let $Q^{(0)}$ denote the uniform distribution. The minimax risk is lower bounded as follows:
\begin{flalign}
 &{R}^*(k,m,n,f(k)) \nn \\
 &=\inf_{\hat {D}}\sup_{(P,Q)\in \mathcal{M}_{k,f(k)} }\mathbb{E}[(\hat {D}(M,N)-D(P\|Q))^2]\nn\\
 &\ge \inf_{\hat {D}}\sup_{(P,Q^{(0)})\in \mathcal{M}_{k,f(k)} }\mathbb{E}[(\hat {D}(M,Q^{(0)})-D(P\|Q^{(0)}))^2]\nn\\
 &\triangleq R^*(k,m,Q^{(0)},f(k)).
\end{flalign}
If $Q=Q^{(0)}$ is  known, then estimating the KL divergence between $P$ and $Q^{(0)}$ is equivalent to estimating the entropy of $P$, because
\begin{flalign}
  D(P\|Q^{(0)})=&\sum_{i=1}^k \left(P_i\log P_i +P_i\log \frac{1}{Q^{(0)}_i}\right)\nn\\
  =&-H(P)+\log k.
\end{flalign}
Hence,  ${R}^*(k,m,Q^{(0)},f(k))$ is equivalent to the following minimax risk of estimating the entropy of distribution $P$ with $P_i\le \frac{f(k)}{k}$ for $i \in [k]$ such that the ratio between $P$ and $Q^{(0)}$ is upper bounded by $f(k)$.
\begin{flalign}\label{eq:entropyfk}
  R^*(k,m,Q^{(0)},f(k))=\inf_{\hat {H}}\sup_{P:P_i\le \frac{f(k)}{k}}\mathbb{E}[( \hat H(M)-H(P))^2].
\end{flalign}
If  $m\gtrsim \frac{k}{\log k}$, as shown in \cite{wu2014minimax}, the minimax lower bound for estimating entropy is given by
\begin{flalign}
  \inf_{\hat {H}}\sup_{P}\mathbb{E}[( \hat H(M)-H(P))^2]\gtrsim (\frac{k}{m\log k})^2.
\end{flalign}
The supremum is achieved for $P_i \le \frac{\log^2 k}{k}$. Comparing this result to \eqref{eq:entropyfk}, if $f(k)\ge \log^2k$, then
\begin{equation}
  \frac{\log^2 k}{k} \le \frac{f(k)}{k}.
\end{equation}
Thus, we can use the minimax lower bound of entropy estimation as the lower bound for divergence estimation on $\mathcal{M}_{k,f(k)}$,
\begin{flalign}
   {R}^*(k,m,n,f(k))\gtrsim {R}^*(k,m,Q^{(0)},f(k))\gtrsim (\frac{k}{m\log k})^2.
\end{flalign}

\subsubsection{Proof of ${R}^*(k,m,n,f(k))\gtrsim (\frac{kf(k)}{n\log k})^2 $}
Since $n\gtrsim \frac{kf(k)}{\log k}$, we assume that $n\ge \frac{C'kf(k)}{\log k}$. If $C'\ge 1$, we set $P=P^{(0)}$, where
\begin{flalign}\label{eq:distribution2}
  P^{(0)}=\bigg(\frac{f(k)}{n\log k},\ \ldots,\ \frac{f(k)}{n\log k},\ 1-\frac{(k-1)f(k)}{n\log k}\bigg).
\end{flalign}
Then, we have
$0\le 1-\frac{(k-1)f(k)}{n\log k}\le 1$. Hence, $P^{(0)}$ is a well-defined probability distribution. If $C'<1$, we set $P^{(0)}$ as follows:
\begin{flalign}
  P^{(0)}=\bigg(\frac{C'f(k)}{n\log k},\ \ldots,\ \frac{C'f(k)}{n\log k},\ 1-\frac{C'(k-1)f(k)}{n\log k}\bigg).
\end{flalign}
which is also a well defined probability distribution. In the following, we focus on the case that $C'\ge 1$. And the results can be easily generalized to the case when $C'<1$.

If $P=P^{(0)}$ given in \eqref{eq:distribution2} and is known, then estimating the KL divergence between $P$ and $Q$ is equivalent to estimating the following  function:
\begin{align}
  D(P^{(0)}\|Q)=&\sum_{i=1}^{k-1} \frac{f(k)}{n\log k}\log \frac{\frac{f(k)}{n\log k}}{Q_i} \nn \\
  &  + (1-\frac{(k-1)f(k)}{n\log k})\log \frac{1-\frac{(k-1)f(k)}{n\log k}}{Q_k},
\end{align}
which is further equivalent to estimating
\begin{flalign}
\sum_{i=1}^{k-1} \frac{f(k)}{n\log k}\log \frac{1}{Q_i} + (1-\frac{(k-1)f(k)}{n\log k})\log \frac{1}{Q_k}.
\end{flalign}

We further consider the following subset of $\mathcal{M}_{k,f(k)} $:
\begin{align}
  \mathcal{N}_{k,f(k)} \triangleq\{(P^{(0)},Q)\in &\mathcal{M}_{k,f(k)}: \nn \\
  \frac{1}{n\log k}\le  &Q_i\le \frac{c_4\log k}{n}, \forall\ i \in [k-1] \},
\end{align}
where $c_4$ is a constant defined later.

The minimax risk can be lower bounded as follows:
\begin{flalign}
 &{R}^*(k,m,n,f(k)) \nn \\
 &=\inf_{\hat {D}}\sup_{(P,Q)\in \mathcal{M}_{k,f(k)} }\mathbb{E}[(\hat {D}(M,N)-D(P\|Q))^2]\nn\\
 &\ge \inf_{\hat {D}}\sup_{(P^{(0)},Q)\in \mathcal{N}_{k,f(k)} }\mathbb{E}[(\hat {D}(P^{(0)},N)-D(P^{(0)}\|Q))^2]\nn\\
 &\triangleq {R}_{\mathcal{N}}^*(k,P^{(0)},n,f(k)).\label{eq:71}
\end{flalign}

For $0<\epsilon<1$, we introduce the following set of approximate probability vectors:
\begin{flalign}
  \mathcal {N}_{k,f(k)}(\epsilon)\triangleq \{(P^{(0)},\mathsf{Q}): \mathsf{Q} \in &\mathbb{R}_+^k,  |\sum_{i=1}^k \mathsf{Q}_i -1|\le \epsilon, \nn \\
  \frac{1}{n\log k}\le \mathsf{Q}_i\le &\frac{c_4\log k}{n}, \forall i \in [k-1]\}.
\end{flalign}
Note that $\mathsf{Q}$ is not a distribution. Furthermore, the set $\mathcal{N}_{k,f(k)}(\epsilon)$ reduces to $\mathcal{N}_{k,f(k)}$ if $\epsilon=0$.

We further consider the minimax quadratic risk \eqref{eq:71} under Poisson sampling on the set $\mathcal{N}_{k,f(k)}(\epsilon)$ as follows:
\begin{flalign}\label{eq:eps}
&\widetilde{R}_{\mathcal{N}}^*(k,P^{(0)},n,f(k),\epsilon) \nn \\
&=\inf_{\hat {D}}\sup_{(P^{(0)},\mathsf{Q}) \in \mathcal{N}_{k,f(k)}(\epsilon)} \mathbb{E}[(\hat D(P^{(0)},N)-D(P^{(0)},\mathsf{Q}))^2],
\end{flalign}
where $N_i\sim \text{Poi}(n\mathsf{Q}_i)$, for $i \in [k]$.
The risk \eqref{eq:eps} is connected to the risk \eqref{eq:71} for multinomial sampling by the following lemma.
\begin{lemma}\label{lemma:1}
For any $k$, $n\in \mathbb{N}$ and $\epsilon<1/3$,
\begin{align}
{R}_{\mathcal{N}}^*(k,P^{(0)},\frac{n}{2},f(k)) \ge &\frac{1}{2}\widetilde{R}_{\mathcal{N}}^*(k,P^{(0)},n,f(k),\epsilon) \nn \\
    &-\log^2 f(k)\exp{(-\frac{n}{50}) }-\log^2(1+\epsilon).
\end{align}
\end{lemma}
\begin{proof}
  See Appendix \ref{app:lemma2}.
\end{proof}
For $(P^{(0)},\mathsf{Q})\in \mathcal{N}_{k,f(k)}(\epsilon)$, we then apply the generalized Le Cam's method which involves two composite hypotheses as follows:
\begin{align}\label{eq:143}
  H_0 &:D(P^{(0)}\|\mathsf{Q})\le t \quad  \mathrm{versus} \nn \\
  H_1 &: D(P^{(0)}\|\mathsf{Q}) \ge t+\frac{(k-1)f(k)}{n\log k}d.
\end{align}
In the following we construct tractable prior distributions. Let $V$ and $V'$ be two $\mathbb{R}^+$ valued random variables defined on the interval $[\frac{1}{n\log k},\frac{c_4\log k}{n}]$ and have equal mean $\mathbb{E}(V)=\mathbb{E}(V')=\alpha$. We construct two random vectors
\begin{align}\label{eq:prior}
\mathsf{Q} = (V_1,\ \dots,\ V_{k-1},\ 1-(k-1)\alpha), \nn \\
\mathsf{Q}' = (V'_1,\ \dots,\ V'_{k-1},\ 1-(k-1)\alpha)
\end{align}
consisting of $k-1$ i.i.d. copies of $V$ and $V'$ and a deterministic term $1-{(k-1)\alpha}$, respectively. It can be verified that $(P^{(0)},\mathsf{Q}),\ (P^{(0)}, \mathsf{Q}') \in \mathcal{N}_{k,f(k)}(\epsilon)$ satisfy the density ratio constraint. Then the averaged divergences are separated by the distance of
\begin{align}
&|\mathbb E[D(P^{(0)}\|\mathsf{Q})]-\mathbb E[D(P^{(0)}\|\mathsf{Q}')]| \nn \\
&= \frac{(k-1)f(k)}{n\log k}|\mathbb E[\log V]-\mathbb E[\log V']|.
\end{align}

Thus, if we construct $V$ and $V'$ such that
\begin{equation}
  |\mathbb E[\log V]-\mathbb E[\log V']|\ge d,
\end{equation}
then the constructions in \eqref{eq:prior} satisfy \eqref{eq:143}, serving as the two composite hypotheses which are separated.

By such a construction, we have the following lemma via the generalized Le Cam's method:
\begin{lemma}\label{lemma:2}
Let $V$ and $V'$ be random variables such that $V$, $V'\in[\frac{1}{n\log k},\frac{c_4\log k}{n}]$, $\mathbb E[V]=\mathbb E[V']=\alpha$, and $|\mathbb E[\log V]-\mathbb E[\log V']|\ge d$. Then,
\begin{small}
\begin{flalign}
&\widetilde{R}_{\mathcal{N}}^*(k,P^{(0)},n,f(k),\epsilon) \nn \\
&\ge  \frac{(\frac{(k-1)f(k)d}{n\log k})^2}{32}\bigg(1-\frac{2(k-1)c_4^2\log^2 k}{n^2\epsilon^2}-\frac{32(\log n+\log\log k)^2}{(k-1)d^2}\nn\\
&\quad\qquad\qquad\qquad\qquad-k\mathrm{TV}(\mathbb E [\text{Poi}(nV)], \mathbb E [ \text{Poi}(nV')])\bigg),
\end{flalign}
\end{small}
where $\mathrm{TV}(P,Q)=\frac{1}{2}\sum_{i=1}^k |P_i-Q_i|$ denotes the total variation between two distributions.
\end{lemma}
\begin{proof}
  See Appendix \ref{app:lemma3}.
\end{proof}
To establish the impossibility of hypothesis testing between $V$ and $V'$, we also have the following lemma which provides an upper bound on the total variation of the two mixture Poisson distributions.
\begin{lemma}\cite[Lemma 3]{wu2014minimax}\label{lemma:4}
  Let $V$ and $V'$ be random variables on $[\frac{1}{n\log k},\frac{c_4\log k}{n}]$. If $\mathbb E[V^j]=\mathbb E[V'^j]$ for $j=1,\ldots,L$, and $L>\frac{2c_4\log k}{n}$, then,
  \begin{flalign}
    \mathrm{TV}&\left( \mathbb E [\text{Poi}(nV)], \mathbb E [ \text{Poi}(nV')]\right)\nn \\
    &\le 2\exp\left(-\Big(\frac{L}{2}\log\frac{L}{ {2ec_4\log k} }-2c_4\log k\Big)\right)\wedge 1.
  \end{flalign}
\end{lemma}

What remains is to construct $V$ and $V'$ to maximize $d = |\mathbb{E}[\log V']  -\mathbb{E}[\log V]|$, subject to the constraints in Lemma \ref{lemma:4}. Consider the following optimization problem over random variables $X$ and $X'$.
\begin{flalign}\label{eq:max}
  \mathcal{E}^*=&\max \mathbb E[\log  {X}]-\mathbb E[\log  {X'}]\nn\\
  &\text{s.t. } \mathbb E[X^j]=\mathbb E[X'^j],\quad j=1,\ldots,L\nn\\
  &\quad\quad X,X'\in[\frac{1}{c_4\log^2 k},1].
\end{flalign}
As shown in Appendix E in \cite{wu2014minimax}, the maximum $\mathcal{E}^*$ is equal to twice the error in approximating $\log x$ by a polynomial with degree $L$:
\begin{flalign}
  \mathcal{E}^*=2E_L(\log,[\frac{1}{c_4\log^2 k},1]).
\end{flalign}
The following lemma provides a lower bound on the error in the approximation of $\log x$ by a polynomial with degree $L$ over $[L^{-2},1]$.
\begin{lemma}\cite[Lemma 4]{wu2014minimax}\label{lemma:approx}
There exist universal positive constants $c$, $c'$, $L_0$ such that for any $L>L_0$,
\begin{flalign}
  E_{\lfloor cL\rfloor}(\log, [L^{-2},1])>c'.
\end{flalign}
\end{lemma}
Let $X$ and $X'$ be the maximizer of \eqref{eq:max}. We let $V=\frac{c_4\log k}{n} X$ and $V'=\frac{c_4\log k}{n} X'$, such that $V, V'\in[\frac{1}{n\log k},\frac{c_4\log k}{n}]$. Then it can be shown that
\begin{flalign}
  \mathbb E[\log {V}]-\mathbb E[\log {V'}]=\mathcal{E}^*,
\end{flalign}
where $V$ and $V'$ match up to $L$-th moment. We choose the value of $d$ to be $\mathcal{E}^*$.

Hence, we set $L=\lfloor c\log k\rfloor$. Then from Lemma \ref{lemma:approx}, $d= \mathcal{E}^* >2c'$. We further assume  that $\log ^2 n \le c_5 k$, set $c_4$ and $c_5$ such that $2c_4^2+\frac{8c_5}{c'^2}<1$ and $\frac{c}{2}\log\frac{c}{2ec_4}-2c_4>2$. Then from Lemma \ref{lemma:2} and Lemma \ref{lemma:1}, with $\epsilon=\frac{\sqrt{k}\log k}{n}$, the minimax risk is lower bounded as follows:
\begin{flalign}
   {R}^*(k,m,n,f(k))\ge&{R}_{\mathcal{N}}^*(k,P^{(0)},n,f(k))\nn \\
   \gtrsim & (\frac{kf(k)}{n\log k})^2.
\end{flalign}
\subsubsection{Proof of Lemma \ref{lemma:1}}\label{app:lemma2}
Fix $\delta>0$ and $(P^{(0)},Q)\in \mathcal{N}_{k,f(k)}(0)$. Let $\hat D(P^{(0)},n)$ be a near optimal minimax estimator for $D(P^{(0)}\|Q)$ with $n$ samples such that
\begin{flalign}
  \sup_{(P^{(0)},Q)\in \mathcal{N}_{k,f(k)}(0)}&\mathbb{E}[(\hat D(P^{(0)},n)-D(P^{(0)}\|Q))^2] \nn \\
  \le & \delta + {R}_{\mathcal{N}}^*(k,P^{(0)},n,f(k)).
\end{flalign}
For any $(P^{(0)},\mathsf{Q})\in \mathcal{N}_{k,f(k)}(\epsilon)$, $\mathsf{Q}$ is approximately a distribution. We normalize $\mathsf{Q}$ to be a probability distribution, i.e., $\frac{\mathsf{Q}}{\sum_{i=1}^k \mathsf{Q}_i}$, and then we have,
\begin{flalign}
  D(P^{(0)}\|\mathsf{Q}) =&\sum_{i=1}^k P_{0,i}\log \frac{P_{0,i}}{\mathsf{Q}_i} \nn \\
  =&-\log \sum_{i=1}^k \mathsf{Q}_i+D\Big(P^{(0)}\Big\|\frac{\mathsf{Q}}{\sum_{i=1}^k \mathsf{Q}_i}\Big).
\end{flalign}
Fix distributions $(P^{(0)},\mathsf{Q})\in \mathcal {N}_{k,f(k)}(\epsilon)$. Let $N=(N_1,\dots,N_k)$, and $N_i\sim$Poi$(n\mathsf{Q}_i)$. And define $n'=\sum N_i\sim\text{Poi}(n\sum \mathsf{Q}_i)$. We set an estimator under the Poisson sampling by
\begin{flalign}
  \widetilde D(P^{(0)},N)=\hat D(P^{(0)},n').
\end{flalign}
By the triangle inequality, we obtain
\begin{flalign}\label{eq:109}
  &\frac{1}{2}\big(\widetilde D(P^{(0)},N)-D(P^{(0)}\|\mathsf{Q})\big)^2 \nn \\
  &\le \left(\widetilde D(P^{(0)},N)-D\Big(P^{(0)}\Big\|\frac{\mathsf{Q}}{\sum_{i=1}^k \mathsf{Q}_i}\Big)\right)^2 \nn \\
   &\quad +\left(D\Big(P^{(0)}\Big\|\frac{\mathsf{Q}}{\sum_{i=1}^k \mathsf{Q}_i}\Big)-D(P^{(0)}\|\mathsf{Q})\right)^2\nn\\
  &= \left(\widetilde D(P^{(0)},N)-D\Big(P^{(0)}\Big\|\frac{\mathsf{Q}}{\sum_{i=1}^k \mathsf{Q}_i}\Big)\right)^2 + (\log \sum_{i=1}^k \mathsf{Q}_i)^2\nn\\
  &\le \left(\widetilde D(P^{(0)},N)-D\Big(P^{(0)}\Big\|\frac{\mathsf{Q}}{\sum_{i=1}^k \mathsf{Q}_i}\Big)\right)^2 +\log^2(1+\epsilon).
\end{flalign}
Since $n'=\sum N_i\sim$Poi$(n\sum Q_i)$, we can show that
\begin{small}
\begin{flalign}
  &\mathbb E\left[\Big(\widetilde D(P^{(0)},N)-D\Big(P^{(0)}\Big\|\frac{\mathsf{\mathsf{Q}}}{\sum_{i=1}^k \mathsf{\mathsf{Q}}_i}\Big)\Big)^2\right] \nn \\
  &=\sum_{j=1}^\infty \mathbb E\left[\Big(\hat D(P^{(0)},j)-D\Big(P^{(0)}\Big\|\frac{\mathsf{\mathsf{Q}}}{\sum_{i=1}^k \mathsf{\mathsf{Q}}_i}\Big)\Big)^2\bigg|n'=j\right]\mathbb P(n'=j)\nn\\
  &\le\sum_{j=1}^\infty {R}_{\mathcal{N}}^*(k,P^{(0)},j,f(k))\mathbb P(n'=j)+\delta.
\end{flalign}
\end{small}
We note that for fixed $k$, ${R}_{\mathcal{N}}^*(k,P^{(0)},j,f(k))$ is a monotone decreasing function with respect to $n$. We also have ${R}_{\mathcal{N}}^*(k,P^{(0)},j,f(k))\le \log^2 f(k)$, because for any $(P^{(0)},\mathsf{Q})\in \mathcal{N}_{k,f(k)}(0)$, $D(P^{(0)}\|\mathsf{Q})\le \log f(k)$. Furthermore, since $n'\sim \text{Poi}(n\sum \mathsf{Q}_i)$, and $|\sum \mathsf{Q}_i-1|\le\epsilon \le 1/3$, we have $P(n'>\frac{n}{2}) \le e^{-\frac{n}{50}}$. Hence, we obtain
\begin{flalign}\label{eq:111}
  &\mathbb E\left[\bigg(\widetilde D(P^{(0)},N)-D\Big(P^{(0)}\Big\|\frac{\mathsf{Q}}{\sum_{i=1}^k \mathsf{Q}_i}\Big)\bigg)^2\right]\nn\\
  &\le\sum_{j=1}^\infty {R}_{\mathcal{N}}^*(k,P^{(0)},j,f(k))\mathbb P(n'=j)+\delta\nn\\
  &=\sum_{j=1}^{n/2} {R}_{\mathcal{N}}^*(k,P^{(0)},j,f(k))\mathbb P(n'=j) \nn \\
  &\quad +\sum_{j=\frac{n}{2}+1}^{\infty} {R}_{\mathcal{N}}^*(k,P^{(0)},j,f(k))\mathbb P(n'=j)+\delta\nn\\
  &\le {R}_{\mathcal{N}}^*(k,P^{(0)},\frac{n}{2},f(k))+(\log ^2 f(k))P(n'>\frac{n}{2})+\delta\nn\\
  &\le {R}_{\mathcal{N}}^*(k,P^{(0)},\frac{n}{2},f(k))+\log ^2 f(k)e^{-\frac{n}{50}}+\delta.
\end{flalign}
Combining \eqref{eq:109} and \eqref{eq:111} completes the proof because $\delta$ can be arbitrarily small.

\subsubsection{Proof of Lemma \ref{lemma:2}}\label{app:lemma3}
We construct the following pairs of $(P,\mathsf{Q})$ and $(P',\mathsf{Q}')$:
\begin{small}
\begin{flalign}
  P&=P'=  P^{(0)}=\bigg(\frac{f(k)}{n\log k},\ \ldots,\ \frac{f(k)}{n\log k},\ 1-\frac{(k-1)f(k)}{n\log k}\bigg),\\
  \mathsf{Q}&=\left({V_1},\ \ldots,\ V_{k-1},\ 1-(k-1)\alpha\right),\\
  \mathsf{Q}'&=\left({V_1'},\ \ldots,\ {V'_{k-1}},\ 1-{(k-1)\alpha}\right).
\end{flalign}
\end{small}
We further define the following events:
\begin{flalign}
  E\triangleq \bigg\{\left|\sum_{i=1}^{k-1} {V_i} - {(k-1)\alpha} \right|\le&\epsilon, \nn \\
  |D(P\|\mathsf{Q})-\mathbb{E}(D(P\|\mathsf{Q}&))|\le \frac{d(k-1)f(k)}{4n\log k}\bigg\},\\
  E'\triangleq \bigg\{\left|\sum_{i=1}^{k-1} {V'_i} - {(k-1)\alpha} \right|\le&\epsilon, \nn \\
  |D(P'\|\mathsf{Q}')-\mathbb{E}(D(P'\|\mathsf{Q}'&))|\le \frac{d(k-1)f(k)}{4n\log k}\bigg\}.
\end{flalign}
By union bound and Chebyshev's inequality, we have
\begin{flalign}
  P(E^C)&\le \frac{(k-1)\text{Var}(V)}{\epsilon^2}+\frac{16(k-1)\text{Var}(\frac{f(k)}{n\log k}\log V_i)}{(\frac{(k-1)f(k)}{n\log k}d)^2}\nn\\
  &\le \frac{c_4^2 (k-1)\log^2 k}{\epsilon^2n^2}+\frac{16\log^2 ({n\log k}) }{(k-1)d^2}.
\end{flalign}
Similarly, we have
\begin{flalign}
   P(E'^C)\le \frac{c_4^2 (k-1)\log^2 k}{\epsilon^2n^2}+\frac{16\log^2 ({n\log k}) }{(k-1)d^2}.
\end{flalign}
Now, we define two priors on the set $\mathcal{N}_{k,f(k)}(\epsilon)$ by the following conditional distributions:
\begin{flalign}
  \pi=P_{V|E} \quad\text{ and }\quad  \pi'=P_{V'|E'}.
\end{flalign}
Hence, given $\pi$ and $\pi'$ as prior distributions, recall the assumption $|\mathbb E[\log V]-\mathbb E[\log V']|\ge d$, we have
\begin{flalign}
  |D(P\|\mathsf{Q})-D(P'\|\mathsf{Q}')|\ge \frac{d(k-1)f(k)}{2n\log k}.
\end{flalign}
Now, we consider the total variation of observations under $\pi$ and $\pi'$. The observations are Poisson distributed: $N_i\sim \text{Poi}(n\mathsf{Q}_i)$ and $N'_i\sim \text{Poi}(n\mathsf{Q}'_i)$.
By the triangle inequality, we have
\begin{flalign}\label{eq:39}
  &\mathrm{TV}(P_{N|E},P_{N'|E'})\nn \\
  &\le \mathrm{TV}(P_{N|E},P_N)+\mathrm{TV}(P_N,P_{N'})+\mathrm{TV}(P_{N'},P_{N'|E'})\nn\\
  &= P(E^C)+P(E'^C)+\mathrm{TV}(P_N,P_{N'})\nn\\
  &\le \frac{2c_4^2 (k-1)\log^2 k}{\epsilon^2n^2}+\frac{32\log^2 ({n\log k}) }{(k-1)d^2}+\mathrm{TV}(P_N, P_{N'}).
\end{flalign}
From the fact that total variation of product distribution can be upper
bounded by the summation of individual ones we obtain,
\begin{flalign}\label{eq:40}
  \mathrm{TV}(P_N, P_{N'})& \le \sum_{i=1}^{k-1}\mathrm{TV}(\mathbb E(\text{Poi}( {nV_i} )),\mathbb E(\text{Poi}( {nV'_i} ))),\nn\\
  &= k\mathrm{TV}\left(\mathbb E(\text{Poi}( {nV} )),\mathbb E(\text{Poi}( {nV'} ))\right).
\end{flalign}
Applying the generalized Le Cam's method \cite{Tsybakov2008}, and combining \eqref{eq:39} and \eqref{eq:40} completes the proof.

\section{Proof of Proposition \ref{prop:upperbound}}
\label{app:prop4}
We first denote 
\begin{equation}
  D_1\triangleq\sum_{i=1}^k P_i\log P_i,\quad  D_2\triangleq\sum_{i=1}^k P_i\log Q_i.
\end{equation}
Hence, $D(P\|Q)=D_1-D_2$. Recall that our estimator $\hat D_{\mathrm{opt}}$ for $D(P\|Q)$ is:
\begin{flalign}
  \hat D_{\mathrm{opt}}=\widetilde D_{\mathrm{opt}} \vee 0\wedge \log f(k),
\end{flalign}
where
\begin{flalign}
\widetilde D_{\mathrm{opt}}=&\hat D_1-\hat D_2,\\
  \hat D_1=&\sum_{i=1}^k \bigg(g_L'(M_i)\mathds{1}_{\{M_i'\le c'_2\log k\}} \nn \\
  &\qquad\qquad + (\frac{M_i}{m}\log \frac{M_i}{m} - \frac{1}{2m} )\mathds{1}_{\{M_i'> c'_2\log k\}}\bigg)\nn \\
  \triangleq & \sum_{i=1}^k \hat D_{1,i},\\
  \hat D_2=&\sum_{i=1}^k \bigg(\frac{M_i}{m} g_L( {N_i} ) \mathds{1}_{\{N_i'\le c_2\log k\}} \nn \\
  &\qquad + \frac{M_i}{m} \Big(\log \frac{N_i+1}{n} -\frac{1}{2(N_i+1)}\Big)\mathds{1}_{\{N_i'> c_2\log k\}}\bigg)\nn\\
  \triangleq& \sum_{i=1}^k \hat D_{2,i}.
\end{flalign}

We define the following sets:
\begin{flalign}
  E_{1,i}&\triangleq\{N_i'\le c_2\log k, Q_i\le \frac{c_1\log k}{n}\},\\
  E_{2,i}&\triangleq\{N_i'> c_2\log k, Q_i> \frac{c_3\log k}{n}\},
\end{flalign}
and
\begin{flalign}
  E_{1,i}'&\triangleq\{M_i'\le c_2'\log k, P_i\le \frac{c_1'\log k}{m}\},\\
  E_{2,i}'&\triangleq\{M_i'> c_2'\log k, P_i> \frac{c_3'\log k}{m}\},
\end{flalign}
where $c_1> c_2> c_3$ and $c_1'> c_2'> c_3'$.
We further define the following sets:
\begin{flalign}
  &E_1  \triangleq \bigcap_{i=1}^k E_{1,i},\quad   E_2  \triangleq \bigcap_{i=1}^k E_{2,i}, \\
  &E'_1 \triangleq \bigcap_{i=1}^k E'_{1,i},\quad   E'_2  \triangleq \bigcap_{i=1}^k E'_{2,i},\\
  &E\triangleq E_1 \cup E_2,\quad  E'\triangleq E'_1 \cup E'_2,\\
  &\bar E\triangleq E\cap E'= \bigcap_{i=1}^k \Big((E_{1,i}\cup E_{2,i})\cap (E_{1,i}'\cup E_{2,i}')\Big).
\end{flalign}
By union bound and Chernoff bound for Poisson distributions \cite[Theorem 5.4]{mitzenmacher2005probability}, we have
%
\begin{align}
  \mathbb{P}(\bar E^c)=&\mathbb{P}\bigg(  \bigcup_{i=1}^k (E_{1,i}\cup E_{2,i})^c \cup (E_{1,i}'\cup E_{2,i}')^c  \bigg)\nn\\
   \le& k\bigg(\mathbb{P}\Big(N_i'\le c_2\log k, Q_i> \frac{c_1\log k}{n}\Big) \nn \\
  &\quad + \mathbb{P}\Big(N_i'> c_2\log k, Q_i\le \frac{c_3\log k}{n}\Big)\nn \\
  &\quad+\mathbb{P}\Big(M_i'\le c_2'\log k, P_i> \frac{c_1'\log k}{m}\Big)\nn \\
  &\quad+\mathbb{P}\Big(M_i'> c_2'\log k, P_i\le \frac{c_3'\log k}{m}\Big) \bigg)\nn\\
  \le & \frac{1}{k^{c_1-c_2\log\frac{ec_1}{c_2}-1}}+\frac{1}{k^{c_3-c_2\log\frac{ec_3}{c_2}-1}}\nn \\
  &+\frac{1}{k^{c_1'-c_2'\log\frac{ec_1'}{c_2'}-1}}+\frac{1}{k^{c_3'-c_2'\log\frac{ec_3'}{c_2'}-1}}.
\end{align}
Note that $\hat D_{\mathrm{opt}},D(P\|Q)\in[0,\log f(k)]$, and $\hat D_{\mathrm{opt}}=\widetilde D_{\mathrm{opt}} \vee 0\wedge \log f(k)$. Therefore, we have
\begin{flalign}
  &\mE [(\hat D_{\mathrm{opt}}-D(P\|Q))^2]\nn\\
  &=\mE [(\hat D_{\mathrm{opt}}-D(P\|Q))^2 \mathds{1}_{\{\bar E\}} + (\hat D_{\mathrm{opt}}-D(P\|Q))^2 \mathds{1}_{\{\bar E^c\}} ]\nn\\
  &\le \mE [(\widetilde D_{\mathrm{opt}}-D(P\|Q))^2 \mathds{1}_{\{\bar E\}}] +\log^2 f(k)P(\bar E^c)\nn\\
  &=\mE[(\hat D_1-\hat D_2-D_1+D_2)^2\mathds{1}_{\{\bar E\}}] +\log^2 f(k)P(\bar E^c).
\end{flalign}
We choose constants $c_1,c_2,c_3,c_1',c_2',c_3'$ such that $c_1-c_2\log\frac{ec_1}{c_2}-1>C$, $c_1'-c_2'\log\frac{ec_1'}{c_2'}-1>C$, $c_3-c_2\log\frac{ec_3}{c_2}-1>C$, and $c_3'-c_2'\log\frac{ec_3'}{c_2'}-1>C$. Then together with $\log m\le C\log k$, we have
\begin{flalign}
  \log^2 f(k) P(\bar E^c)\le \frac{\log^2f(k)}{m}.
\end{flalign}

Define the index sets $I_1$, $I_2$, $I_1'$ and $I_2'$ as follows:
\begin{flalign}
  &I_1\triangleq \{i: N_i'\le c_2\log k, Q_i\le \frac{c_1\log k}{n}\},\nn\\
  &I_2\triangleq \{i: N_i'> c_2\log k, Q_i> \frac{c_3\log k}{n}\},\nn\\
    &I_1'\triangleq \{i: M_i'\le c_2'\log k, P_i\le \frac{c_1'\log k}{m}\},\nn\\
  &I_2'\triangleq \{i: M_i'> c_2'\log k, P_i> \frac{c_3'\log k}{m}\}.
\end{flalign}

Thus, we can upper bound $\mE[(\hat D_1-\hat D_2-D_1+D_2)^2\mathds{1}_{\{\bar E\}}]$ as follows:
\begin{flalign}\label{eq:135}
  &\mE[(\hat D_1-\hat D_2-D_1+D_2)^2\mathds{1}_{\{\bar E\}}]\nn\\
  &\le  \mE\bigg[\Big(\hat D_1-\hat D_2-D_1+D_2\Big)^2\bigg]\nn\\
  &=\mE\bigg[\mE ^2\Big(\hat D_1-\hat D_2-D_1+D_2 \Big|  I_1,I_2,I_1',I_2'\Big) \nn \\
   & \qquad\qquad + \text{Var}\Big(\hat D_1-\hat D_2\Big|  I_1,I_2,I_1',I_2'\Big)\bigg],
\end{flalign}
where the last step follows from the conditional variance formula. For the second term in \eqref{eq:135},
\begin{flalign}\label{eq:142}
  &\text{Var}\Big(\hat D_1-\hat D_2 \Big|  I_1,I_2,I_1',I_2'\Big)\nn\\
  &\le  4\text{Var}\bigg[\sum_{i\in I_1\cap I_1'} \big(\hat D_{1,i}-\hat D_{2,i}\big)\bigg|I_1, I_1'\bigg] \nn \\
  & \quad + 4\text{Var}\bigg[\sum_{i\in I_2\cap I_1'} \big(\hat D_{1,i}-\hat D_{2,i}\big)\bigg|I_2, I_1'\bigg] \nn\\
  & \quad + 4\text{Var}\bigg[\sum_{i\in I_1\cap I_2'} \big(\hat D_{1,i}-\hat D_{2,i}\big)\bigg|I_1 , I_2'\bigg] \nn \\
  & \quad + 4\text{Var}\bigg[\sum_{i\in I_2\cap I_2'} \big(\hat D_{1,i}-\hat D_{2,i}\big)  \bigg|I_2, I_2'\bigg].
\end{flalign}

Furthermore, we define $\mathcal{E}_1$, $\mathcal{E}_2$ and $\mathcal{E}'$ as follows:
\begin{flalign}
  \mathcal{E}_1\triangleq &\sum_{i\in I_1\cap(I_1'\cup I_2')} (\hat D_{2,i}-P_i\log Q_i),\\
  \mathcal{E}_2\triangleq &\sum_{i\in I_2\cap(I_1'\cup I_2')} (\hat D_{2,i}-P_i\log Q_i),\\
  \mathcal{E}'\triangleq &\sum_{i\in (I_1 \cup I_2)\cap (I_1'\cup I_2')} (\hat D_{1,i}-P_i\log P_i).
\end{flalign}
Then, the first term in \eqref{eq:135} can be bounded by
\begin{small}
\begin{align}\label{eq:100}
  &\mE ^2\Big(\hat D_1-\hat D_2-D_1+D_2 \Big|  I_1,I_2,I_1',I_2'\Big) \nn \\
&= \mE ^2\Big(\mathcal{E}'-\mathcal{E}_1-\mathcal{E}_2 \Big|  I_1,I_2,I_1',I_2'\Big)\nn \\
&\le 2 \mE ^2\Big(\mathcal{E}'|  I_1,I_2,I_1',I_2' \Big) + 2 \mE ^2 (\mathcal{E}_1+\mathcal{E}_2 |  I_1,I_2,I_1',I_2' )\nn\\
&\le 2 \mE ^2\Big(\mathcal{E}'|  I_1,I_2,I_1',I_2' \Big)+4\mE^2[\mathcal{E}_1|I_1,I_1',I_2'] + 4\mE^2[\mathcal{E}_2|I_2,I_1',I_2'].
\end{align}
\end{small}
Following steps similar to those in \cite{wu2014minimax}, it can be shown that
\begin{flalign}
  \mE ^2\Big(\mathcal{E}'|  I_1,I_2,I_1',I_2' \Big)\lesssim \frac{k^2}{m^2\log^2 k}.
\end{flalign}
Thus, in order to bound \eqref{eq:135}, we bound the four terms in \eqref{eq:142} and the last two terms in \eqref{eq:100} one by one.

\subsection{Bounds on the Variance}
\subsubsection{Bounds on Var$\Big[\sum_{i\in I_1\cap I_1'} (\hat D_{1,i}-\hat D_{2,i})\Big|I_1,I_1'\Big]$}\label{app:f21}

We first show that
\begin{align}
&\text{Var}\bigg[\sum_{i\in I_1\cap I_1'}( \hat D_{1,i}-\hat D_{2,i})\bigg|I_1,I_1'\bigg] \nn \\
&\le 2\text{Var}\bigg[\sum_{i\in I_1\cap I_1'} \hat D_{1,i}\bigg|I_1,I_1'\bigg] +2\text{Var}\bigg[\sum_{i\in I_1\cap I_1'}  \hat D_{2,i}\bigg|I_1,I_1'\bigg].
\end{align}
Following  steps similar to those in \cite{wu2014minimax}, it can be shown that
\begin{flalign}
  \text{Var}\bigg[\sum_{i\in I_1\cap I_1'} \hat D_{1,i}\bigg|I_1,I_1'\bigg]\lesssim \frac{k^2 }{m^2\log^2k}.
\end{flalign}
In order to bound $\text{Var}[\sum_{i\in I_1\cap I_1'}  \hat D_{2,i}|I_1,I_1']$, we  bound $\text{Var}(\frac{M_i}{m}g_L(N_i))$ for each $i\in I_1\cap I_1'$. Due to the independence between $M_i$ and $N_i$, $\frac{M_i}{m}$ is independent of $g_L(N_i)$. Hence,
\begin{flalign}
  &\text{Var}\bigg[\sum_{i\in I_1\cap I_1'}  \hat D_{2,i}\bigg|I_1,I_1'\bigg] \nn \\
  &= \sum_{i\in I_1\cap I_1'} \text{Var}\Big(\frac{M_i}{m}g_L(N_i)\Big)\nn\\
  &= \sum_{i\in I_1\cap I_1'} \bigg[\Big(\text{Var}(\frac{M_i}{m})+\mE(\frac{M_i}{m})^2 \Big)\text{Var}\big(g_L(N_i)\big)\nn \\
  &\qquad\qquad\qquad+\text{Var}(\frac{M_i}{m})\Big(\mE\big(g_L(N_i)\big)\Big)^2\bigg].
\end{flalign}
We note that $\text{Var}(\frac{M_i}{m})=\frac{P_i}{m}$, and $\mE(\frac{M_i}{m})=P_i$. We need to upper bound $\text{Var}(g_L(N_i))$ and $\Big(\mE\big(g_L(N_i)\big)\Big)^2$, for $i\in I_1\cap I_1'$. Recall that $g_L(N_i)=\sum_{j=1}^L \frac{a_j}{(c_1\log k)^{j-1}}(N_i)_{j-1}-\log\frac{n}{c_1\log k}$. The following lemma from \cite{wu2014minimax} is useful, which provides an upper bound on the variance of $(N_i)_j$.
\begin{lemma}\cite[Lemma 6]{wu2014minimax}\label{lemma:10}
  If $X\sim$Poi$(\lambda)$ and $(x)_j=\frac{x!}{(x-j)!}$, then the variance of $(X)_j$ is increasing in $\lambda$ and
  \begin{flalign}
    \text{Var}(X)_j\le (\lambda j)^j\left(\frac{(2e)^{2\sqrt{\lambda j}}}{\pi \sqrt{\lambda j}}\vee 1\right).
  \end{flalign}
\end{lemma}
Furthermore,  the polynomial coefficients can be upper bounded as $|a_j|\le 2e^{-1}2^{3L}$ \cite{cai2011}.
Due to the fact that the variance of the sum of random variables is upper bounded by the square of the sum of the individual standard deviations, we obtain
\begin{flalign}\label{eq:1b}
  &\text{Var}\Big(g_L(N_i)\Big) \nn \\
  &=\text{Var}\Big(\sum_{j=2}^L \frac{a_j}{(c_1\log k)^{j-1}}(N_i)_{j-1}\Big)\nn\\
  &\le \bigg( \sum_{j=2}^L\frac{a_j}{(c_1\log k)^{j-1}} \sqrt{\text{Var}\big((N_i)_{j-1}\big)}\bigg)^2\nn\\
  &\le \bigg( \sum_{j=2}^L\frac{2e^{-1}2^{3L}}{(c_1\log k)^{j-1}} \sqrt{\text{Var}\big((N_i)_{j-1}\big)}\bigg)^2.
\end{flalign}
By Lemma \ref{lemma:10}, we obtain
\begin{flalign}\label{eq:1a}
  &\text{Var}\Big((N_i)_{j-1}\Big) \nn \\
  &\le \big(c_1\log k (j-1)\big)^{j-1}\bigg(\frac{(2e)^{2\sqrt{c_1\log k(j-1)}}}{\pi\sqrt{c_1\log k (j-1)}}\vee1\bigg)\nn\\
  &\le (c_1c_0\log^2 k )^{j-1}\bigg(\frac{(2e)^{2\sqrt{c_1c_0\log^2 k}}}{\pi\sqrt{c_1c_0\log^2 k }}\vee1\bigg).
\end{flalign}
Substituting \eqref{eq:1a} into \eqref{eq:1b}, we obtain
\begin{flalign}
  \text{Var}\Big(g_L(N_i)\Big)\le & L  \sum_{j=2}^L \Big(\frac{2e^{-1}2^{3L}}{(c_1\log k)^{j-1}}\Big)^2 \text{Var}\Big((N_i)_{j-1}\Big)\nn \\
  &\lesssim  k^{2(c_0\log 8+\sqrt{c_0c_1}\log 2e)}\log k.
\end{flalign}
Furthermore, for $i\in I_1\cap I_1'$, we bound $\big|\mE \big(g_L(N_i)\big)\big|$ as follows:
\begin{flalign}
  \Big|\mE \big(g_L(N_i)\big)\Big|&=\left| \sum_{j=1}^L \frac{a_j}{(c_1\log k)^{j-1}}(nQ_i)^{j-1}-\log\frac{n}{c_1\log k}\right|\nn\\
  &\le \sum_{j=1}^L\frac{2e^{-1}2^{3L}}{(c_1\log k)^{j-1}}(c_1\log k)^{j-1}+\log\frac{n}{c_1\log k}\nn\\
  &\lesssim k^{c_0\log 8}\log k +\log n.
\end{flalign}

So far, we have all the ingredients we need to bound $\text{Var}\big(\frac{M_i}{m}g_L(N_i)\big)$. Note that  $P_i\le f(k)Q_i$, and $Q_i\le \frac{c_1\log k}{n}$ for $i\in I_1$.
First, we derive the following bound:
\begin{flalign}
  \text{Var}(\frac{M_i}{m})\text{Var}\big(g_L(N_i)\big)\lesssim & \frac{f(k)\log^2 kk^{2(c_0\log 8 +\sqrt{c_0c_1}\log 2e)}}{mn}\nn \\
  \lesssim & \frac{k f(k)}{mn \log^2 k},
\end{flalign}
if $2(c_0\log 8 +\sqrt{c_0c_1}\log 2e)<\frac{1}{2}$.

Secondly, we derive
\begin{flalign}
  \mE(\frac{M_i}{m})^2\text{Var}\big(g_L(N_i)\big)\lesssim& \frac{f^2(k)\log^3 k k^{2(c_0\log 8+\sqrt{c_0c_1}\log 2e)}}{n^2}\nn \\
  \lesssim& \frac{kf^2(k)}{n^2\log^2 k},
\end{flalign}
if $2(c_0\log 8 +\sqrt{c_0c_1}\log 2e)<\frac{1}{2}$.

Thirdly, we have
\begin{flalign}
  &\text{Var}(\frac{M_i}{m})\Big(\mE\big(g_L(N_i)\big)\Big)^2 \nn \\
  &\lesssim \frac{f(k)\log ^3k k^{2c_0\log 8}}{mn}+\frac{f(k)\log k\log ^2 n}{mn}\nn\\
  &\lesssim \frac{kf(k)}{mn \log^2 k} + \frac{k^{1-\epsilon}f(k)\log k}{mn}\nn\\
  &\lesssim \frac{kf(k)}{mn\log^2 k},
\end{flalign}
if $2c_0\log 8<\frac{1}{2}$ and $\log^2n \lesssim k^{1-\epsilon}$.

Combining these three terms together, we obtain
\begin{flalign}
  \text{Var}\bigg[\sum_{i\in I_1\cap I_1'}  \hat D_{2,i}\bigg|I_1, I_1'\bigg]\lesssim \frac{k^2f(k)}{mn\log^2 k} + \frac{k^2f^2(k)}{n^2\log^2 k}.
\end{flalign}
Due to the fact that $ \frac{k^2f(k)}{mn\log^2k}\lesssim \frac{k^2f^2(k)}{n^2\log^2 k}+\frac{k^2}{m^2\log^2k}$,
\begin{flalign}
  \mE\bigg[  \text{Var}\Big[\sum_{i\in I_1\cap I_1'}  \hat D_{2,i}\Big|I_1, I_1'\Big]\bigg]\lesssim \frac{f^2(k)k^2}{n^2\log^2 k} + \frac{k^2}{m^2\log^2k}.
\end{flalign}

\subsubsection{Bounds on $\text{Var}\Big[\sum_{i\in I_2\cap I_1'}\big( \hat D_{1,i}-\hat D_{2,i}\big)\Big|I_2,I_1'\Big]$}
Note that for $i\in I_2\cap I_1'$, $Q_i>\frac{c_3\log k}{n}$ and $P_i\le \frac{c_1'\log k}{m}$.
Following  steps similar to those in \cite{wu2014minimax}, it can be shown that
\begin{flalign}\label{eq:160}
  \text{Var}\bigg[\sum_{i\in I_2\cap I_1'} \hat D_{1,i}\bigg|I_2,I_1'\bigg]\lesssim \frac{k^2}{m^2\log^2k}.
\end{flalign}
We further consider $\text{Var}\left[\sum_{i\in I_2\cap I_1'} \hat D_{2,i}\right]$. By the definition of $\hat D_{2,i}$, for $i\in I_2\cap I_1'$, we have $\hat D_{2,i}=
\frac{M_i}{m} \Big(\log \frac{N_i+1}{n} -\frac{1}{2(N_i+1)}\Big)$. Therefore
\begin{flalign}\label{eq:161}
&\text{Var}\bigg[\sum_{i\in I_2\cap I_1'} \hat D_{2,i}\bigg|I_2,I_1'\bigg]\nn \\
&= \sum_{i\in I_2\cap I_1'}  \text{Var}\left[ \frac{M_i}{m} \Big(\log \frac{N_i+1}{n} -\frac{1}{2(N_i+1)}\Big)\right]\nn\\
&\le  2\sum_{i\in I_2\cap I_1'}  \text{Var}\left[ \frac{M_i}{m} \log \frac{N_i+1}{n} \right] \nn \\
& \qquad+2\sum_{i\in I_2\cap I_1'} \text{Var}\left[\frac{M_i}{m}\frac{1}{2(N_i+1)}\right].
\end{flalign}

The first term in \eqref{eq:161} can be bounded as follows:
\begin{flalign}\label{eq:164}
  &\sum_{i\in I_2\cap I_1'}  \text{Var}\left[\frac{M_i}{m} \log \frac{N_i+1}{n} \right]\nn\\
  &\le \sum_{i\in I_2\cap I_1'} \mE\left[\Big(\frac{M_i}{m} \log \frac{N_i+1}{n}-P_i\log \frac{nQ_i+1}{n}\Big)^2\right]\nn\\
  &=\sum_{i\in I_2\cap I_1'}  \mE\bigg[\Big(\frac{M_i}{m}\log \frac{N_i+1}{n}-\frac{M_i}{m}\log \frac{nQ_i+1}{n}\nn \\
  &\qquad\qquad\qquad+\frac{M_i}{m}\log  \frac{nQ_i+1}{n}  -P_i\log \frac{nQ_i+1}{n}\Big)^2\bigg]\nn\\
  &\le \sum_{i\in I_2\cap I_1'} 2\mE\left[\Big(\frac{M_i}{m}\Big)^2\Big( \log \frac{N_i+1}{n}-\log  \frac{nQ_i+1}{n} \Big)^2\right]\nn \\
  &\quad + \sum_{i\in I_2\cap I_1'} 2\mE\left[\Big( \big(\frac{M_i}{m}-P_i\big)\log  \frac{nQ_i+1}{n}  \Big)^2 \right].
\end{flalign}
Following similar steps as in \eqref{eq:var_plogq} with $c=1$, but for Poisson distribution, for $i\in I_2\cap I_1'$, $Q_i>\frac{c_3\log k}{n}$ and $P_i\le \frac{c_1'\log k}{m}$, we have the following bound on the first term in \eqref{eq:164}:
\begin{flalign}\label{eq:163}
  &\sum_{i\in I_2\cap I_1'} \mE\bigg[\Big(\frac{M_i}{m}\Big)^2\Big( \log \frac{N_i+1}{n}-\log \frac{nQ_i+1}{n} \Big)^2 \bigg] \nn \\
  &\lesssim  \sum_{i\in I_2\cap I_1'} \Big(\frac{P_i}{m} +P_i^2 \Big) \frac{1}{nQ_i}\lesssim \frac{f(k)}{n}+\frac{1}{m}.
 %
 %
\end{flalign}

We next bound the second term in \eqref{eq:164} as follow,
\begin{flalign}\label{eq:166}
  &\sum_{i\in I_2\cap I_1'} \mE\left[\Big( \big(\frac{M_i}{m}-P_i\big)\log  \frac{nQ_i+1}{n}  \Big)^2 \right] \nn \\
  &=\sum_{i\in I_2\cap I_1'} \frac{P_i}{m}\log^2 \Big(Q_i+\frac{1}{n}\Big) \nn\\
  &\le \sum_{i\in I_2\cap I_1'} \frac{P_i\log^2\frac{P_i}{f(k)}}{m}\nn\\
  &\le \sum_{i\in I_2\cap I_1'} \frac{2P_i(\log^2{P_i}+\log^2{f(k)})}{m}\nn\\
  &\overset{(a)}{\le} \frac{2k\frac{c_1'\log k}{m}\log^2(\frac{c_1'\log k}{m})}{m} + \frac{2\log^2 f(k)}{m}  \nn\\
  &\overset{(b)}{\lesssim} \frac{k^2}{m^2\log^2 k}+\frac{\log^2 f(k)}{m},
\end{flalign}
where $(a)$ is due to the facts that $x\log^2x$ is monotone increasing when $x$ is small and $P_i\le \frac{c_1'\log k}{m}$, and $(b)$ is due to the assumption that $\log m\lesssim \log k$.

Substituting \eqref{eq:166} and \eqref{eq:163} into \eqref{eq:164}, we obtain
\begin{flalign}\label{eq:167}
&\sum_{i\in I_2\cap I_1'}  \text{Var}\left[\frac{M_i}{m} \Big(\log \frac{N_i+1}{n}\Big) \right] \nn \\
&\lesssim  \frac{k^2}{m^2\log ^2k}+\frac{\log^2 f(k)}{m}+\frac{f(k)}{n}.
\end{flalign}

We then consider the second term in \eqref{eq:161}.
\begin{flalign}\label{eq:168}
  &\sum_{i\in I_2\cap I_1'} \text{Var}\left[\frac{M_i}{m}\Big(\frac{1}{2(N_i+1)}\Big)\right]\nn\\
  &=\sum_{i\in I_2\cap I_1'}\bigg(\mE^2[\frac{M_i}{m}] \text{Var}[\frac{1}{2(N_i+1)}] \nn \\
  &\quad + \text{Var}[\frac{M_i}{m}]\Big(\mE^2[\frac{1}{2(N_i+1)}]+\text{Var}[\frac{1}{2(N_i+1)}]\Big)\bigg).
\end{flalign}
In order to bound \eqref{eq:168}, we bound each term as follows. Note that $M_i\sim$Poi$(mP_i)$, and $N_i\sim$Poi$(nQ_i)$. Therefore, $\mE^2[\frac{M_i}{m}]=P_i^2$, $\text{Var}[\frac{M_i}{m}]=\frac{P_i}{m}$, and
\begin{flalign}
  &\text{Var}[\frac{1}{2(N_i+1)}]+\mE^2[\frac{1}{2(N_i+1)}] \nn \\
  &= \mE[\frac{1}{4(N_i+1)^2}]\nn\\
  &\le \mE[\frac{1}{(N_i+1)(N_i+2)}]\nn\\
  &=\sum_{i=0}^\infty \frac{1}{(i+1)(i+2)}\frac{e^{-nQ_i}(nQ_i)^i}{i!}\nn\\
  &=\sum_{i=0}^\infty \frac{1}{(nQ_i)^2}\frac{e^{-nQ_i}(nQ_i)^{i+2}}{(i+2)!}\nn\\
  &\le \frac{1}{(nQ_i)^2}.
\end{flalign}
Therefore, \eqref{eq:168} can be further upper bounded as follows:
\begin{flalign}\label{eq:170}
  &\sum_{i\in I_2\cap I_1'} \text{Var}\left[\frac{M_i}{m}\Big(\frac{1}{2(N_i+1)}\Big)\right]\nn \\
  &\le\sum_{i\in I_2\cap I_1'} \Big(P_i^2+\frac{P_i}{m}\Big)\frac{1}{(nQ_i)^2}\nn \\
  &\lesssim \frac{f(k)}{n\log k}+\frac{1}{m\log^2k}
  \lesssim \frac{f(k)}{n}+\frac{1}{m}.
\end{flalign}

Substituting \eqref{eq:170} and \eqref{eq:167} into \eqref{eq:161}, we obtain
\begin{flalign}
  &\text{Var}\left[\sum_{i\in I_2\cap I_1'} \hat D_{2,i}\bigg|I_2,I_1'\right]\lesssim \frac{k^2}{m^2\log^2k}+\frac{f(k)}{n}+\frac{\log^2f(k)}{m}.
\end{flalign}

Therefore,
\begin{flalign}
&\text{Var}\bigg[\sum_{i\in I_2\cap I_1'} \big( \hat D_{1,i}-\hat D_{2,i}\big)\bigg|I_2,I_1'\bigg]\nn \\
&\lesssim  \frac{k^2}{m^2\log^2k}+\frac{f(k)}{n}+\frac{\log^2f(k)}{m}.
\end{flalign}

\subsubsection{Bounds on $\text{Var}\Big[\sum_{i\in I_1\cap I_2'}\big( \hat D_{1,i}-\hat D_{2,i}\big)\Big|I_1,I_2'\Big]$}
We first note that given $i\in I_1\cap I_2'$, $P_i>\frac{c_3'\log k}{m}$, $Q_i\le \frac{c_1\log k}{n}$, and $\frac{P_i}{Q_i}\le f(k)$. Hence, $\frac{c_3'\log k}{m}<P_i\le \frac{c_1f(k)\log k}{n}$.
Following  steps similar to those  in \cite{wu2014minimax}, it can be shown that
\begin{flalign}\label{eq:162}
  &\text{Var}\bigg[\sum_{i\in I_1\cap I_2'} \hat D_{1,i}\bigg|I_1,I_2'\bigg] \nn \\
  &\le \frac{4}{m}+\frac{12k}{m^2}+\frac{4k}{c_3'm^2\log k}+\sum_{i\in I_1\cap I_2'} \frac{2P_i}{m}\log ^2 P_i.
\end{flalign}
Consider the last term $\sum_{i\in I_1\cap I_2'} \frac{2P_i}{m}\log ^2 P_i$ in \eqref{eq:162}, under the condition that $\frac{c_3'\log k}{m}<P_i\le \frac{c_1f(k)\log k}{n}$. Then,
\begin{flalign}
  \sum_{i\in I_1\cap I_2'} \frac{P_i}{m}\log ^2 P_i
  &\le \sum_{i\in I_1\cap I_2'} \frac{c_1f(k)\log k}{mn}\log^2 \frac{c_3'\log k}{m}\nn\\
  &\le \frac{c_1kf(k)\log k}{mn} \log^2 \frac{c_3'\log k}{m}\nn\\
  &\overset{(a)}{\lesssim} \frac{kf(k)\log k}{mn}  \log^2 m\nn\\
  &\overset{(b)}{\lesssim} \frac{kf(k)\log^3 k}{mn} \nn\\
  &\lesssim \frac{k^2f(k)}{mn\log^2k}\nn\\
  &\overset{(c)}{\lesssim }\frac{f^2(k)k^2}{n^2\log^2 k} + \frac{k^2}{m^2\log^2 k},
\end{flalign}
where $(a)$ is due to the assumption that $m\gtrsim \frac{k}{\log k}$, $(b)$ is due to the assumption that $\log m\le C\log k$, and $(c)$ is due to the fact that $2ab\le a^2+b^2$.
Therefore, we obtain
\begin{flalign}
&\text{Var}\left[\sum_{i\in I_1\cap I_2'} \hat D_{1,i}\bigg|I_1,I_2'\right] \nn \\
&\lesssim  \frac{\log^2 f(k)}{m} + \frac{f^2(k)k^2}{n^2\log^2 k} + \frac{k^2}{m^2\log^2 k}.
\end{flalign}

Following  steps similar to those in Appendix \ref{app:f21}, we can show that
\begin{flalign}
  \text{Var}\bigg[\sum_{i\in I_1\cap I_2'}\hat D_{2,i}\bigg|I_1,I_2'\bigg]\lesssim\frac{f^2(k)k^2}{n^2\log^2 k} + \frac{k^2}{m^2\log^2 k}.
\end{flalign}

Hence,
\begin{flalign}
&\text{Var}\bigg[\sum_{i\in I_1\cap I_2'} \big( \hat D_{1,i}-\hat D_{2,i} \big)\bigg|I_1,I_2'\bigg] \nn \\
&\lesssim  \frac{\log^2 f(k)}{m}+\frac{f^2(k)k^2}{n^2\log^2 k} + \frac{k^2}{m^2\log^2 k}.
\end{flalign}
\subsubsection{Bounds on $\text{Var}\Big[\sum_{i\in I_2\cap I_2'}\big( \hat D_{1,i}-\hat D_{2,i} \big) \Big|I_2,I_2'\Big]$}
We note that for $i\in I_2\cap I_2'$, $P_i>\frac{c_3'\log k}{m}$, $Q_i>\frac{c_3\log k}{n}$, and
\begin{flalign}
&\hat D_{1,i}-\hat D_{2,i}\nn \\
&=\frac{M_i}{m}\log \frac{M_i}{m} - \frac{1}{2m} -\frac{M_i}{m} \Big(\log \frac{N_i+1}{n} -\frac{1}{2(N_i+1)}\Big).
\end{flalign}
It can be shown that
\begin{flalign}\label{eq:179}
  &\text{Var}\bigg[\sum_{i\in I_2\cap I_2'} \big(\hat D_{1,i}-\hat D_{2,i}  \big) \bigg|I_2,I_2'\bigg]\nn\\
  &\le 2\text{Var}\bigg[\sum_{i\in I_2\cap I_2'} \frac{M_i}{m}\log \frac{M_i}{m}-\frac{M_i}{m} \log \frac{N_i+1}{n}   \bigg|I_2,I_2'\bigg] \nn \\
  &\quad+2\text{Var}\bigg[\sum_{i\in I_2\cap I_2'}\frac{M_i}{m}\frac{1}{2(N_i+1)}\bigg|I_2,I_2'\bigg].
\end{flalign}
Following  steps similar to those used in showing \eqref{eq:170}, we bound the second term in \eqref{eq:179} as follows:
\begin{flalign}\label{eq:180}
  \text{Var}\bigg[\sum_{i\in I_2\cap I_2'} \frac{M_i}{m}\frac{1}{2(N_i+1)}\bigg|I_2,I_2'\bigg]\lesssim \frac{f(k)}{n}+\frac{1}{m}.
\end{flalign}

We next bound the first term in \eqref{eq:179}, recall the decomposition in the proof of Proposition \ref{prop:upper-plug},
\begin{align}\label{eq:181}
&\text{Var}\bigg[\sum_{i\in I_2\cap I_2'} \frac{M_i}{m}\log \frac{M_i}{m}-\frac{M_i}{m} \log \frac{N_i+1}{n}   \bigg|I_2,I_2'\bigg]\nn\\
&\le  \sum_{i\in I_2\cap I_2'} \mathbb{E}\left[ \bigg(  \frac{M_i}{m} \log\frac{M_i/m}{(N_i+1)/n} -  P_i \log \frac{P_i}{(nQ_i+1)/n} \bigg)^2\right]\nn \\
&\le 3\sum_{i\in I_2\cap I_2'} \mathbb{E}\left[ \bigg(  \frac{M_i}{m} \Big(\log\frac{M_i}{m} - \log P_i\Big)\bigg)^2\right] \nn \\
&\quad + 3\sum_{i\in I_2\cap I_2'}\mathbb{E}\left[ \bigg( \frac{M_i}{m}\Big(\log \frac{N_i+1}{n}-\log \frac{nQ_i+1}{n}\Big) \bigg)^2\right]\nn\\
&\quad + 3\sum_{i\in I_2\cap I_2'}\mathbb{E}\left[ \bigg( \Big(\frac{M_i}{m}-P_i\Big) \log \frac{P_i}{(nQ_i+1)/n}\bigg)^2\right].
\end{align}
The first and the second terms can be upper bounded by similar steps as in \eqref{eq:var_plogp} and \eqref{eq:163} with $c=1$, note that the $M_i$ and $N_i$ are Poisson random variables.

For the third term in \eqref{eq:181}, we derive the following bound:
\begin{flalign}\label{eq:186}
&\sum_{i\in I_2\cap I_2'}\mE\bigg[\bigg(\Big(\frac{M_i}{m}-P_i\Big)\log \frac{P_i}{Q_i+1/n}\bigg)^2\bigg]\nn \\
&=\sum_{i\in I_2\cap I_2'}\frac{P_i}{m}\log ^2\frac{P_i}{Q_i+1/n} \lesssim \frac{\log^2f(k)}{m},
\end{flalign}
where the last inequality is because
\begin{flalign}\label{eq:187}
  &P_i\log ^2\frac{P_i}{Q_i+1/n} \nn \\
  &=  {P_i}\Big(\log \frac{P_i}{Q_i+1/n}\Big)^2 \mathds 1_{\{1 \le \frac{P_i}{Q_i+1/n}\le f(k)\}} \nn \\
  & \quad +P_i \Big(\log \frac{P_i}{Q_i+1/n}\Big)^2\mathds 1_{\{\frac{P_i}{Q_i+1/n}\le 1 \}}\nn \\
  &\overset{(a)}{\le}  P_i\log^2 f(k) \nn \\
  &\quad+ \frac{Q_i(1+c_3/\log k)P_i}{Q_i(1+c_3/\log k)}\Big(\log \frac{P_i}{Q_i(1+c_3/\log k)}\Big)^2\nn\\ 
  &\overset{(b)}{\lesssim}  P_i\log^2 f(k) + Q_i(1+c_3/\log k),
\end{flalign}
where $(a)$ follows from the condition $Q_i>\frac{c_3\log k}{n}$; and $(b)$ is because the $x\log^2 x$ is bounded by a constant on the interval $[0,1]$.

%

Combining \eqref{eq:var_plogp}, \eqref{eq:163} and \eqref{eq:186}, we obtain
\begin{flalign}
  &\text{Var}\bigg[\sum_{i\in I_2\cap I_2'}  \Big (\hat D_{1,i}-\hat D_{2,i} \Big) \bigg|I_2,I_2'\bigg]\nn \\
  &\lesssim  \frac{k}{m^2}+\frac{\log^2 f(k)}{m}+\frac{f(k)}{n}.
\end{flalign}


\subsection{Bounds on the Bias:}
Consider the $\mathcal{E}_1$ term in \eqref{eq:100}. Based on the definition of the set $I_1$, $\mathcal{E}_1$ can be written as follows:
\begin{flalign}
  \mathcal{E}_1=\sum_{i\in I_1\cap(I_1'\cup I_2')} \left(\frac{M_i}{m}g_L(N_i)-P_i\log Q_i\right).
\end{flalign}
For $i\in I_1\cap(I_1'\cup I_2')$, we have $0\le Q_i\le \frac{c_1\log k}{n}$ and $\Big|P_i\frac{\mu_L(Q_i)}{Q_i}-\frac{P_i}{Q_i}Q_i\log Q_i\Big|\lesssim \frac{f(k)}{n\log k}$. Therefore,
\begin{flalign}
&\left|\mE \Big[\frac{M_i}{m}g_L(N_i)-P_i\log Q_i\Big|I_1,I_1',I_2'\Big] \right| \nn \\
&=\left|P_i\frac{\mu_L(Q_i)}{Q_i}-P_i\log Q_i\right|
  \lesssim \frac{f(k)}{n\log k}.
\end{flalign}
Hence, $\big|\mE [\mathcal{E}_1|I_1,I_1',I_2']\big|$ can be bounded as follows:
\begin{flalign}
  &\big|\mE(\mathcal{E}_1|I_1,I_1',I_2')\big| \nn \\
  &\le  \sum_{i\in I_1\cap(I_1'\cup I_2')} \left|\mE\Big[\frac{M_i}{m}g_L(N_i)-P_i\log Q_i\Big|I_1,I_1',I_2'\Big]\right|\nn \\
  &\lesssim  \frac{kf(k)}{n\log k}.
\end{flalign}
Therefore,
\begin{flalign}
\mE\Big[\mE^2\big[\mathcal{E}_1|I_1,I_1',I_2'\big]\Big]\lesssim \frac{k^2f^2(k)}{n^2\log^2 k}.
\end{flalign}

Now consider the $\mathcal{E}_2$ term in \eqref{eq:100}. Based on how we define $I_2$, $\mathcal E_2$ can be written as follows:
\begin{small}
\begin{flalign}
  \mathcal E_2=&\sum_{i\in I_2\cap(I_1'\cup I_2')}\left( \frac{M_i}{m} \Big(\log \frac{N_i+1}{n} -\frac{1}{2(N_i+1)}\Big)-P_i\log Q_i\right)\nn\\
=&\sum_{i\in I_2\cap(I_1'\cup I_2')} \bigg(\Big(\frac{M_i}{m}-P_i\Big)\log Q_i \nn \\
 &\qquad \qquad\qquad \qquad+ \frac{M_i}{m}\log \frac{N_i+1}{nQ_i}-\frac{P_i}{2(N_i+1)}\bigg).
\end{flalign}
\end{small}
Taking expectations on both sides, we obtain
\begin{flalign}
&\mE\big[\mathcal E_2\big|I_2,I_1',I_2'\big]\nn \\
&=\sum_{i\in I_2\cap(I_1'\cup I_2')}\mE\left[P_i\log \frac{N_i+1}{nQ_i}-\frac{P_i}{2(N_i+1)}\bigg|I_1,I_1',I_2'\right].
\end{flalign}
Consider $\sum_{i\in I_2\cap(I_1'\cup I_2')}\mE\Big[P_i\log \frac{N_i+1}{nQ_i}\Big|I_1,I_1',I_2'\Big]$. Note that for any $x>0$,
\begin{flalign}
  \log x \le (x-1)-\frac{1}{2}(x-1)^2+\frac{1}{3}(x-1)^3.
\end{flalign}
Since $N_i\sim$Poi$(nQ_i)$,
\begin{small}
\begin{flalign}
  &\mE\left[P_i\log \frac{N_i+1}{nQ_i}\right] \nn \\
  &\le P_i\mE\left[\Big(\frac{N_i+1}{nQ_i}-1\Big)-\frac{1}{2}\Big(\frac{N_i+1}{nQ_i}-1\Big)^2
  +\frac{1}{3}\Big(\frac{N_i+1}{nQ_i}-1\Big)^3\right]\nn\\
  &=P_i\Big(\frac{1}{2nQ_i}+\frac{5}{6(nQ_i)^2}+\frac{1}{3(nQ_i)^3}\Big).
\end{flalign}
\end{small}
It can be shown that
\begin{flalign}
  \mE\left[\frac{P_i}{2(N_i+1)}\right]=\frac{P_i}{2nQ_i}(1-e^{-nQ_i}).
\end{flalign}
Hence, we obtain
\begin{flalign}\label{eq:152a}
&\mE\big[\mathcal E_2\big|I_2,I_1',I_2'\big] \nn \\
&\le \sum_{i\in I_2\cap(I_1'\cup I_2')} P_i\Big(\frac{1}{2nQ_i}+\frac{5}{6(nQ_i)^2}+\frac{1}{3(nQ_i)^3}\Big)\nn \\
&\qquad\qquad \qquad\qquad -\frac{P_i}{2nQ_i}(1-e^{-nQ_i})\nn\\
&\overset{(a)}{\lesssim} \sum_{i\in I_2\cap(I_1'\cup I_2')}\frac{P_i}{n^2Q_i^2}\nn\\
&\lesssim \frac{kf(k)}{n\log k}.
\end{flalign}
where $(a)$ is due to the fact that $xe^{-x}$ is bounded by a constant for $x\ge 0$.

We further derive a lower bound on $\mE\big[\mathcal E_2\big|I_2,I_1',I_2'\big]$. For any $x\ge \frac{1}{5}$, it can be shown that
\begin{flalign}\label{eq:log_lower}
  \log x \ge (x-1)-\frac{1}{2}(x-1)^2+\frac{1}{3}(x-1)^3-(x-1)^4.
\end{flalign}
Define the following event: $A_i=\{\frac{N_i}{nQ_i}>\frac{1}{5}\}$. We then rewrite $\mE\big[\mathcal E_2\big|I_2,I_1',I_2'\big]$ as follows:
\begin{flalign}
&\mE\big[\mathcal E_2\big|I_2,I_1',I_2'\big] \nn \\
&=\sum_{i\in I_2\cap(I_1'\cup I_2')} \mE\bigg[P_i\log\frac{N_i+1}{nQ_i}\mathds 1_{\{A_i\}}\nn \\
&\qquad \qquad \qquad+P_i\log\frac{N_i+1}{nQ_i}\mathds 1_{\{A_i^c\}}-\frac{P_i}{2(N_i+1)}\bigg|I_2,I_1',I_2'\bigg]\nn\\
&\ge \sum_{i\in I_2\cap(I_1'\cup I_2')} \mE\left[P_i\log\frac{N_i+1}{nQ_i}\mathds 1_{\{A_i\}}-\frac{P_i}{2(N_i+1)}\bigg|I_2,I_1',I_2'\right]\nn \\
&\quad- \sum_{i\in I_2\cap(I_1'\cup I_2')}\left|\mE\left[ P_i\log\frac{N_i+1}{nQ_i}\mathds 1_{\{A_i^c\}}\bigg|I_2,I_1',I_2'\right]\right|.
\end{flalign}
Using \eqref{eq:log_lower}, we obtain
\begin{flalign}
  &\mathbb{E}\left[ P_i\log \Big( \frac{N_i+1}{nQ_i}\Big)\mathds{1}_{\{A_i\}}\bigg|I_2,I_1',I_2'\right] \nn\\
  &\ge\mathbb{E}\bigg[P_i \Big(\big(\frac{N_i+1}{nQ_i}-1\big)-\frac{1}{2}\big(\frac{N_i+1}{nQ_i}-1\big)^2 \nn \\
  &\quad +\frac{1}{3}\big(\frac{N_i+1}{nQ_i}-1\big)^3 -\big(\frac{N_i+1}{nQ_i}-1\big)^4\Big)\mathds{1}_{\{A_i\}}\bigg|I_2,I_1',I_2'\bigg].
\end{flalign}
Note that
\begin{flalign}
  &\mE\left[ \Big(\frac{N_i+1}{nQ_i}-1\Big)\mathds{1}_{\{A_i\}}\bigg|I_2,I_1',I_2'\right] \nn \\
  &=\mE\left[ \Big(\frac{N_i+1}{nQ_i}-1\Big)\bigg|I_2,I_1',I_2'\right]\nn \\
  &\quad -\mE\left[ \Big(\frac{N_i+1}{nQ_i}-1\Big)\mathds{1}_{\{A_i^c\}}\bigg|I_2,I_1',I_2'\right]\nn\\
  &\overset{(a)}{\ge} \mE\left[ \Big(\frac{N_i+1}{nQ_i}-1\Big)\bigg|I_2,I_1',I_2'\right]\nn\\
  &=\frac{1}{nQ_i},
\end{flalign}
where $(a)$ follows because $(\frac{N_i+1}{nQ_i}-1)\mathds{1}_{\{A_i^c\}} \le 0$.
Similarly,
\begin{flalign}
& \mathbb{E}\left[ \Big(\frac{N_i+1}{nQ_i}-1\Big)^3\mathds{1}_{\{A_i\}}\bigg|I_2,I_1',I_2'\right] \nn \\
&\ge  \mathbb{E}\left[ \Big(\frac{N_i+1}{nQ_i}-1\Big)^3\right]=\frac{4}{(nQ_i)^2} +\frac{1}{(nQ_i)^3}.
\end{flalign}
For the term $\mE\left[\Big(\frac{N_i+1}{nQ_i}-1\Big)^2\bigg|I_2,I_1',I_2'\right]$, it can be shown that
\begin{flalign}
  &\mE\left[\Big(\frac{N_i+1}{nQ_i}-1\Big)^2\mathds{1}_{\{A_i\}}\bigg|I_2,I_1',I_2'\right] \nn \\
  &\le\mE\left[\Big(\frac{N_i+1}{nQ_i}-1\Big)^2\bigg|I_2,I_1',I_2'\right]\nn \\
  &=\frac{1}{nQ_i}+\frac{1}{(nQ_i)^2}.
\end{flalign}
Similarly, it can be shown that
\begin{flalign}
  &\mE\left[\Big(\frac{N_i+1}{nQ_i}-1\Big)^4\mathds{1}_{\{A_i\}}\bigg|I_2,I_1',I_2'\right]\nn \\
  &\le  \mE\left[\Big(\frac{N_i+1}{nQ_i}-1\Big)^4\bigg|I_2,I_1',I_2'\right]\nn \\
  &= \frac{1+3nQ_i}{(nQ_i)^3}+\frac{10}{(nQ_i)^3}+\frac{1}{(nQ_i)^4}.
\end{flalign}
Combining these results together, we obtain
\begin{flalign}\label{eq:160a}
&\mathbb{E}\left[ P_i\log \Big(\frac{N_i+1}{nQ_i}\Big)\mathds{1}_{\{A_i\}}\bigg|I_2,I_1',I_2'\right] \nn \\
&\ge  \frac{P_i}{2nQ_i}-\frac{13P_i}{6(nQ_i)^2}-\frac{32P_i}{3(nQ_i)^3}-\frac{P_i}{(nQ_i)^4}.
\end{flalign}


From the previous results, we know that
\begin{flalign}\label{eq:163a}
  \mE\left[-\frac{P_i}{2(N_i+1)}\bigg|I_2,I_1',I_2'\right]=-\frac{P_i}{2nQ_i}(1-e^{-nQ_i}).
\end{flalign}
Combining \eqref{eq:160a} and \eqref{eq:163a}, it can be shown that
\begin{flalign}\label{eq:163b}
&\sum_{i\in I_2\cap(I_1'\cup I_2')} \mE\left[P_i\log \Big(\frac{N_i+1}{nQ_i}\Big)\mathds{1}_{\{A_i\}}-\frac{P_i}{2(N_i+1)}\bigg|I_2,I_1',I_2'\right]\nn\\
&\ge  \sum_{i\in I_2\cap(I_1'\cup I_2')}  \bigg(\frac{P_i}{2nQ_i}-\frac{13P_i}{6(nQ_i)^2}-\frac{32P_i}{3(nQ_i)^3} \nn \\
& \qquad \qquad \qquad \qquad-\frac{P_i}{(nQ_i)^4} -\frac{P_i}{2nQ_i}(1-e^{-nQ_i})\bigg)\nn\\
&=  \sum_{i\in I_2\cap(I_1'\cup I_2')} \bigg(-\frac{13P_i}{6(nQ_i)^2}-\frac{32P_i}{3(nQ_i)^3} \nn \\
&\qquad \qquad \qquad \qquad -\frac{P_i}{(nQ_i)^4} +\frac{P_i}{2nQ_i}e^{-nQ_i}\bigg)\nn \\
&\gtrsim  -\frac{kf(k)}{n\log k},
\end{flalign}
where we use the facts that $\frac{P_i}{Q_i}\le f(k)$, $nQ_i>c_3\log k$, and $xe^{-x}$ is upper bounded by a constant for any value of $x>0$.

For the $\mathbb{E}\big[P_i\log\frac{N_i+1}{nQ_i}\mathds 1_{\{A_i^c\}}\big|I_2,I_1',I_2'\big]$ term, it can be shown that
\begin{flalign}\label{eq:162a}
  &\sum_{i\in I_2\cap(I_1'\cup I_2')}\bigg|\mathbb{E}\left[P_i\log\frac{N_i+1}{nQ_i}\mathds 1_{\{A_i^c\}}\bigg|I_2,I_1',I_2'\right] \bigg| \nn \\
  &\overset{(a)}{\le}\sum_{i\in I_2\cap(I_1'\cup I_2')}P_i\log (nQ_i) P(A_i^c)\nn\\
  &\overset{(b)}{\le}\sum_{i\in I_2\cap(I_1'\cup I_2')}\frac{P_i}{(nQ_i)^2}(nQ_i)^2\log (nQ_i) e^{-(1-\frac{\log(5e)}{5})nQ_i}\nn\\
  &\overset{(c)}{\lesssim}  \frac{kf(k)}{n\log k}.
\end{flalign}
where $(a)$ is due to the fact that $N_i+1\ge 1$, and the fact that $Q_i>\frac{c_3\log k}{n}$, hence $|\log\frac{N_i+1}{nQ_i}|\le \log (nQ_i)$ for large $k$; $(b)$ is due to the Chernoff bound, where$1-\frac{\log(5e)}{5}>0$; $(c)$ is due to the fact that $x^2\log x e^{-(1-\frac{\log(5e)}{5})x}$ is bounded by a constant for $x>1$, and the fact that $nQ_i>c_3\log k$. Thus, \eqref{eq:163b} and \eqref{eq:162a} yield
\begin{flalign}\label{eq:last}
  \mE\big[\mathcal E_2\big|I_2,I_1',I_2'\big]
  \gtrsim -\frac{kf(k)}{n\log k}.
\end{flalign}

Combining \eqref{eq:152a} and \eqref{eq:last}, we obtain,
\begin{flalign}
  \left|\mE\big[\mathcal E_2\big|I_2,I_1',I_2'\big]\right|
  \lesssim \frac{kf(k)}{n\log k}.
\end{flalign}

For the constant $c_0$, $c_1$, $c_2$ and $c_3,$, note that $\log m \le C\log k$ for some constant $C$, we can choose $c_1=50(C+1)$, $c_2=e^{-1}c_1$, $c_3 = e^{-1}c_2$, such that $c_1-c_2\log\frac{ec_1}{c_2}-1>C$, $c_3-c_2\log\frac{ec_3}{c_2}-1>C$ and $\frac{c_3(1-\log 2)}{2}+1-C>0$ hold simultaneously. Also, we can choose $c_0>0$ sufficiently small, satisfying condition $2c_0\log 8<\frac{1}{2}$ and $2(c_0\log 8 +\sqrt{c_0c_1}\log 2e)<\frac{1}{2}$. Thus, we show the existence of $c_0$, $c_1$ and $c_2$.

\bibliographystyle{IEEEbib}
\bibliography{minimax}

\begin{IEEEbiographynophoto}{Yuheng Bu}(S'16) received the M.S. degree in Electrical and Computer Engineering from University of Illinois at Urbana-Champaign, Champaign, IL, USA,
in 2016, and the B.S. degree (with honors) in Electrical Engineering from Tsinghua University, Beijing, China in 2014. He is currently working toward the Ph.D. degree in
the Coordinated Science Laboratory and the Department of Electrical and Computer Engineering, University of Illinois at Urbana-Champaign. His research
interests include detection and estimation theory, machine learning and information theory.
\end{IEEEbiographynophoto}

\begin{IEEEbiographynophoto}{Shaofeng Zou}(S'14-M'16) received the Ph.D. degree in Electrical and Computer Engineering from Syracuse University in 2016. He received the B.E.\ degree (with honors) from Shanghai Jiao Tong University, Shanghai, China, in 2011. Since July 2016, he has been a postdoctoral research associate at University of Illinois at Urbana-Champaign. Dr. Zou's research interests include information theory, machine learning, and statistical signal processing. He received the National Scholarship from the Ministry of Education of China in 2008. He also received the award of the outstanding graduate of Shanghai in 2011.
\end{IEEEbiographynophoto}

\begin{IEEEbiographynophoto}{Yingbin Liang} (S'01-M'05-SM'16) received the Ph.D. degree in Electrical Engineering from the University of Illinois at Urbana-Champaign in 2005. In 2005-2007, she was working as a postdoctoral research associate at Princeton University. In 2008-2009, she was an assistant professor at University of Hawaii. In 2010-2017, she was an assistant and then an associate professor at Syracuse University. Since August 2017, she has been an associate professor at the Department of Electrical and Computer Engineering at the Ohio State University. Dr. Liang's research interests include machine learning, statistical signal processing, optimization, information theory, and wireless communication and networks.

Dr. Liang was a Vodafone Fellow at the University of Illinois at Urbana-Champaign during 2003-2005, and received the Vodafone-U.S. Foundation Fellows Initiative Research Merit Award in 2005. She also received the M. E. Van Valkenburg Graduate Research Award from the ECE department, University of Illinois at Urbana-Champaign, in 2005. In 2009, she received the National Science Foundation CAREER Award, and the State of Hawaii Governor Innovation Award. In 2014, she received EURASIP Best Paper Award for the EURASIP Journal on Wireless Communications and Networking. She served as an Associate Editor for the Shannon Theory of the IEEE Transactions on Information Theory during 2013-2015.
\end{IEEEbiographynophoto}

\begin{IEEEbiographynophoto}{Venugopal V. Veeravalli} (M'92-SM'98-F'06) received the B.Tech. degree (Silver Medal Honors) from the Indian Institute of Technology, Bombay, India, in 1985, the M.S. degree from Carnegie Mellon University, Pittsburgh, PA, USA, in 1987, and the Ph.D. degree from the University of Illinois at Urbana-Champaign, Champaign, IL, USA, in 1992, all in electrical engineering.

He joined the University of Illinois at Urbana-Champaign in 2000, where he is currently the Henry Magnuski Professor in the Department of Electrical and Computer Engineering, and where he is also affiliated with the Department of Statistics, the Coordinated Science Laboratory, and the Information Trust Institute. From 2003 to 2005, he was a Program Director for communications
research at the U.S. National Science Foundation in Arlington, VA, USA. He has previously held academic positions at Harvard University, Rice University, and Cornell University, and has been on sabbatical at MIT, IISc Bangalore, and Qualcomm, Inc. His research interests include statistical signal processing, machine learning, detection and estimation theory, information theory, and stochastic control, with applications to sensor networks, cyberphysical systems, and wireless communications. A recent emphasis of his research has been on signal processing and machine learning for data science applications.

Prof. Veeravalli was a Distinguished Lecturer for the IEEE Signal Processing Society during 2010 and 2011. He has been on the Board of Governors of the IEEE Information Theory Society. He has been an Associate Editor for Detection and Estimation for the IEEE Transactions on Information Theory and for the IEEE Transactions on Wireless Communications. The awards
he received for research and teaching are the IEEE Browder J. Thompson Best Paper Award, the National Science Foundation CAREER Award, and the Presidential
Early Career Award for Scientists and Engineers, and the Wald Prize in Sequential Analysis.
\end{IEEEbiographynophoto}

\end{document}